\documentclass[11pt,onecolumn]{IEEEtran}
\usepackage{amsmath}
\usepackage{graphicx}
\usepackage{subfigure}
\newtheorem{theorem}{\bf Theorem}
\newtheorem{corollary}{\bf Corollary}
\newtheorem{lemma}{\bf Lemma}

\begin{document}

\title{Capacity of a Class of Symmetric SIMO Gaussian Interference Channels within $\mathcal{O}(1)$}

\author{\authorblockN{Tiangao Gou, Syed A. Jafar}\\
\authorblockA{Electrical Engineering and Computer Science\\
University of California Irvine, Irvine, California, 92697, USA\\
Email: \{tgou,syed\}@uci.edu}}

\maketitle

\begin{abstract}
The $N+1$ user, $1 \times N$ single input multiple output (SIMO)
Gaussian interference channel where each transmitter has a single
antenna and each receiver has $N$ antennas is studied. The symmetric
capacity within $\mathcal{O}(1)$ is characterized for the symmetric
case where all direct links have the same signal-to-noise ratio
(SNR) and all undesired links have the same interference-to-noise
ratio (INR). The gap to the exact capacity is a constant which is
{\em independent} of $\text{SNR}$ and $\text{INR}$. To get this
result, we first generalize the deterministic interference channel
introduced by El Gamal and Costa in \cite{El Gamal_Costa} to model
interference channels with multiple antennas. We derive the capacity
region of this deterministic interference channel. Based on the
insights provided by the deterministic channel, we characterize the
generalized degrees of freedom (GDOF) of Gaussian case, which
directly leads to the $\mathcal{O}(1)$ capacity approximation. On
the achievability side, an interesting conclusion is that the
generalized degrees of freedom (GDOF) regime where treating
interference as noise is found to be optimal in the 2 user
interference channel, does not appear in the $N+1$ user, $1 \times
N$ SIMO case. On the converse side, new multi-user outer bounds
emerge out of this work that do not follow directly from the 2 user
case. In addition to the GDOF region, the outer bounds identify a
strong interference regime where the capacity region is established.
\end{abstract}
\newpage

\section{Introduction}
The capacity of the interference channel has been an open problem
for over thirty years. The key to making progress on this problem is
to pursue capacity approximations. Taking this approach, seminal
work by Etkin, Tse and Wang \cite{onebit} has produced the capacity
characterization within one bit for the two user Gaussian
interference channel. By further tightening one outer bound in
\cite{onebit}, the sum capacity has been shown to be achievable by
treating interference as noise in a very weak interference regime
(also known as noisy interference regime)
\cite{MK_int}\cite{Shang}\cite{Sreekanth_Veeravalli}. The succuss of
this characterization follows from two important techniques,
deterministic approach and generalized degrees of freedom. By
focusing on the interaction between desired signals and
interference, i.e., de-emphasizing local noise, the deterministic
channels provide fundamental insights into their Gaussian
counterparts
\cite{Bresler_Tse:deterministic}\cite{Bresler:manytoone}. The GDOF
perspective, introduced in \cite{onebit}, presents a coarse, but
insightful picture of the optimal achievable schemes and outer
bounds for the interference channel. In particular, the GDOF picture
for the symmetric case - which we refer to as the ``W'' curve- has
come to represent the interference channel in the same way as the
pentagonal capacity region is associated with the multiple access
channel. The W curve delineates very weak, weak, moderate, strong
and very strong interference regimes, each with a distinct
character.

The next logical step is to extend these insights to interference
{\em networks} - i.e., interference channels with more than 2 users.
Extensions to more than 2 users have turned out to be non-trivial
due to the emergence of some fundamentally new issues that are
unique to interference networks. In particular, the idea of
interference alignment is introduced in the context of interference
networks by Cadambe and Jafar in \cite{Cadambe_Jafar_int} as the
principal determinant of the network degrees of freedom (capacity
pre-log). The extent to which interference can be aligned is very
difficult to determine precisely in general. For this reason, {\em
even the exact capacity pre-log factor is unknown} for almost all
interference networks, including, e.g. the 3 user interference
channel with constant channel
coefficients\cite{Cadambe_Jafar_Wang}\cite{Etkindof}. Since the
capacity pre-log dominates all other factors in a capacity
approximation, most attempts to gain useful insights for
interference networks get caught in the intricacies of interference
alignment and do not make it past the degree-of-freedom question.
Notable exceptions include the perfectly symmetric $K$ user
interference channel considered in \cite{Jafar_Vishwanath_GDOF}, and
many-to-one (and one-to-many) interference channels considered in
\cite{Bresler:manytoone}. Jafar and Vishwanath
\cite{Jafar_Vishwanath_GDOF} use interference alignment to show that
the (per-user) GDOF characterization of the $K$ user interference
channel in a perfectly symmetric setting is identical to the 2 user
W curve except for one point of discontinuity - which does not allow
the results to be translated into a capacity approximation within
$\mathcal{O} (1)$. Bresler et. al. \cite{Bresler:manytoone}
successfully navigate the issue of interference alignment to find a
capacity approximation within a constant number of bits, but for the
limited case where only one receiver sees interference. We note that
while interference alignment is an important element in the
many-to-one interference channel, it is not necessary to determine
the capacity pre-log. In fact, the degrees-of-freedom region for the
many-to-one interference channel is achieved quite simply through
time-division.

Our goal is to explore whether, and in what form, the $W$ curve
generalizes to interference networks with more than 2 users and,
more importantly, to be able to go beyond GDOF, to capacity
characterization within $\mathcal{O}(1)$ and to exact capacity in
certain regimes.

\subsection{Motivating Example}
\begin{figure}[t]
 \centering
\includegraphics[width=3.70in]{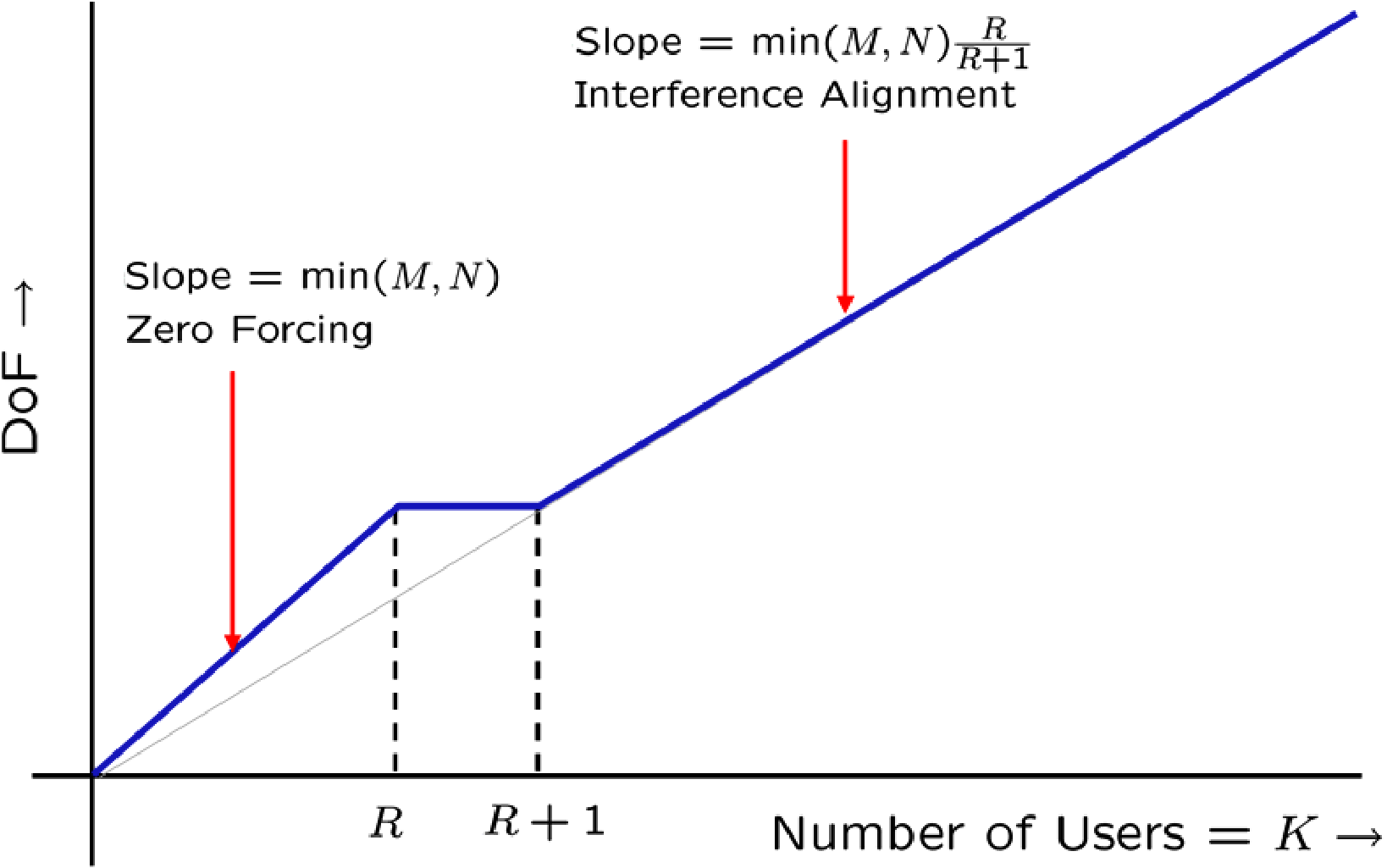}
\caption{Degrees of freedom for the K user MIMO interference channel
\cite{Gou_Jafar_dofmimo} }\label{dof}
\end{figure}
Consider the $K$ user  MIMO Gaussian interference channel with $M$
antennas at each transmitter and $N$ antennas at each receiver. The
degrees of freedom of this channel are characterized in
\cite{Gou_Jafar_dofmimo} when the ratio
$\frac{\max(M,N)}{\min(M,N)}=R$ is equal to an integer. The main
results of \cite{Gou_Jafar_dofmimo} are summarized in Fig.
\ref{dof}. As we can see, there are three distinct regimes. For the
first regime, i.e., $K \leq R$, there is no competition among the
users for degrees of freedom. Each user can access $\min(M,N)$
degree of freedom (the maximum possible) by zero forcing all the
interference, regardless of the strength of the interference. The
degrees of freedom, as well as the GDOF and the $\mathcal{O}(1)$
capacity characterization in this case are trivial (excluding some
degenerate cases). For $K > R$, \cite{Gou_Jafar_dofmimo} shows that
the degrees of freedom per user cannot be more than $\min(M,N)
\frac{R}{R+1}$. For $K>R+1$, the interference alignment problem
becomes challenging and the exact degrees of freedom are unknown
(\cite{Gou_Jafar_dofmimo} provides a tight inner bound only for
time-varying/frequency-selective channels) in general, making it
difficult to go beyond degrees of freedom. However, if $K=R+1$, the
MIMO interference channel has exactly $\min(M,N)\frac{R}{R+1}$
degrees of freedom per user (excluding degenerate cases) and
achievability follows by simple zero forcing and time sharing. While
the degrees of freedom problem is simple, the competition among the
users for the channel degrees of freedom means that the GDOF problem
is interesting and depends on the relative strength of the desired
and interfering signals. This is the case we study in this paper.
Note that for $R=1$ and $M=N=1$, our channel model reduces to the
classical two user interference channel and the results, e.g. the W
curve, of \cite{onebit} should be recovered in that case.

\subsection{Overview of Results}

We study the $N+1$ user SIMO Gaussian interference channel with $N$
antennas at each receiver. In order to obtain a compact
characterization, we focus on the symmetric case where all direct
links have the same signal-to-noise ratio (SNR) and undesired links
have the same interference-to-noise ratio (INR). However, no
symmetry is assumed for the directions of the channel vectors.
Inspired by the connection between the deterministic approach and
the GDOF of its Gaussian counterpart, which leads to the constant
bits characterization in the two user case, we also would like to
first investigate the problem through the corresponding
deterministic channel. However, the deterministic channel model
proposed in \cite{Avestimehr} cannot be applied to multiple antennas
cases\cite{Wang_Tse}. Instead, we generalize the El Gamal and Costa
model \cite{El Gamal_Costa} to interference channels with multiple
antennas. The key assumption of the El Gamal and Costa model for the
two user interference channel is the invertibility, i.e., at each
receiver, given the desired signal, the interference from the other
transmitter can be uniquely determined. This assumption makes this
model tractable and also emulates the two user Gaussian interference
channel. How to generalize this model to emulate the SIMO Gaussian
interference channels? Again the invertibility is essential which
also captures the essential feature of this class of  SIMO Gaussian
interference channels. Consider the 3 user SIMO interference channel
with 2 antennas at each receiver. Given the desired signal (in the
absence of noise), each receiver can determine the individual
interference from each of the interfering transmitters. This is
because the number of antennas at each receiver is equal to the
number of transmit antennas at all interferers combined. Therefore
after the desired signal has been removed, each interference signal
is individually isolated by a simple channel matrix inversion
operation at each receiver. Thus, we model the SIMO interference
channel in the deterministic framework of El Gamal and Costa, by the
assumption that, given the signal from the desired transmitter, each
receiver is able to recover each of the interfering signals from its
received signal. We characterize the capacity region of this
deterministic channel. The optimal achievable scheme for this $N+1$
user symmetric deterministic interference channel turns out to be a
natural extension of the Han Kobayashi scheme previously shown to be
optimal for the 2 user ($N=1$) case. The capacity region for the
deterministic channel also reveals interesting new forms of rate
bounds that are not trivial extensions of the 2 user case.

\begin{figure}[t]
 \centering
\includegraphics[width=5.20in, trim=0 90 0 0]{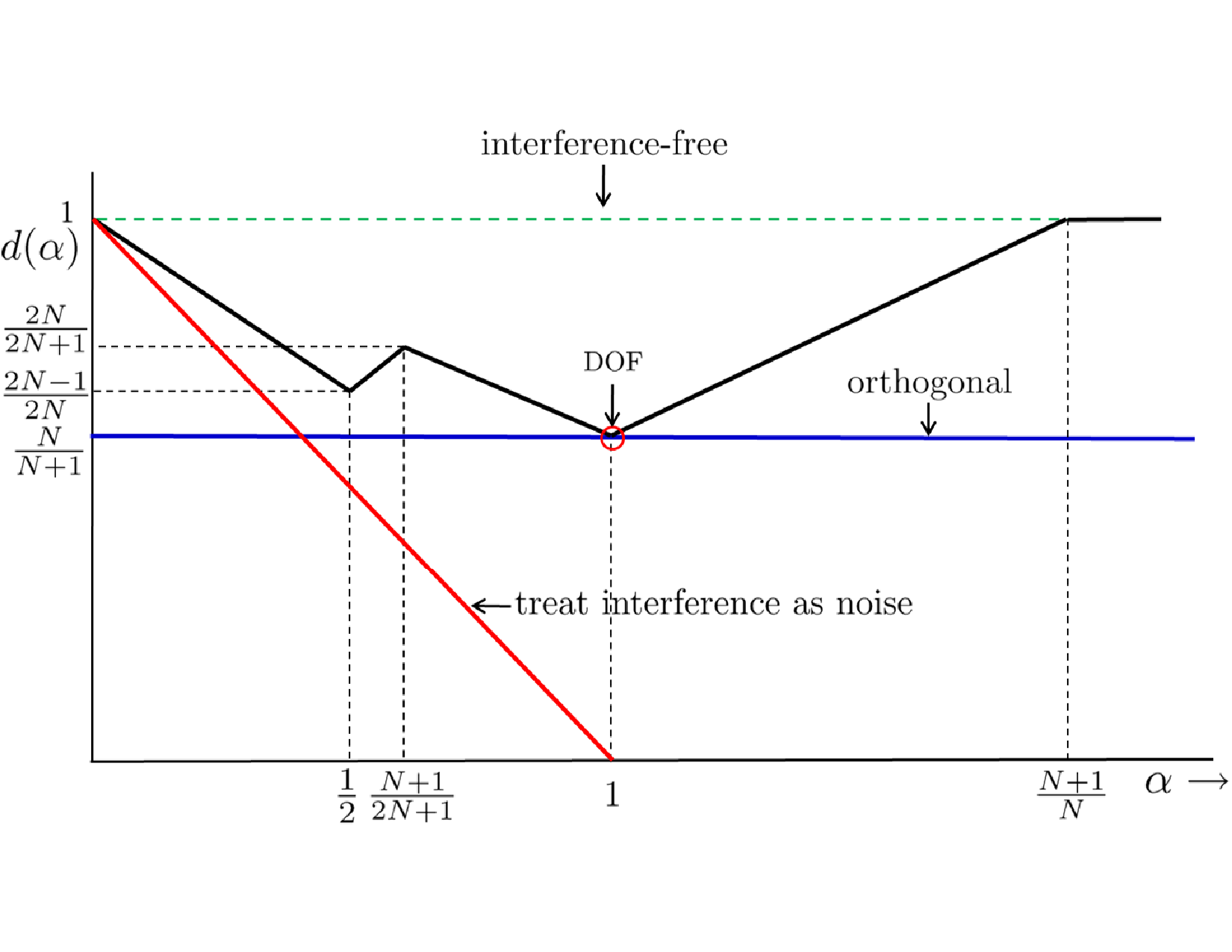}
\caption{GDOF of the $N+1$ user SIMO Gaussian interference channel
with $N$ antennas at each receiver } \label{gdof}
\end{figure}
Based on the insights provided by the deterministic channel, we
characterize the generalized degrees of freedom of the $N+1$ user
SIMO symmetric Gaussian interference channels with $N$ antennas at
each receiver. The GDOF curve is shown in Fig. \ref{gdof}.  Note
that for $N=1$, the 2 user GDOF of \cite{onebit} is obtained.
Similar to the 2 user Gaussian interference channel with single
antenna nodes, there are five distinct regimes. The key elements of
the achievability and converse for the GDOF characterization are
summarized as follows. In each case, we highlight one  major
similarity to the 2 user case and one major difference.
\begin{enumerate}
\item[1.] Achievability:
\begin{itemize}
\item Similarity: The idea of setting the private message power so
that it is received at the noise floor of the undesired receivers
carries over from the 2 user interference channel.

\item Difference: Unlike the 2 user case, the noisy interference
regime disappears from the GDOF picture. Intuitively, this may be
understood as follows. The limiting factor for this regime
($\alpha<1/2$)  in the 2 user case is the ``noise'' from the
\emph{desired} users' private message, which reduces the pre-log
factor of the common message rates to zero (The private messages of
the remaining users do not matter because they are received at the
level of the noise floor). However, in the SIMO case, the noise from
the desired users' private message can be nulled by losing one
dimension, while still leaving $N-1$ dimensions to decode the common
messages from the remaining $N$ users with a non-zero pre-log
factor.
\end{itemize}

\item[2.] Converse:
\begin{itemize}
\item Similarity: The outer bound that is tight for the second ``V''
of the W curve comes from a ``many-to-one'' interference channel
outer bound. This is the counterpart to the ``Z'' interference
channel sum capacity outer bound used for the corresponding regime
in the 2 user case.

\item Difference: The outer bound that is tight for the first ``V'' of
the W curve (the weak interference regime) is a new outer bound
which, unlike the two user case, is not in the form of a direct
sum-rate bound. For example, with 3 users the outer bound comes from
rate bounds that take the form $R_1+2R_2+R_3$. However, as in the 2
user case, the outer bound emerges from studying the El Gamal and
Costa deterministic channel model.
\end{itemize}
\end{enumerate}
The GDOF characterization leads directly to the capacity
approximation within $\mathcal{O}(1)$, i.e., the gap to the exact
capacity is a constant which is {\em independent} of $\text{SNR}$
and $\text{INR}$. Instead, the gap depends on other channel
parameters, e.g., the {\em angles} between channel vectors. To
further investigate how angles among channel vectors affect the gap,
we study the 3 user completely symmetric Gaussian interference
channel with 2 antennas at each receiver. By completely symmetric,
we mean that not only all SNRs, INRs are equal, respectively, but
also the relative orientations of the desired signal and
interference vectors are identical at each receiver. It turns out
that the gap only depends on the angle between two interfering
vectors (it does not depend on the angle between the desired channel
vector and the interfering channel vector). When the angle is large,
the gap is small, but if the angle is small, the gap is large. In
fact, this angle indicates the possibility of interference
alignment. The role of the angles between channel vectors is also
highlighted by Wang and Tse in \cite{Wang_Tse} for a three-to-one
Gaussian interference channel where only one receiver equipped with
2 antennas sees interference. They show that Han-Kobayashi-type
scheme with Gaussian codebook can achieve the capacity region within
a number of bits, which depends on the angle between two interfering
channel vectors. The gap becomes unbounded when the angle becomes
small. Wang and Tse \cite{Wang_Tse} provide a partial interference
alignment scheme to get a better performance where the transmit
signal is a superposition of Gaussian codewords and lattice
codewords. After decoding the Gaussian codewords, the signal is
projected onto one direction such that interference alignment can be
done using lattice code as in \cite{Bresler:manytoone}. However,
since only one receiver sees interference, the many-to-one setting
is significantly different from the fully connected interference
channel considered in this work.

We also derive an outer bound on the capacity region of the 3 user
SIMO Gaussian interference channel (not necessarily symmetric) with
2 antennas at each receiver. This outer bound directly leads to the
capacity region of the strong interference regime of this channel,
where each receiver can decode all messages.

\section{deterministic channel model}
\subsection{El Gamal and Costa Deterministic Channel Model and Connection to the Gaussian Channel}

The 2 user deterministic interference  channel studied in \cite{El
Gamal_Costa} is shown in Fig. \ref{twousermodel}. The interference
$V_1, V_2$ and channel outputs $Y_1, Y_2$ are deterministic
functions of inputs $X_1, X_2$, respectively:
\begin{eqnarray*}
V_1&=&g_1(X_1)\\
V_2&=&g_2(X_2)\\
Y_1&=&f_1(X_1, V_2)\\
Y_2&=&f_1(X_2, V_1)
\end{eqnarray*}
where $g_i, f_i~ \forall i=1,2$ are deterministic functions. In
addition, $f_i$ satisfies conditions
\begin{eqnarray}
H(Y_1|X_1)=H(V_2)\\
H(Y_2|X_2)=H(V_1)
\end{eqnarray}
for all product distributions on $X_1X_2$. In other words, given
$X_i$, function $f_i$ is invertible. The capacity region of this
class of interference channels is found in \cite{El Gamal_Costa}.
The achievable scheme is the Han-Kobayashi scheme which assigns the
common information to interfering signal which is visible to the
other link, i.e., $V_1$ and $V_2$. The connection of this
deterministic channel to the two user Gaussian interference channel
is discussed in \cite{onebit}. On the achievability side, it is
argued that the portion of the received interfering signal above the
noise level, i.e., the interfering signal which is visible to the
other link, should be the common information \cite{onebit}. On the
converse side, the outer bounds for the deterministic channel give
some hints of what genie information should be given to receivers
for the Gaussian case. We further illustrate this point through an
example. For the 2 user deterministic interference channel, the sum
capacity is shown \cite{El Gamal_Costa} to be bounded by
\begin{eqnarray}
R_1+R_2 \leq H(Y_{1}|V_{1})+H(Y_{2}|V_{2}).\label{deterministicetw}
\end{eqnarray}
For the 2 user Gaussian interference channel, a new outer bound
derived by Etkin, Tse and Wang \cite{onebit}, which we refer to as
the ETW bound, is
\begin{eqnarray}
R_1+R_2 \leq h(y_1|s_1)+h(y_2|s_2)-h(z_1)-h(z_2)\label{etw}
\end{eqnarray}
where $s_1$ and $s_2$ are interfering signals from Transmitter 1 and
2, respectively, i.e.,
\begin{eqnarray}
s_1 &=& h_{21}x_1+z_2\\
s_2 &=& h_{12}x_2+z_1
\end{eqnarray}
Comparing \eqref{deterministicetw} with \eqref{etw}, we can see that
$s_1$ and $s_2$ for the Gaussian channel are counterparts to $V_1$
and $V_2$ for the deterministic case. The outer bounds for the
deterministic channel and the Gaussian channel are very similar,
except that there is no noise terms in the deterministic case.
Therefore, the outer bound for the deterministic channel gives some
hints of what side information should be given to receivers in the
Gaussian case. For example, in order to have $h(y_{1}|s_{1})$ and
$h(y_{2}|s_{2})$, we may give $s_1$ to Receiver 1 and $s_2$ to
Receiver 2, which leads to the ETW bound.
\begin{figure}[t]
 \centering
\includegraphics[width=4.20in, trim=0 100 0 0]{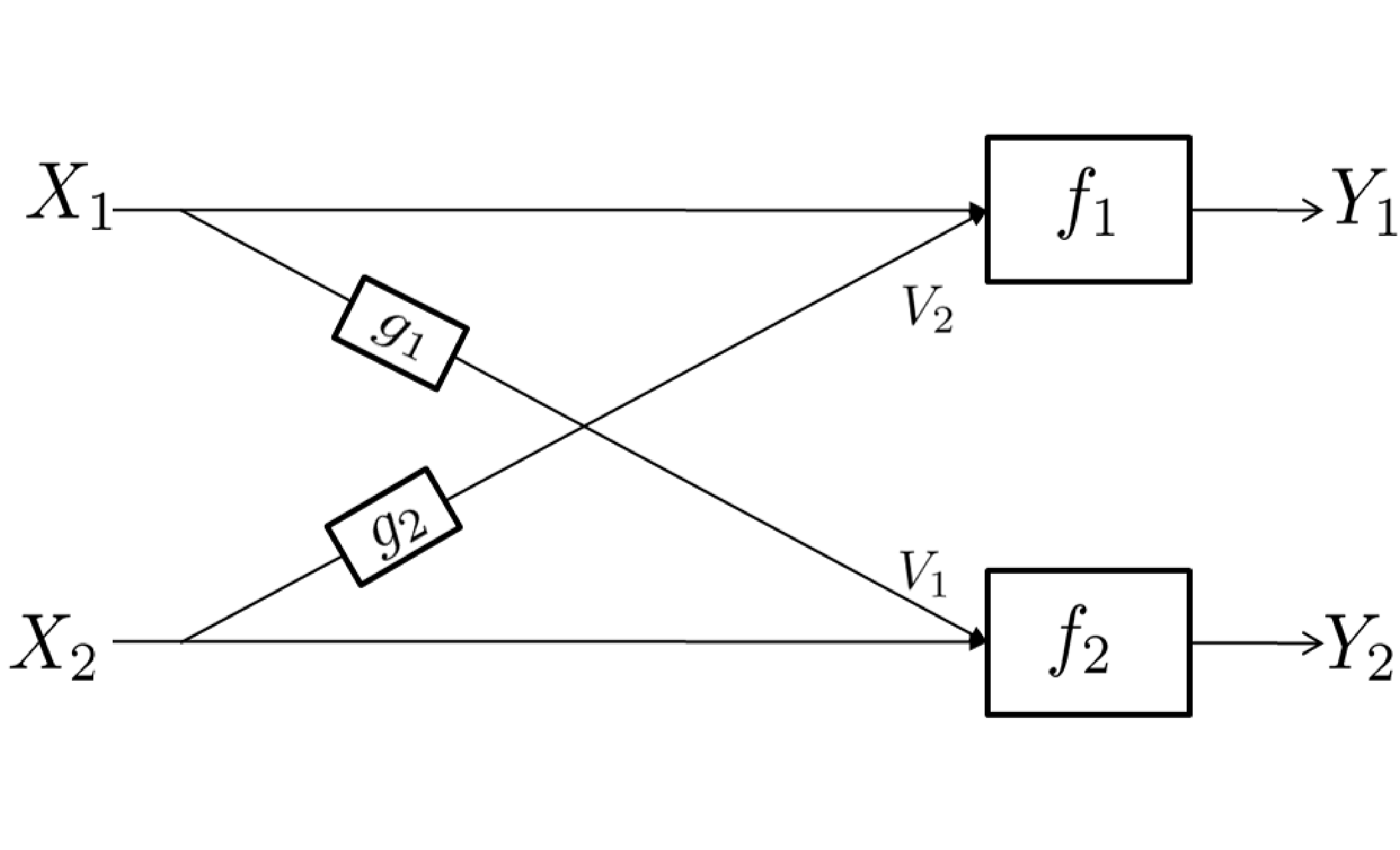}
\caption{Two user deterministic interference channel \cite{El
Gamal_Costa} }\label{twousermodel}
\end{figure}

\begin{figure}[t]
 \centering
\includegraphics[width=5.20in, trim=0 100 0 0]{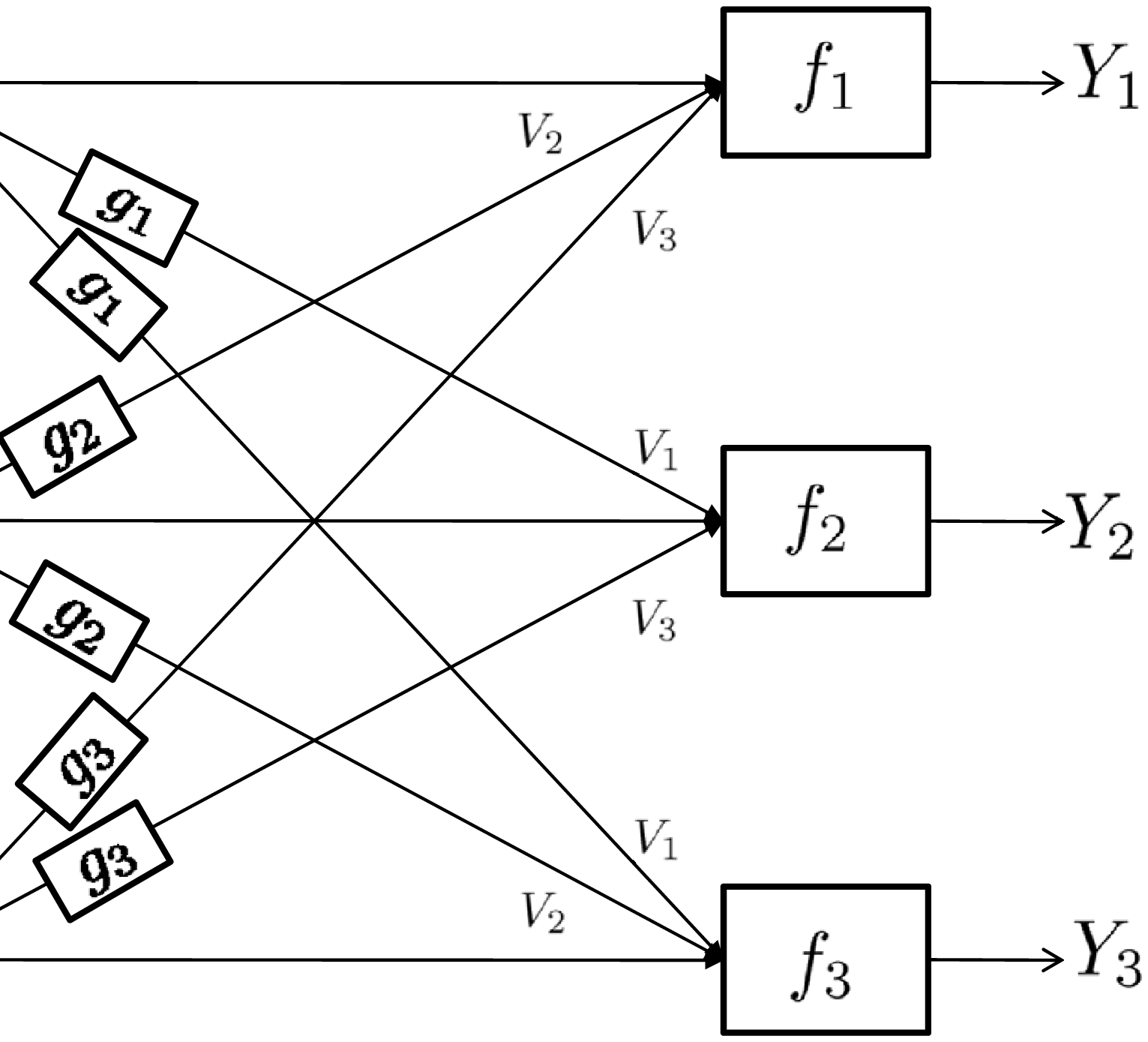}
\caption{3 user deterministic interference channel }
\label{deterministic}
\end{figure}

\subsection{Deterministic Channel for SIMO Interference Channel}
Inspired by the connection of the 2 user deterministic interference
channel to the Gaussian case, we would like to study the
corresponding deterministic channel of the 3 user SIMO Gaussian
interference channel with 2 antennas at each receiver. While we are
studying only the 3 user case, the insights from this allow us to
tackle the $N+1$ user, $1 \times N$ SIMO Gaussian interference
channel. Note that the key assumption of the El Gamal and Costa
model for the two user interference channel is the invertibility. We
model the 3 user $1\times 2$ SIMO interference channel in the
deterministic framework of El Gamal and Costa, by the assumption
that, given the signal from the desired transmitter, each receiver
is able to recover each of the interfering signals from its received
signal. This captures the corresponding essential feature of the
SIMO Gaussian interference channel. Each receiver has 2 antennas
which is equal to the total number of transmit antennas of all
interferers. After the desired signal has been removed, each
interference signal can be individually isolated by a simple channel
matrix inversion operation at each receiver. The deterministic
interference channel is shown in Fig. \ref{deterministic}. Note that
there is some avoidable loss of generality in assuming that the same
$V_1, V_2, V_3$ appear at more than one receiver. This assumption is
used primarily to obtain a compact description of the capacity
region, which is already cumbersome as we will see later. Also, it
is consistent with our ultimate objective of studying the symmetric
SIMO Gaussian interference channel.

The interference $V_1, V_2, V_3$ and channel outputs $Y_1, Y_2, Y_3$
are deterministic functions of inputs $X_1, X_2, X_3$, respectively:
\begin{eqnarray*}
V_1&=&g_1(X_1)\\
V_2&=&g_2(X_2)\\
V_3&=&g_3(X_3)\\
Y_1&=&f_1(X_1, V_2, V_3)\\
Y_2&=&f_1(X_2, V_1, V_3)\\
Y_3&=&f_1(X_3, V_1, V_2)
\end{eqnarray*}
where $g_i, f_i~ \forall i=1,2,3$ are deterministic functions. In
addition, $f_i$ satisfies conditions
\begin{eqnarray}
H(Y_1|X_1)=H(V_2 V_3)=H(V_2)+H(V_3)\label{dc1}\\
H(Y_2|X_2)=H(V_1 V_3)=H(V_1)+H(V_3)\label{dc2}\\
H(Y_3|X_3)=H(V_1 V_2)=H(V_1)+H(V_2)\label{dc3}
\end{eqnarray}
for all product distributions on $X_1X_2X_3$. In other words, given
$X_i$, function $f_i$ is invertible.

Transmitter $i$ has a message $W_i\in \{1,\cdots, 2^{nR_i}\}$ for
receiver $i$. A rate tuple $(R_1, R_2, R_3)$ is achievable if
$\forall \epsilon
>0$  there exists a block length $n$ and encoder $e_i$: $\{1,\cdots,2^{nR_i}\} ~\rightarrow
~\mathcal{X}_i^n$ and decoder $d_i$: $\mathcal{Y}_i^n ~\rightarrow
~\{1,\cdots,2^{nR_i}\}$ such that
\begin{eqnarray*}
\frac{1}{2^{n(R_1+R_2+R_3)}}\sum_{(w_1,w_2,w_3)}\text{Pr}\{(d_1(Y_1^n),d_2(Y_2^n),d_3(Y_3^n))
\neq (w_1,w_2,w_3) | W_1=w_1, W_2=w_2, W_3=w_3\} < \epsilon
\end{eqnarray*}
The capacity region of this channel is the closure of all achievable
rate tuples ($R_1, R_2, R_3$).

\section{capacity region of the deterministic channel}
Define a rate region $\mathcal{R}(1,2,3)$ specified by following
inequalities:
\begin{eqnarray}
R_1 &\leq& H(Y_1|V_2V_3)\\
R_1+R_2&\leq&H(Y_1|V_1V_2V_3)+H(Y_2|V_3)\\
R_1+R_2&\leq&H(Y_1|V_1V_3)+H(Y_2|V_2V_3)\\
2R_1+R_2&\leq&H(Y_1|V_3)+H(Y_1|V_1V_2V_3)+H(Y_2|V_2V_3)\\
R_1+R_2+R_3&\leq&H(Y_1|V_1)+H(Y_2|V_2V_3)+H(Y_3|V_1V_2V_3)\\
R_1+R_2+R_3&\leq&H(Y_1|V_1V_3)+H(Y_2|V_1V_2)+H(Y_3|V_2V_3)\\
R_1+R_2+R_3&\leq&H(Y_1|V_1V_2V_3)+H(Y_2|V_1)+H(Y_3|V_2V_3)\\
R_1+R_2+R_3&\leq&H(Y_1|V_1V_2V_3)+H(Y_2|V_1V_2V_3)+H(Y_3)\\
2R_1+R_2+R_3&\leq&H(Y_1)+H(Y_1|V_1V_2V_3)+H(Y_2|V_2V_3)+H(Y_3|V_1V_2V_3)\\
2R_1+R_2+R_3&\leq&H(Y_1|V_1)+H(Y_1|V_1V_2V_3)+H(Y_2|V_2V_3)+H(Y_3|V_2V_3)\\
2R_1+R_2+R_3&\leq&H(Y_1|V_1V_3)+H(Y_1|V_2)+H(Y_2|V_1V_2V_3)+H(Y_3|V_2V_3)\\
2R_1+R_2+R_3&\leq&H(Y_1|V_1V_2V_3)+H(Y_1|V_3)+H(Y_2|V_1V_2)+H(Y_3|V_2V_3)\\
2R_1+R_2+R_3&\leq&H(Y_1|V_1V_2V_3)+H(Y_1|V_3)+H(Y_2|V_2)+H(Y_3|V_1V_2V_3)\\
2R_1+R_2+R_3&\leq&H(Y_1|V_1V_2V_3)+H(Y_1|V_1V_2)+H(Y_2|V_2V_3)+H(Y_3|V_3)\\
2R_1+R_2+R_3&\leq&2H(Y_1|V_1V_2V_3)+H(Y_2)+H(Y_3|V_2V_3)\\
2R_1+R_2+R_3&\leq&2H(Y_1|V_1V_2V_3)+H(Y_2|V_2)+H(Y_3|V_3)\\
3R_1+R_2+R_3&\leq&2H(Y_1|V_1V_2V_3)+H(Y_1)+H(Y_2|V_2V_3)+H(Y_3|V_2V_3)\\
3R_1+R_2+R_3&\leq&2H(Y_1|V_1V_2V_3)+H(Y_1|V_2)+H(Y_2|V_2V_3)+H(Y_3|V_3)\\
2R_1+2R_2+R_3&\leq&H(Y_1)+H(Y_1|V_1V_3)+2H(Y_2|V_1V_2V_3)+H(Y_3|V_2V_3)\\
2R_1+2R_2+R_3&\leq&H(Y_1)+H(Y_1|V_1V_2V_3)+H(Y_2|V_1V_2V_3)+H(Y_2|V_2V_3)+H(Y_3|V_1V_3)\\
2R_1+2R_2+R_3&\leq&H(Y_1)+H(Y_1|V_1V_2V_3)+2H(Y_2|V_1V_2V_3)+H(Y_3|V_3)\\
2R_1+2R_2+R_3&\leq&H(Y_1|V_1)+H(Y_1|V_1V_2V_3)+H(Y_2|V_1V_2V_3)+H(Y_2|V_2V_3)+H(Y_3|V_3)\\
2R_1+2R_2+R_3&\leq&2H(Y_1|V_1V_3)+H(Y_2|V_1V_2V_3)+H(Y_2|V_2)+H(Y_3|V_2V_3)\\
3R_1+2R_2+R_3&\leq&2H(Y_1|V_1V_2V_3)+H(Y_1)+H(Y_2|V_1V_2V_3)+H(Y_2|V_2V_3)+H(Y_3|V_3)\\
3R_1+2R_2+R_3&\leq&2H(Y_1|V_1V_2V_3)+H(Y_1)+2H(Y_2|V_2V_3)+H(Y_3|V_1V_3)\\
3R_1+2R_2+R_3&\leq&2H(Y_1|V_1V_2V_3)+H(Y_1|V_1)+2H(Y_2|V_2V_3)+H(Y_3|V_3)\\
3R_1+2R_2+R_3&\leq&3H(Y_1|V_1V_2V_3)+H(Y_2|V_2V_3)+H(Y_2)+H(Y_3|V_3)\\
4R_1+2R_2+R_3&\leq&3H(Y_1|V_1V_2V_3)+H(Y_1)+2H(Y_2|V_2V_3)+H(Y_3|V_3)
\end{eqnarray}

\begin{theorem}\label{dcaparegion}
The capacity region of the deterministic channel is the closure of
the convex hull of the set of all rate tuples ($R_1, R_2, R_3$)
satisfying the conditions specified by
\begin{eqnarray*}
\mathcal{R}(1,2,3)\cup\mathcal{R}(1,3,2)\cup\mathcal{R}(2,1,3)\cup\mathcal{R}(2,3,1)\cup\mathcal{R}(3,1,2)\cup\mathcal{R}(3,2,1)
\end{eqnarray*}
over all product distributions $p_1(x_1)p_2(x_2)p_3(x_3)$.
\end{theorem}

\subsection{Achievability}

1) {\emph {Codebook generation}}: Transmitter $i$, $ \forall i=1, 2,
3$ generates $2^{nR_{ic}}$ independent codewords of length $n$,
$V_i^n(j_i), j_i \in \{1,2, \cdots, 2^{nR_{ic}}\}$, according to
$\prod_{k=1}^n p_{V_i}(v_{ik})$. For each codeword $V_i^n(j_i)$,
generate $2^{nR_{ip}}$ independent codewords $X^n_i(j_i,l_i)$
according to $\prod_{k=1}^n p_{X_i|V_i}(x_{ik}|v_{ik})$.

2) {\emph {Encoding}}: Transmitter $i$ sends codeword
$X_i^n(j_i,l_i)$ corresponding to the message indexed by
$(j_i,l_i)$.

3) {\emph {Decoding}}: Receiver $1$ looks for a unique $(\hat{j}_1,
\hat{l}_1)$ and a pair $(\hat{j}_2, \hat{j}_3)$ such that
\begin{eqnarray*}
\big( X_1^n(\hat{j}_1, \hat{l}_1), V_1^n(\hat{j}_1),
V_2^n(\hat{j}_2), V_3^n(\hat{j}_3), Y_1^n \big)\in A_n^{\epsilon} (
X_1, V_1, V_2, V_3, Y_1).
\end{eqnarray*}
Similar decoding is done at Receiver 2 and 3.

4) {\emph {Error Analysis}}: The detailed analysis is provided in
the Appendix. The probability of error at Receiver 1 can be made
arbitrarily small as $n \rightarrow \infty$ if the following
conditions are satisfied:
\begin{eqnarray}
R_{1p}+R_{1c}+R_{2c}+R_{3c} &\leq& H(Y_1)\label{condition1}\\
R_{1c}+R_{1p}+R_{2c} &\leq& H(Y_1|V_3)\\
R_{1c}+R_{1p}+R_{3c} &\leq& H(Y_1|V_2)\\
R_{1p}+R_{2c}+R_{3c} &\leq&H(Y_1|V_1)\\
R_{1p}+R_{2c} &\leq& H(Y_1|V_1V_3)\\
R_{1p}+R_{3c} &\leq& H(Y_1|V_1V_2)\\
R_{1c}+R_{1p}&\leq& H(Y_1|V_2V_3)\\
R_{1p}&\leq& H(Y_1|V_1V_2V_3)\label{condition8}\\
R_{ip},R_{ic} &\ge& 0 ~\forall i=1,2,3
\end{eqnarray}
By swapping indices 1 and 2 everywhere in the
$\eqref{condition1}-\eqref{condition8}$, we have the conditions for
Receiver 2. Similarly, the conditions for Receiver 3 are obtained by
swapping indices 1 and 3 everywhere in
$\eqref{condition1}-\eqref{condition8}$. All these conditions
specify the achievable rate region for rate vector $(R_{1p}, R_{2p},
R_{3p}, R_{1c}, R_{2c}, R_{3c})$. Using Fourier-Motzkin elimination
with $R_i=R_{ic}+R_{ip}$, we have the achievable region as in
Theorem \ref{dcaparegion}.

\subsection{Converse}
The proof for the converse uses the following inequality also used
in \cite{El Gamal_Costa}:
\begin{eqnarray}
H(X)-H(Y)\leq H(X|Y)
\end{eqnarray}
This is because
\begin{eqnarray}
H(X|Y)+H(Y) = H(X,Y) \geq H(X)
\end{eqnarray}
Here we only provide the proof for $4R_1+2R_2+R_3 \leq
3H(Y_1|V_1V_2V_3)+H(Y_1)+2H(Y_2|V_2V_3)+H(Y_3|V_3)$. The proof for
all other inequalities can be constructed in the same manner. From
Fano's inequality, we have
\begin{eqnarray*}
&&n(4R_1+2R_2+R_3-7\epsilon_n)\\
&\leq& 4I(X_1^n;Y_1^n)+2I(X_2^n;Y_2^n)+I(X_3^n;Y_3^n)\\
&\leq&
4H(Y_1^n)-4H(Y_1^n|X_1^n)+2H(Y_2^n)-2H(Y_2^n|X^n_2)+H(Y_3^n)-H(Y_3^n|X^n_3)\\
&=& 4H(Y_1^n)-4H(V_2^nV_3^n)+2H(Y_2^n)-2H(V_1^nV_3^n)+H(Y_3^n)-H(V_1^nV_2^n)\\
&=& 4H(Y_1^n)+2H(Y_2^n)+H(Y_3^n)-3H(V_1^n)-5H(V_2^n)-6H(V_3^n)\\
&=& 3H(Y_1^n)-3H(V_1^nV_2^nV_3^n)+H(Y_1^n)+2H(Y_2^n)-2H(V_2^nV_3^n)+H(Y_3^n)-H(V_3^n)\\
&\leq&
3H(Y_1^n|V_1^nV_2^nV_3^n)+H(Y_1^n)+2H(Y_2^n|V_2^nV_3^n)+H(Y_3^n|V_3^n)\\
&\leq&
\sum_{i=1}^n(3H(Y_{1i}|V_{1i}V_{2i}V_{3i})+H(Y_{1i})+2H(Y_{2i}|V_{2i}V_{3i})+H(Y_{3i}|V_{3i}))
\end{eqnarray*}
In the derivation above we use assumptions \eqref{dc1}, \eqref{dc2},
\eqref{dc3} and independence among $V_1$, $V_2$ and $V_3$.

Before we move on to the Gaussian interference channel, we summarize
the insights gained from this deterministic channel. On the
achievability side, since the Han-Kobayashi scheme which assigns the
common information to interfering signal which is visible to each
receiver achieves the capacity, we expect that the simple
Han-Kobayashi scheme with the private message's power set to be
received at the noise level used in \cite{onebit} should be good for
the SIMO Gaussian interference channel. On the converse side, the
outer bounds for the deterministic channel can help us to derive
outer bounds for the Gaussian case. Here we highlight two outer
bounds which will be used later.
\begin{eqnarray}
R_1+R_2+R_3\leq H(Y_1|V_1V_2V_3)+H(Y_2|V_1V_2V_3)+H(Y_3)\\
R_1+2R_2+R_3 \leq H(Y_1|V_1)+ 2H(Y_2|V_1V_2V_3)+H(Y_3|V_3)
\end{eqnarray}
As we can see, the first one is a many-to-one interference channel
outer bound (the counterpart to the Z channel bound in the 2 user
case). The second outer bound highlighted above is the basis for the
new outer bound that we derive for the Gaussian case, which is tight
in the weak interference regime. Note that the corresponding outer
bound for the 2 user case is a direct sum rate bound $R_1+R_2 \leq
H(Y_1|V_1) + H(Y_2|V_2)$ which leads to the ETW bound.

\section{Preliminaries for Symmetric SIMO Gaussian interference channels}
\subsection{Channel Model}\label{section:model}
Consider the symmetric $K$ user SIMO Gaussian interference channel,
where all direct links have the same signal-to-noise ratio (SNR),
and all cross links have the same interference-to-noise ratio (INR).
Each transmitter has a single antenna and each receiver has $N$
antennas. The channel's input-output relationship is described as
\begin{equation}\label{channelmodel}
\mathbf{Y}_j=\sqrt{\text{SNR}}\mathbf{H}_{jj}x_j+
\sqrt{\text{INR}}\sum_{i=1, i \neq
j}^{K}\mathbf{H}_{ji}x_i+\mathbf{Z}_j ~~\forall j=1,\ldots,K
\end{equation}
where $\mathbf{Y}_{j}$ is the $N \times 1$ received signal vector at
receiver $j$, $\mathbf{H}_{ji}$ is the $N \times 1$  channel vector
from transmitter $i$ to receiver $j$, $x_i$ is the input signal
which satisfies the average power constraint
$\mathbf{E}[|x_i|^2]\leq 1 $ and $\mathbf{Z}_j$ is the additive
circularly symmetric complex Gaussian noise vector with zero mean
and identity covariance matrix, i.e., $\mathbf{Z}_j\sim
\mathcal{CN}(\mathbf{0},\mathbf{I})$. We assume the norm of channel
vectors is equal to unity, i.e., $\|\mathbf{H}_{ji}\|=1, \forall
i,j=1,2,\cdots,K$. In order to avoid degenerate channel conditions,
we assume that the channel coefficients are drawn from a continuous
distribution, so that the channel vectors are in general position.
For example, if all channel vectors are collinear at each receiver,
it is essentially a single input and single output (SISO) channel.

The probability of error, achievable rates $R_1,\cdots,R_K$ and
capacity region $\mathcal{C}$ are defined in the standard Shannon
sense. Our focus is the {\em symmetric} capacity, i.e.,
\begin{eqnarray}
C_{\text{sym}}= \max_{(R,\cdots,R)\in \mathcal{C}}R
\end{eqnarray}

\subsection{Generalized Degrees of Freedom}

As in \cite{onebit}, we define
\begin{eqnarray*}
\alpha=\frac{\log{\text{INR}}}{\log{\text{SNR}}} \Rightarrow
\text{INR} = \text{SNR}^{\alpha}
\end{eqnarray*}
For simplicity, we denote $\text{SNR}$ by $\rho$ and $\text{INR}$ by
$\rho^{\alpha}$. The generalized degrees of freedom \emph{per user}
are defined as
\begin{equation*}
d_{\text{sym}}(\alpha)=\lim_{\rho \rightarrow
\infty}\frac{C_{\text{sym}}(\rho,\alpha)}{\log{\rho}}=\frac{1}{K}\lim_{\rho
\rightarrow \infty}\frac{C_{\Sigma}(\rho,\alpha)}{\log{\rho}}
\end{equation*}
where $C_{\text{sym}}$ is the \emph{symmetric} capacity and
$C_{\Sigma}$ is the \emph{sum} capacity.

\subsection{$\mathcal{O}(1)$ Approximation}
The $\mathcal{O}(1)$ approximation means that the approximation
error is bounded as the SNR, INR go to infinity. We will use two
$\mathcal{O}(1)$ approximations introduced in
\cite{parker:Gdofmimo}. We restate the two approximations from
\cite{parker:Gdofmimo} to accommodate to the notations in this
paper. The proof can be found in \cite{parker:Gdofmimo} and is
omitted here. The first one is the multiple access approximation:
\begin{lemma}\label{lemma:O(1)1}
Suppose $\mathbf{H}_1$ is an $N \times r_1$ matrix where $r_1$ is
the rank of $\mathbf{H}_1$. $\mathbf{H}_2$ is an $N \times r_2$
matrix where  $r_2$ is the rank of $\mathbf{H}_2$. For $r_1+r_2\ge
N$ and $\alpha \ge \beta$, almost surely
\begin{eqnarray}
\log|\mathbf{I}+\rho^{\alpha}\mathbf{H}_1\mathbf{H}_1^{\dagger}+\rho^{\beta}\mathbf{H}_2\mathbf{H}_2^{\dagger}|=r_1\alpha\log\rho+(N-r_1)\beta\log\rho+\mathcal{O}
(1)
\end{eqnarray}
\end{lemma}
The second one is the interference limited approximation:
\begin{lemma}\label{lemma:O(1)2}
For matrices $\mathbf{H}_1$ and $\mathbf{H}_2$ that satisfy the same
conditions in Lemma \ref{lemma:O(1)1}, almost surely
\begin{eqnarray}
\log|\mathbf{I}+(\mathbf{I}+\rho^{\beta}\mathbf{H}_2\mathbf{H}_2^{\dagger})^{-1}\rho^{\alpha}\mathbf{H}_1\mathbf{H}_1^{\dagger}|=r_1\alpha\log\rho+(N-r_1-r_2)\beta\log\rho+\mathcal{O}
(1)
\end{eqnarray}
\end{lemma}
Note that Lemma \ref{lemma:O(1)2} follows directly from Lemma
\ref{lemma:O(1)1}, because
\begin{eqnarray*}
&&\log|\mathbf{I}+(\mathbf{I}+\rho^{\beta}\mathbf{H}_2\mathbf{H}_2^{\dagger})^{-1}\rho^{\alpha}\mathbf{H}_1\mathbf{H}_1^{\dagger}|\\
&=&\log|\mathbf{I}+\rho^{\alpha}\mathbf{H}_1\mathbf{H}_1^{\dagger}+\rho^{\beta}\mathbf{H}_2\mathbf{H}_2^{\dagger}|-\log|\mathbf{I}+\rho^{\beta}\mathbf{H}_2\mathbf{H}_2^{\dagger}|\\
&=&
r_1\alpha\log\rho+(N-r_1)\beta\log\rho-r_2\beta\log\rho+\mathcal{O}
(1)\\
&=& r_1\alpha\log\rho+(N-r_1-r_2)\beta\log\rho+\mathcal{O} (1)
\end{eqnarray*}

\section{Generalized Degrees of Freedom of symmetric SIMO Gaussian Interference
channels}\label{section:gdof}
In this section, we focus on the SIMO
Gaussian interference channel when number of users, $K = N+1$ and
each receiver has $N$ antennas. The result is presented in the
following theorem.

\begin{theorem}\label{thm:gdof}
The generalized degrees of freedom of the $N+1$ user symmetric
 SIMO Gaussian interference channel with $N$ antennas at each
receiver are
\begin{eqnarray*}
d_{\text{sym}}(\alpha)=\left\{\begin{array}{ccc} 1-\frac{\alpha}{N} &~& 0 < \alpha \leq \frac{1}{2}\\ \frac{N-1}{N}+\frac{\alpha}{N}&~& \frac{1}{2} \leq \alpha \leq \frac{N+1}{2N+1}\\
1-\frac{\alpha}{N+1} &~& \frac{N+1}{2N+1} \leq \alpha \leq 1 \\
\frac{N}{N+1}\alpha&~&  1 \leq \alpha \leq \frac{N+1}{N}\\
1 &~& \alpha \geq \frac{N+1}{N}
\end{array}\right.
\end{eqnarray*}
\end{theorem}

Note that for $N=1$, the 2 user GDOF of \cite{onebit} is obtained.
To prove Theorem \ref{thm:gdof}, we derive inner bounds and outer
bounds on the GDOF and show that they match for each range of
$\alpha$. For comparison, we also plot the GDOF achieved by two
suboptimal schemes in Fig. \ref{gdof}: orthogonal transmission and
treating interference as noise. We can also make an interesting
observation for the very weak interference regime, i.e., $\alpha
\leq \frac{1}{2}$. Unlike the $K$ user symmetric Gaussian
interference channel with single antenna nodes where treating
interference as noise is optimal in terms of GDOF for the very weak
interference regime \cite{Jafar_Vishwanath_GDOF}, treating
interference as noise is strictly suboptimal for the GDOF of the
$N+1$ user SIMO interference channel with $N$ antennas at each
receiver.

\subsection{Inner Bounds on the Generalized Degrees of Freedom}
We establish the achievable GDOF for each regime in this section.
\subsubsection{$\alpha \geq 1$}
In this regime, the interference is stronger than the desired
signal. The achievable scheme is to let every receiver decode all
messages. Then, the achievable rate region is the intersection of
$N+1$ MAC capacity regions, one at each receiver. The region is
specified by
\begin{equation}
\sum_{k\in \mathcal{S}}R_{jk} \leq \log|\mathbf{I}+\sum_{k \in
\mathcal{S}}P_{k}\mathbf{H}_{jk}\mathbf{H}_{jk}^{\dagger}|~~~\forall
\mathcal{S} \subseteq \{1,2,\cdots,N+1\}
\end{equation}
where $P_{k}=\rho^{\alpha}$ if $k\neq j$ and $P_{j}=\rho$.  Based on
the rate region, we can calculate the generalized degrees of freedom
region. For every $\mathcal{S}\neq \{1, \cdots, N+1\}$,
\begin{eqnarray}
\sum_{k\in \mathcal{S}}d_{jk} (\alpha) &\leq& \left\{
\begin{array}{cc}|\mathcal{S}|\alpha & j \notin \mathcal{S}\\
(|\mathcal{S}|-1)\alpha+1 & j \in
\mathcal{S}\end{array}\right.\\
\Rightarrow d_{\text{sym}}(\alpha)&\leq& 1\label{strongconstraint1}
\end{eqnarray}
Note that $d_{\text{sym}}(\alpha)$ is the GDOF \emph{per user}. For
$\mathcal{S}=\{1, \cdots, N+1\}$,
\begin{eqnarray}
\sum_{k\in \mathcal{S}}d_{jk} (\alpha) \leq N\alpha \Rightarrow
d_{\text{sym}}(\alpha)\leq \frac{N}{N+1} \alpha
\label{strongconstraint2}
\end{eqnarray}
From \eqref{strongconstraint1}, \eqref{strongconstraint2}, we can
see that the achievable symmetric degrees of freedom is $\min\{1,
\frac{N}{N+1}\alpha\}$. Therefore,
\begin{eqnarray*}
d_{\text{sym}}(\alpha)=\left\{ \begin{array}{cc} 1 & \alpha \ge \frac{N+1}{N}\\
\frac{N}{N+1}\alpha & 1\leq \alpha \leq
\frac{N+1}{N}\end{array}\right.
\end{eqnarray*}

\subsubsection{$0 \leq \alpha \leq 1$}
This is the weak interference regime. The transmission scheme is a
natural generalization of the simple Han-Kobayashi type scheme used
in \cite{onebit}. Transmitter $j, \forall j=1,\ldots,N+1$ splits its
message $W_j$ into two sub-messages: a common message $W_{j,c}$ and
a private message $W_{j,p}$. The common message will be decoded by
all receivers while the private message is only decoded by the
desired receiver. The common message $W_{j,c}$ is encoded using a
Gaussian codebook with rate $R_{j,c}$ and power $P_{j,c}$. The
private message $W_{j,p}$ is encoded using a Gaussian codebook with
rate $R_{j,p}$ and power $P_{j,p}$. We set $R_{j,c}=R_c$ and
$R_{j,p}=R_p$. In addition, $P_{j,c}=P_c$ and $P_{j,p}=P_p$ such
that $P_p+P_c = 1$. Moreover, the private power $P_p$ is set such
that the private message is received at the noise floor at the
unintended receiver, i.e. $\rho^{\alpha}P_p=1$. If
$\rho^{\alpha}<1$, then set $P_p=1$. For this case, there is no
common message and each receiver decodes its message by treating
interference as noise. Finally, the transmitted signal $x_j$ is the
superposition of the common and private signals. The decoding order
is fixed as decoding the common messages first while decoding the
private message last. The achievable scheme for $N=2$ is illustrated
in Fig. \ref{fig:achievability}.

\begin{figure}[t]
 \centering
\includegraphics[width=3.2in]{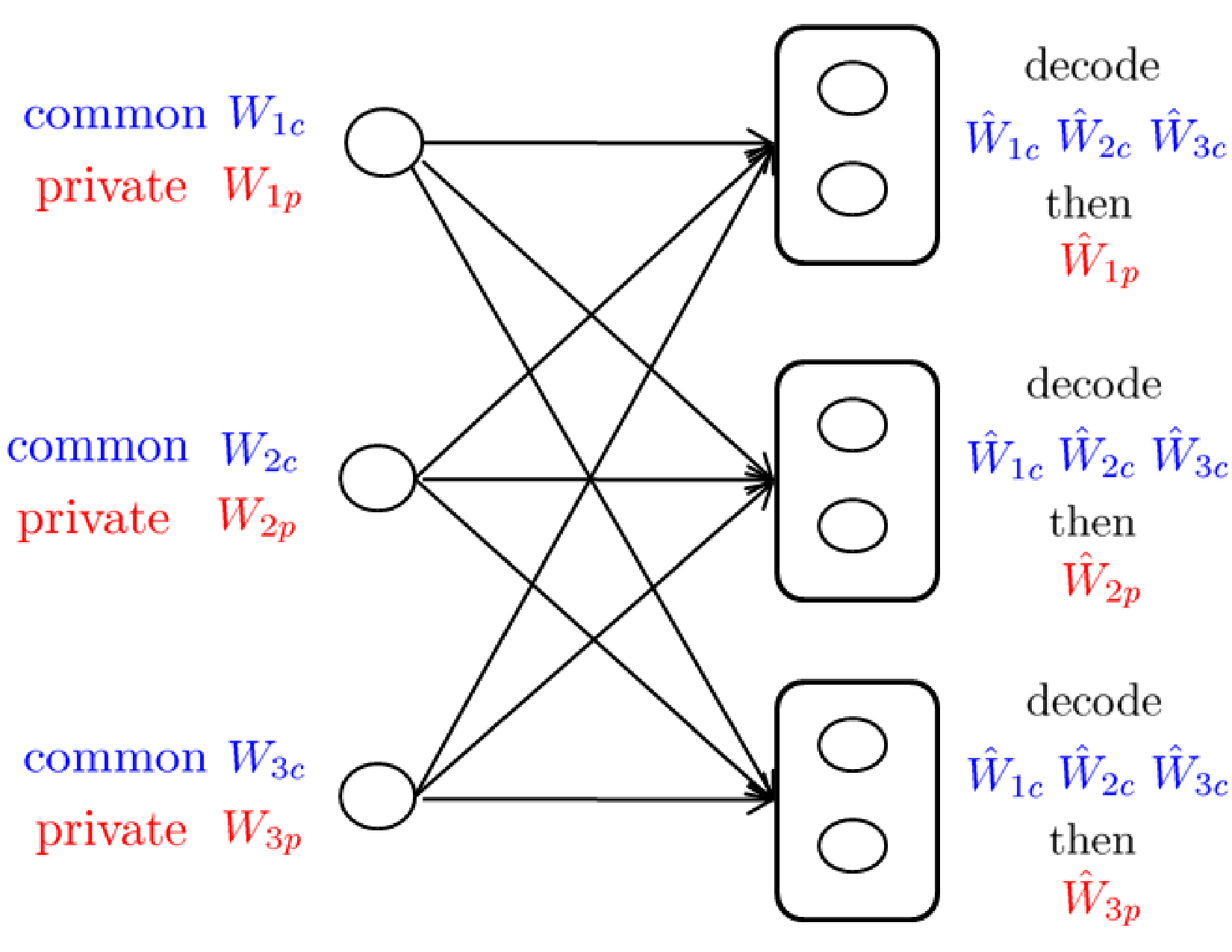}
\caption{Achievable scheme for 3 user $1 \times 2$ SIMO interference
channel in the weak interference regime ($0 \leq \alpha \leq 1$)}
\label{fig:achievability}
\end{figure}
To calculate the achievable rate using this scheme, it is useful to
determine the received SNR (INR) of the common messages and private
messages at the desired (unintended) receivers. Let $\text{SNR}_c$,
$\text{SNR}_p$, $\text{INR}_c$ and $\text{INR}_p$ denote the
received SNR for common messages and private messages at the desired
receiver, and the received INR for common messages and private
messages at the unintended receivers, respectively. It can be easily
seen that
\begin{eqnarray*}
\text{SNR}_c=\rho-\rho^{1-\alpha},~~\text{SNR}_p=\rho^{1-\alpha},~~\text{INR}_c=\rho^{\alpha}-1,~~\text{INR}_p=1
\end{eqnarray*}

We first calculate the rate for the private messages. Since the
private message is decoded after the common messages are decoded, it
is decoded by treating the private messages from unintended
transmitters as noise.
\begin{eqnarray}
R_p &=&
\min_j\{\log|\mathbf{I}+(\mathbf{I}+\sum_{j \neq i}\mathbf{H}_{ji}\mathbf{H}^{\dag}_{ji})^{-1}\rho^{1-\alpha}\mathbf{H}_{jj}\mathbf{H}_{jj}^{\dag}|\}\label{achprivate}\\
&=& (1-\alpha)\log\rho +\mathcal{O}(1)
\end{eqnarray}
Therefore,
\begin{equation}
d_p(\alpha)=1-\alpha
\end{equation}
The achievable rate region for the common messages is the
intersection of $N+1$ MAC capacity regions, one at each receiver.
Due to symmetry, consider the MAC at Receiver 1. Since the common
messages are decoded first by treating private messages as noise,
the achievable rate region is described by the $N+1$ user MAC
constraints, i.e., $\forall \mathcal{S} \subseteq
\{1,2,\cdots,N+1\}$
\begin{equation}
\sum_{k\in \mathcal{S}}R_{kc} \leq
\log|\mathbf{I}+(\mathbf{I}+\sum_{i \neq
1}\mathbf{H}_{1i}\mathbf{H}^{\dag}_{1i}+\rho^{1-\alpha}\mathbf{H}_{11}\mathbf{H}_{11}^{\dag})^{-1}\sum_{k
\in \mathcal{S}}P_{ck}\mathbf{H}_{1k}\mathbf{H}_{1k}^{\dagger}|
\end{equation}
where $P_{ck}=\rho^{\alpha}-1$ if $k\neq 1$ and
$P_{c1}=\rho-\rho^{1-\alpha}$.  Based on the rate region, we can
calculate the generalized degrees of freedom region. First, for
every $ \mathcal{S}\neq \{2,3,\cdots, N+1\}$ or $\{1,2,\cdots,
N+1\}$,
\begin{eqnarray}
&&\log|\mathbf{I}+(\mathbf{I}+\sum_{i \neq
1}\mathbf{H}_{1i}\mathbf{H}^{\dag}_{1i}+\rho^{1-\alpha}\mathbf{H}_{11}\mathbf{H}_{11}^{\dag})^{-1}\sum_{k
\in \mathcal{S}}P_{ck}\mathbf{H}_{1k}\mathbf{H}_{1k}^{\dagger}|\notag\\
&=&\log|\mathbf{I}+\sum_{i \neq
1}\mathbf{H}_{1i}\mathbf{H}^{\dag}_{1i}+\rho^{1-\alpha}\mathbf{H}_{11}\mathbf{H}_{11}^{\dag}+\sum_{k
\in
\mathcal{S}}P_{ck}\mathbf{H}_{1k}\mathbf{H}_{1k}^{\dagger}|-\log|\mathbf{I}+\sum_{i
\neq
1}\mathbf{H}_{1i}\mathbf{H}^{\dag}_{1i}+\rho^{1-\alpha}\mathbf{H}_{11}\mathbf{H}_{11}^{\dag}|\label{macrate}
\end{eqnarray}
If $1\in \mathcal{S}$, then \eqref{macrate} is
\begin{eqnarray}
&&\log|\mathbf{I}+\sum_{i \neq 1, i \notin
\mathcal{S}}\mathbf{H}_{1i}\mathbf{H}^{\dag}_{1i}+\rho\mathbf{H}_{11}\mathbf{H}_{11}^{\dag}+\sum_{k
\in \mathcal{S}, k\neq
1}\rho^{\alpha}\mathbf{H}_{1k}\mathbf{H}_{1k}^{\dagger}|-\log|\mathbf{I}+\sum_{i
\neq
1}\mathbf{H}_{1i}\mathbf{H}^{\dag}_{1i}+\rho^{1-\alpha}\mathbf{H}_{11}\mathbf{H}_{11}^{\dag}|\notag\\
&\stackrel{(a)}{=}&\log|\mathbf{I}+\rho\mathbf{H}_{11}\mathbf{H}_{11}^{\dag}+\sum_{k
\in \mathcal{S}, k\neq
1}\rho^{\alpha}\mathbf{H}_{1k}\mathbf{H}_{1k}^{\dagger}|-\log|\mathbf{I}+\sum_{i
\neq
1}\mathbf{H}_{1i}\mathbf{H}^{\dag}_{1i}+\rho^{1-\alpha}\mathbf{H}_{11}\mathbf{H}_{11}^{\dag}|+\mathcal{O}(1)\\
&\stackrel{(b)}{=}&(1+(|\mathcal{S}|-1)\alpha)\log\rho
-(1-\alpha)\log\rho
+\mathcal{O}(1)\\
&=&|\mathcal{S}|\alpha\log\rho +\mathcal{O}(1)
\end{eqnarray}
where $(a)$ follows from the fact that $\sum_{i \neq 1, i \notin
\mathcal{S}}\mathbf{H}_{1i}\mathbf{H}^{\dag}_{1i}$ is constant and
$(b)$ follows from Lemma 1. If $1\notin \mathcal{S}$, then
\eqref{macrate} is
\begin{eqnarray}
&&\log|\mathbf{I}+\sum_{i \neq 1, i \notin
\mathcal{S}}\mathbf{H}_{1i}\mathbf{H}^{\dag}_{1i}+\rho^{1-\alpha}\mathbf{H}_{11}\mathbf{H}_{11}^{\dag}+\sum_{k
\in
\mathcal{S}}\rho^{\alpha}\mathbf{H}_{1k}\mathbf{H}_{1k}^{\dagger}|-\log|\mathbf{I}+\sum_{i
\neq
1}\mathbf{H}_{1i}\mathbf{H}^{\dag}_{1i}+\rho^{1-\alpha}\mathbf{H}_{11}\mathbf{H}_{11}^{\dag}|\notag\\
&=&(1-\alpha+|\mathcal{S}|\alpha)\log\rho -(1-\alpha)\log\rho
+\mathcal{O}(1)\\
&=&|\mathcal{S}|\alpha\log\rho +\mathcal{O}(1)
\end{eqnarray}
Thus,
\begin{eqnarray}
&&\sum_{k\in \mathcal{S}}d_{kc}(\alpha) \leq |\mathcal{S}|\alpha\\
&\Rightarrow& d_{c}(\alpha)\leq \alpha
\end{eqnarray}
For $\mathcal{S}= \{2, 3, \cdots, N+1\}$,
\begin{eqnarray}
&&R_{2c}+\cdots+R_{(N+1)c}\notag\\
&\leq&
\log|\mathbf{I}+(\mathbf{I}+\sum_{i=2}^{N+1}\mathbf{H}_{1i}\mathbf{H}^{\dag}_{1i}+\rho^{1-\alpha}\mathbf{H}_{11}\mathbf{H}_{11}^{\dag})^{-1}(\rho^{\alpha}-1)\sum_{i=2}^{N+1}\mathbf{H}_{1i}\mathbf{H}^{\dag}_{1i}|\\
&=&\log|\mathbf{I}+\rho^{\alpha}\sum_{i=2}^{N+1}\mathbf{H}_{1i}\mathbf{H}^{\dag}_{1i}+\rho^{1-\alpha}\mathbf{H}_{11}\mathbf{H}_{11}^{\dag}|
-\log|\mathbf{I}+\sum_{i=2}^{N+1}\mathbf{H}_{1i}\mathbf{H}^{\dag}_{1i}+\rho^{1-\alpha}\mathbf{H}_{11}\mathbf{H}_{11}^{\dag}|\label{achcommon1}\\
&\stackrel{(a)}{=}&(\max\{\alpha,1-\alpha\}+(N-1)\alpha-(1-\alpha))\log\rho+\mathcal{O}(1)\\
&=&\max\{(N+1)\alpha-1,(N-1)\alpha\}\log\rho+\mathcal{O}(1)
\end{eqnarray}
where (a) follow from Lemma \ref{lemma:O(1)1}. Hence,
\begin{eqnarray}
&&d_{2c}(\alpha)+\cdots+d_{(N+1)c}(\alpha) \leq \max\{(N+1)\alpha-1,(N-1)\alpha\}\\
&\Rightarrow& d_{c}(\alpha)\leq
\frac{1}{N}\max\{(N+1)\alpha-1,(N-1)\alpha\}\label{active1}
\end{eqnarray}
For $\mathcal{S}= \{1, 2,\cdots, N+1\}$,
\begin{eqnarray}
&&R_{1c}+\cdots+R_{(N+1)c} \notag\\
&\leq& \log|\mathbf{I}+(\mathbf{I}+\sum_{i=2}^{N+1}\mathbf{H}_{1i}\mathbf{H}^{\dag}_{1i}+\rho^{1-\alpha}\mathbf{H}_{11}\mathbf{H}_{11}^{\dag})^{-1}[(\rho-\rho^{1-\alpha})\mathbf{H}_{11}\mathbf{H}_{11}^{\dag}+(\rho^{\alpha}-1)\sum_{i=2}^{N+1}\mathbf{H}_{1i}\mathbf{H}^{\dag}_{1i}]|\\
&=& \log|\mathbf{I}+\rho\mathbf{H}_{11}\mathbf{H}_{11}^{\dag}+\rho^{\alpha}\sum_{i=2}^{N+1}\mathbf{H}_{1i}\mathbf{H}^{\dag}_{1i}|-\log|\mathbf{I}+\sum_{i=2}^{N+1}\mathbf{H}_{1i}\mathbf{H}^{\dag}_{1i}+\rho^{1-\alpha}\mathbf{H}_{11}\mathbf{H}_{11}^{\dag}|\label{achcommon2}\\
&=& (1+(N-1)\alpha-(1-\alpha))\log\rho+\mathcal{O}(1)\\
&=& N\alpha\log\rho+\mathcal{O}(1)
\end{eqnarray}
Hence,
\begin{eqnarray}
&&d_{1c}(\alpha)+\cdots+d_{(N+1)c}(\alpha)\leq N \alpha \\
 &\Rightarrow& d_{c}(\alpha)\leq \frac{N}{N+1}\alpha\label{active2}
\end{eqnarray}
Due to symmetry, there are similar MAC constraints at each receiver.
The achievable rate region is the intersection of the MAC capacity
regions of each receiver. Note that only two sum constraints may be
active in terms of GDOF depending on $\alpha$. One is the sum rate
of all interfering messages given by \eqref{active1}. The other one
is the sum rate of all messages given by \eqref{active2}. The
maximum achievable symmetric generalized degrees of freedom in this
rate region is
\begin{equation}
d_c(\alpha) = \min\lbrace\frac{N\alpha}{N+1},
\max\{\frac{(N+1)\alpha-1}{N},\frac{(N-1)\alpha}{N}\}\rbrace\label{cgdof}
\end{equation}
 From \eqref{cgdof}, we have
\begin{eqnarray*}
d_{c}(\alpha)=\left \{\begin{array}{ccc} \frac{(N-1)\alpha}{N}&~& 0
\leq \alpha \leq \frac{1}{2}\\ \frac{(N+1)\alpha-1}{N} &~&
\frac{1}{2} \leq \alpha
\leq\frac{N+1}{2N+1}\\
\frac{N\alpha}{N+1}&~& \frac{N+1}{2N+1} \leq \alpha \leq 1
\end{array}\right.
\end{eqnarray*}
Therefore, the symmetric generalized degrees of freedom are
\begin{eqnarray*}
d_{\text{sym}}(\alpha)=d_c(\alpha)+d_p(\alpha)=\left
\{\begin{array}{ccc} 1-\frac{\alpha}{N}&~& 0 \leq
\alpha \leq \frac{1}{2}\\
\frac{N-1}{N}+\frac{\alpha}{N} &~& \frac{1}{2} \leq \alpha
\leq\frac{N+1}{2N+1}\\1-\frac{\alpha}{N+1} &~&\frac{N+1}{2N+1} \leq
\alpha \leq 1
\end{array}\right.
\end{eqnarray*}
It is interesting to compare the achievable GDOF using this scheme
with that achieved by treating all interference as noise. The rate
achieved by treating interference as noise is
\begin{eqnarray*}
R &=& \min_j\{\log|\mathbf{I}+(\mathbf{I}+\rho^{\alpha}\sum_{j \neq
i}\mathbf{H}_{ji}\mathbf{H}^{\dag}_{ji})^{-1}\rho
\mathbf{H}_{jj}\mathbf{H}^{\dag}_{jj}|\}\\
&=& (1-\alpha)\log\rho+\mathcal{O}(1)
\end{eqnarray*}
Therefore, the GDOF achieved by treating interference as noise is
$d(\alpha)=1-\alpha$, which is strictly less optimal than that
achieved using Han-Kobayashi type scheme.

Recall that for the 2 user interference channel with single antenna
nodes, when $\alpha \leq \frac{1}{2}$, treating interference as
noise is optimal in terms of GDOF. Why is treating interference as
noise so suboptimal for the SIMO case? Notice that the private
message already achieves $1-\alpha$ degrees of freedom which is the
same as that achieved by treating interference as noise. Therefore,
whether the Han-Kobayashi scheme is able to achieve more degrees of
freedom depends on if the common messages can achieve a non-zero
degrees of freedom. Let us first consider the 2 user SISO
interference channel. The common messages are decoded first by
treating private messages as noise. Due to symmetry, let us consider
Receiver 1. Although the private message from  Transmitter 2 is
received at the noise floor, the private message from Transmitter 1
is received at power $\rho^{1-\alpha}$. This essentially raises the
noise floor by $\rho^{1-\alpha}$ at the receiver when decoding the
common messages. When $\alpha \leq \frac{1}{2}$, the degrees of
freedom of the common message is limited by the degrees of freedom
achieved by the common message from the interfering transmitter. The
common message from Transmitter 2 is received roughly with power
$\rho^{\alpha}$, which is below the noise level, resulting in zero
degrees of freedom. This is illustrated in Fig. \ref{dofcommon1}.

\begin{figure*}
\centerline{\subfigure[2 user SISO interference channel
]{\includegraphics[width=2.6in, height=1.8in, trim= 0 100 0
0]{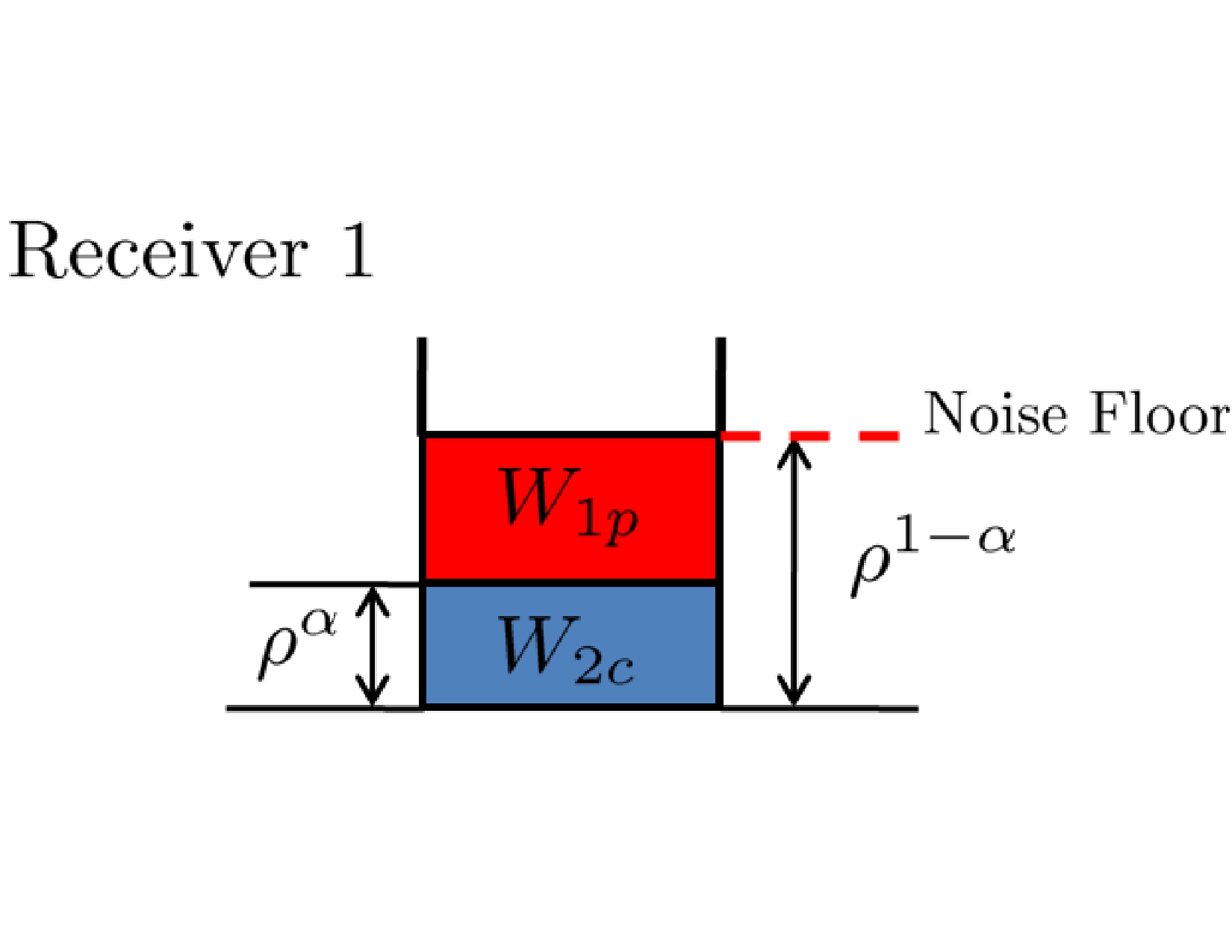} \label{dofcommon1}} \hfil \subfigure[3 user
$1\times 2$ SIMO interference channel]{\includegraphics[width=2.6in,
height=1.8in, trim= 0 150 0 0]{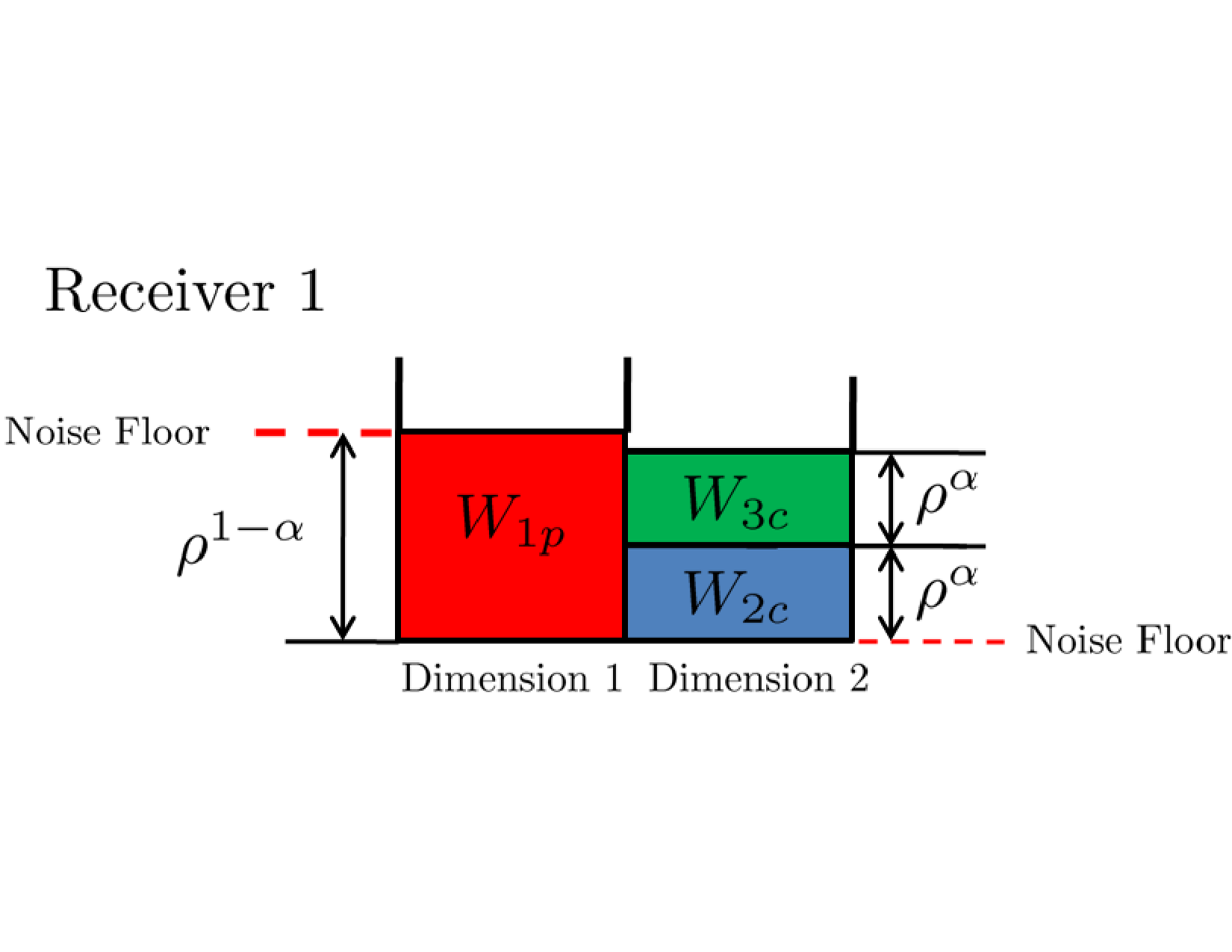}\label{dofcommon2} }}
\caption{Power of common messages in Receiver 1's signal space for 2
user SISO interference channel  and 3 user $1\times2$ SIMO case }
\end{figure*}
Now let us consider the $N+1$ user, $1 \times N$ SIMO interference
channel. For simplicity, consider the case when $N=2$. Again, the
private message achieves $1-\alpha$ degrees of freedom. Different
from the 2 user SISO case, common messages can achieve positive
degrees of freedom. Due to symmetry, let us consider Receiver 1. For
the $1\times 2$ SIMO interference channel, the receiver has a 2
dimensional signal space. The desired signal along channel vector
$\mathbf{H}_{11}$ occupies one dimension. In this one dimensional
subspace, similar to the analysis for the 2 user SISO case,
Transmitter 1's private message is received at power
$\rho^{1-\alpha}$ raising the noise level in this dimension;
however, in the other dimension, the noise level is not affected.
Thus, the common messages from Transmitter 2 and 3 together can
achieve $\alpha$ degrees of freedom in that orthogonal dimension.
This is illustrated in Fig. \ref{dofcommon2}.

\subsection{Outer bounds for the Generalized Degrees of Freedom}
We present three outer bounds which are tight for different regimes.
\begin{figure}[b]
 \centering
\includegraphics[width=5.2in, trim=0 90 0 0]{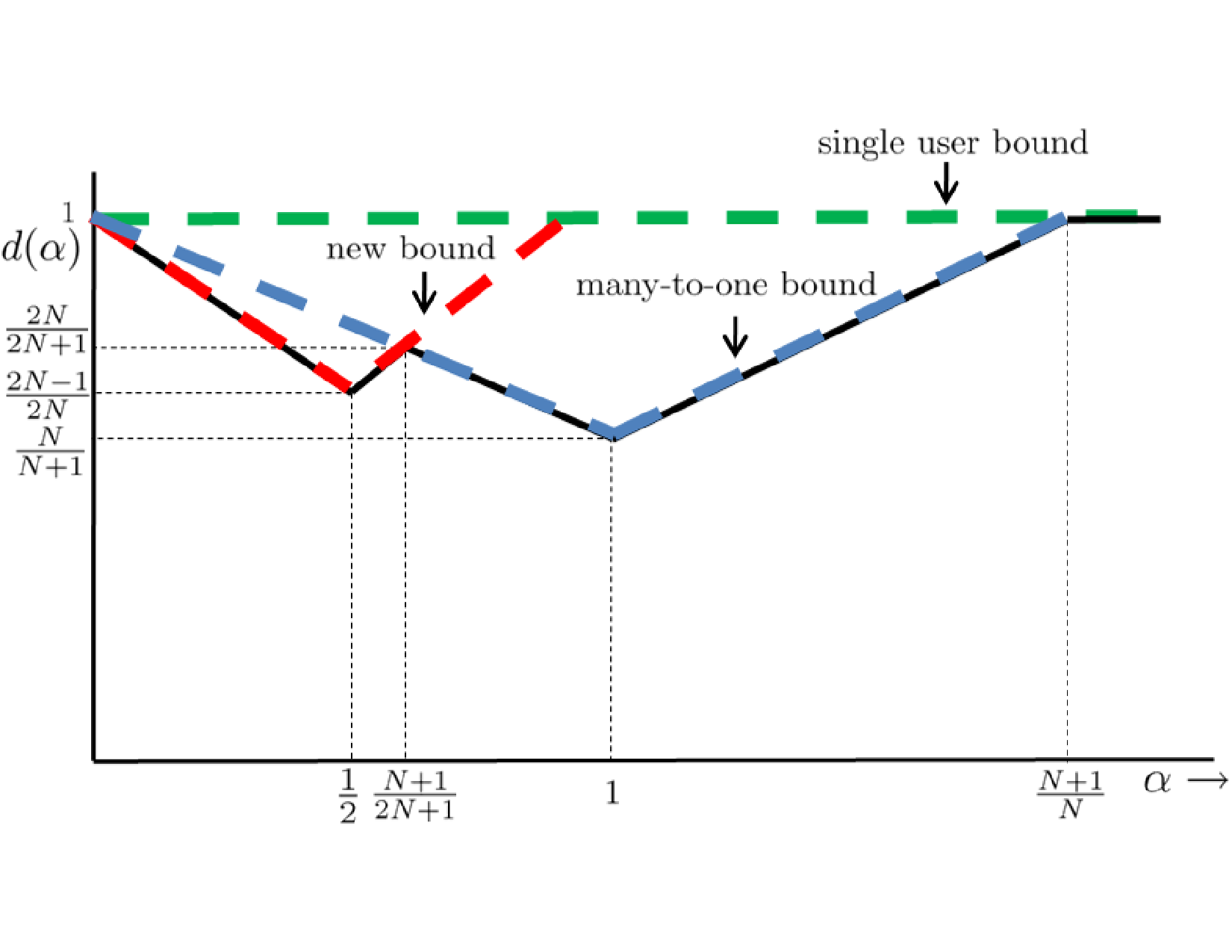}
\caption{Outer bounds for $N+1$ user $1\times N$ SIMO interference
channel} \label{fig:outerbound}
\end{figure}

\subsubsection{Single user bound}
For the point to point channel with a single transmit antenna and
$N$ receive antennas, the degrees of freedom is 1. The generalized
degrees of freedom \emph{per user} cannot be more than 1 with
interference, i.e., $d(\alpha)\leq 1$. As shown in Fig.
\ref{fig:outerbound}, this bound is tight for very strong
interference regime, i.e, $\alpha \geq \frac{N+1}{N}$.

\subsubsection{Many-to-one outer bound}
We first derive an outer bound for the general $K$ user SIMO
Gaussian interference channel (not necessarily symmetric) with $N$
antennas at each receiver. The channel is given by
\begin{eqnarray*}
\mathbf{Y}_j=\sum_{i=1}^K \mathbf{H}_{ji}x_i+\mathbf{Z}_j\quad
\forall j=1,\cdots, K
\end{eqnarray*}
where $E[|x_i|^2]\leq P_i$ and $\mathbf{Z}_j\sim
\mathcal{CN}(\mathbf{0},\mathbf{I})$.

\begin{lemma}\label{lemma:mtoo}
For the $K$ user SIMO Gaussian interference channel with $N$
antennas at each receiver, the sum capacity is bounded above by
\begin{eqnarray}
R_{\text{sum}}&\leq& \log|\mathbf{I}+(
\mathbf{I}+\sum_{i=2}^KP_i\mathbf{H}_{1i}\mathbf{H}_{1i}^{\dagger})^{-1}P_1\mathbf{H}_{11}\mathbf{H}_{11}^{\dagger}|+
\sum_{i=2}^K \log(1+P_i\|\mathbf{H}_{ii}\|^2)\notag\\&&+\log|
\mathbf{I}+\sum_{i=2}^K\frac{P_i}{1+\|\mathbf{H}_{ii}\|^2P_i}\mathbf{H}_{1i}\mathbf{H}_{1i}^{\dagger}|\label{manytoonebound}
\end{eqnarray}
\end{lemma}

\begin{proof}
This outer bound is a natural generalization of the one-sided
interference channel outer bound for the two user interference
channel with single antenna nodes \cite{onebit}. More specifically,
we will derive an outer bound on the interference channel where only
Receiver 1 sees interference. Clearly, this outer bound is also an
outer bound for the original interference channel. Suppose the genie
provides side information $x_i$, $\forall i=1,\ldots,K~i\neq j$ to
Receiver $j$ $\forall j=2,\ldots, K$. That is, for all receiver but
Receiver 1, the genie provides all interference signals to it. So
Receiver $j$ $\forall j=2,\ldots, K$ can subtract the interference
from their received signals. Then, the genie-aided channel is
\begin{eqnarray}
\mathbf{Y}_1 &=& \mathbf{H}_{11}x_1+\mathbf{H}_{12}x_2+\cdots+\mathbf{H}_{1K}x_{K}+\mathbf{Z}_1\notag\\
&=& \mathbf{H}_{11}x_1
+\underline{\mathbf{H}}_{12}\underline{\mathbf{X}}_2+\mathbf{Z}_1\\
\mathbf{Y'}_j &=& \mathbf{H}_{jj}x_j+\mathbf{Z}_j\label{gr2}
~\forall j=2,\ldots, K
\end{eqnarray}
where
$\underline{\mathbf{H}}_{12}=[\mathbf{H}_{12}~\cdots~\mathbf{H}_{1K}]$.
\eqref{gr2} is equivalent to
\begin{eqnarray*}
y_j=\|\mathbf{H}_{jj}\|x_j+z_j
\end{eqnarray*}
where $z_j \sim \mathcal{CN}(0,1)$. Now let all transmitters but
Transmitter 1 cooperate as one transmitter and their corresponding
receivers cooperate as one receiver. Then it is equivalent to a two
user one-sided interference channel. Since allowing transmitters to
cooperate cannot decrease the capacity, the capacity region of this
channel is an outer bound of the capacity region of the genie-aided
channel. The received signal at Receiver 2 of this two user
one-sided interference channel is
\begin{eqnarray*}
\underline{\mathbf{Y}}_2=\underline{\mathbf{H}}_2
\underline{\mathbf{X}}_2+\underline{\mathbf{Z}}_2
\end{eqnarray*}
where $\underline{\mathbf{Y}}_2=[y_2~\cdots~y_{K}]^T$,
$\underline{\mathbf{Z}}_2=[z_2~\cdots~z_{K}]^T$ and
\begin{eqnarray*}
\underline{\mathbf{H}}_2&=&\left[
\begin{array}{ccc}\|\mathbf{H}_{22}\|& \cdots& 0\\
\vdots & \ddots& \vdots\\ 0 &\cdots &
\|\mathbf{H}_{KK}\|\end{array}\right].
\end{eqnarray*}
Now we can bound the sum rate on this two user one-sided
interference channel by providing side information
$\mathbf{S}=\underline{\mathbf{H}}_{12}\underline{\mathbf{X}}_2+\mathbf{Z}_1$
to Receiver 2. Starting from Fano's inequality, we have
\begin{eqnarray}
&&nR_{\text{sum}}-n\epsilon_n \notag\\
&\leq&
I(x_1^n;\mathbf{Y}_1^n)+I(\underline{\mathbf{X}}_2^n;\underline{\mathbf{Y}}_2^n,\mathbf{S}^n)\notag\\
&=&
h(\mathbf{Y}_1^n)-h(\mathbf{Y}_1^n|x_1^n)+I(\underline{\mathbf{X}}_2^n;\mathbf{S}^n)+I(\underline{\mathbf{X}}_2^n;\underline{\mathbf{Y}}_2^n|\mathbf{S}^n)\notag\\
&=&
h(\mathbf{Y}_1^n)-h(\mathbf{S}^n)+h(\mathbf{S}^n)-h(\mathbf{S}^n|\underline{\mathbf{X}}_2^n)+I(\underline{\mathbf{X}}_2^n;\underline{\mathbf{Y}}_2^n|\mathbf{S}^n)\notag\\
&=&h(\mathbf{Y}_1^n)-h(\mathbf{Z}_1^n)+h(\underline{\mathbf{Y}}_2^n|\mathbf{S}^n)-h(\underline{\mathbf{Z}}_2^n|\mathbf{Z}_1^n)\notag\\
&\stackrel{(a)}{\leq}&
nh(\mathbf{Y}_{1}^*)-nh(\mathbf{Z}_1)+nh(\underline{\mathbf{Y}}_{2}^*|\mathbf{S}^*)-nh(\underline{\mathbf{Z}}_2|\mathbf{Z}_1)\label{manytooneentroy}\\
&=& n I(x_{1}^*;\mathbf{Y}_{1}^*)+n
I(\underline{\mathbf{X}}_{2}^*;\underline{\mathbf{Y}}_{2}^*,\mathbf{S}^*)
\label{outerbound}
\end{eqnarray}
where * denotes the inputs are i.i.d Gaussian with maximum power,
i.e. $x_i^* \sim \mathcal{CN}(0,P_i)$ and $\mathbf{Y}_{i}^*$ and
$\mathbf{S}^*$ are the corresponding signals. The fact that
$h(\underline{\mathbf{Y}}_2^n|\mathbf{S}^n) \leq
nh(\underline{\mathbf{Y}}_{2}^*|\mathbf{S}^*)$ in step (a) follows
from Lemma 1 in \cite{Sreekanth_Veeravalli}.

Now we calculate each term in \eqref{outerbound}. First,
\begin{eqnarray}
I(x_{1}^*;\mathbf{Y}_{1}^*)=\log|\mathbf{I}+(
\mathbf{I}+\sum_{i=2}^KP_i\mathbf{H}_{1i}\mathbf{H}_{1i}^{\dagger})^{-1}P_1\mathbf{H}_{11}\mathbf{H}_{11}^{\dagger}|\label{mto1}
\end{eqnarray}
Since $
I(\underline{\mathbf{X}}_{2}^*;\underline{\mathbf{Y}}_{2}^*,\mathbf{S}^*)=I(\underline{\mathbf{X}}_{2}^*;\underline{\mathbf{Y}}_{2}^*)+I(\underline{\mathbf{X}}_{2}^*;\mathbf{S}^*|\underline{\mathbf{Y}}_{2}^*)$
where
\begin{equation}
I(\underline{\mathbf{X}}_{2}^*;\underline{\mathbf{Y}}_{2}^*)=\sum_{i=2}^K
\log(1+P_i\|\mathbf{H}_{ii}\|^2) \label{mto2}
\end{equation}
and
$I(\underline{\mathbf{X}}_{2}^*;\mathbf{S}^*|\underline{\mathbf{Y}}_{2}^*)=h(\mathbf{S}^*|\underline{\mathbf{Y}}_{2}^*)-h(\mathbf{S}^*|\underline{\mathbf{Y}}_{2}^*,\underline{\mathbf{X}}_{2}^*)$.
Let $\Sigma_{\mathbf{S}^*|\underline{\mathbf{Y}}_{2}^*}$ be the
covariance matrix of $\mathbf{S}^*|\underline{\mathbf{Y}}_{2}^*$.
Then
\begin{equation*}
h(\mathbf{S}^*|\underline{\mathbf{Y}}_{2}^*) =\log|\pi e
\Sigma_{\mathbf{S}^*|\underline{\mathbf{Y}}_{2}^*}|
\end{equation*}
where
\begin{eqnarray*}
\Sigma_{\mathbf{S}^*|\underline{\mathbf{Y}}_{2}^*}&=&E[\mathbf{S}^*
\mathbf{S}^{* \dagger}]-E[\mathbf{S}^* \underline{\mathbf{Y}}_2^{*
\dagger}]E[\underline{\mathbf{Y}}_{2}^*\underline{\mathbf{Y}}_2^{*
\dagger}]^{-1}E[\underline{\mathbf{Y}}_{2}^*\mathbf{S}^{*
\dagger}]\\
&=&\mathbf{I}+
\mathbf{\underline{H}}_{12}\mathbf{P}_2\mathbf{\underline{H}}_{12}^{\dagger}-\mathbf{\underline{H}}_{12}\mathbf{P}_2\mathbf{\underline{H}}_{2}^{\dagger}(\mathbf{I}+\mathbf{\underline{H}}_{2}\mathbf{P}_2\mathbf{\underline{H}}_{2}^{\dagger})^{-1}\mathbf{\underline{H}}_{2}\mathbf{P}_2\mathbf{\underline{H}}_{12}^{\dagger}\\
&=&\mathbf{I}+\mathbf{\underline{H}}_{12}\big(
\mathbf{P}_2-\mathbf{P}_2\mathbf{\underline{H}}_{2}^{\dagger}(\mathbf{I}+\mathbf{\underline{H}}_{2}\mathbf{P}_2\mathbf{\underline{H}}_{2}^{\dagger})^{-1}\mathbf{\underline{H}}_{2}\mathbf{P}_2\big)\mathbf{\underline{H}}_{12}^{\dagger}\\
&\stackrel{(a)}{=}&\mathbf{I}+\mathbf{\underline{H}}_{12}(
\mathbf{P}_2^{-1}+\mathbf{\underline{H}}_{2}^{\dagger}\mathbf{\underline{H}}_{2})^{-1}\mathbf{\underline{H}}_{12}^{\dagger}
\end{eqnarray*}
where $(a)$ follows from Woodbury matrix identity \cite{woodbury},
which is
\begin{eqnarray}
(\mathbf{A}+\mathbf{B}\mathbf{C}\mathbf{D})^{-1}=\mathbf{A}^{-1}-\mathbf{A}^{-1}\mathbf{B}(\mathbf{C}^{-1}+\mathbf{D}\mathbf{A}^{-1}\mathbf{B})^{-1}\mathbf{D}\mathbf{A}^{-1}
\end{eqnarray}
where $\mathbf{A}$, $\mathbf{B}$, $\mathbf{C}$ and $\mathbf{D}$ are
$n\times n$, $n\times k$, $k\times k$ and $k \times n$ matrices,
respectively. And
\begin{eqnarray*}
\mathbf{\underline{H}}_{2}^{\dagger}\mathbf{\underline{H}}_{2}=\left[
\begin{array}{ccc}\|\mathbf{H}_{22}\|^2& \cdots& 0\\
\vdots & \ddots& \vdots\\ 0 &\cdots &
\|\mathbf{H}_{KK}\|^2\end{array}\right] \quad \mathbf{P}_2=\left[
\begin{array}{ccc}P_2& \cdots& 0\\
\vdots & \ddots& \vdots\\ 0 &\cdots & P_K\end{array}\right]
\end{eqnarray*}
Therefore, we have
\begin{eqnarray}
h(\mathbf{S}^*|\underline{\mathbf{Y}}_{2}^*) =\log|\pi e
(\mathbf{I}+\mathbf{\underline{H}}_{12}(
\mathbf{P}_2^{-1}+\mathbf{\underline{H}}_{2}^{\dagger}\mathbf{\underline{H}}_{2})^{-1}\mathbf{\underline{H}}_{12}^{\dagger})|
\end{eqnarray}
and
\begin{eqnarray}
I(\underline{\mathbf{X}}_{2}^*;\mathbf{S}^*|\underline{\mathbf{Y}}_{2}^*)
&=&h(\mathbf{S}^*|\underline{\mathbf{Y}}_{2}^*)-h(\mathbf{S}^*|\underline{\mathbf{Y}}_{2}^*,\underline{\mathbf{X}}_{2}^*)\notag\\
&=&h(\mathbf{S}^*|\underline{\mathbf{Y}}_{2}^*)-h(\mathbf{Z}_1)\\
&=&\log|\pi e (\mathbf{I}+\mathbf{\underline{H}}_{12}(
\mathbf{P}_2^{-1}+\mathbf{\underline{H}}_{2}^{\dagger}\mathbf{\underline{H}}_{2})^{-1}\mathbf{\underline{H}}_{12}^{\dagger})|-N\log(\pi e )\notag\\
&=&\log|\mathbf{I}+\mathbf{\underline{H}}_{12}(
\mathbf{P}_2^{-1}+\mathbf{\underline{H}}_{2}^{\dagger}\mathbf{\underline{H}}_{2})^{-1}\mathbf{\underline{H}}_{12}^{\dagger}|\\
&=&\log|
\mathbf{I}+\sum_{i=2}^K\frac{P_i}{1+\|\mathbf{H}_{ii}\|^2P_i}\mathbf{H}_{1i}\mathbf{H}_{1i}^{\dagger}|\label{mto3}
\end{eqnarray}
Adding \eqref{mto1}, \eqref{mto2} and \eqref{mto3}, we prove the
lemma.
\end{proof}
Applying the outer bound to the symmetric case, we have an outer
bound for the GDOF which is tight in the regime where
$\frac{N+1}{2N+1} \leq \alpha \leq \frac{N+1}{N}$. Before we present
the GDOF bound, let us see intuitively why this bound is good for
$\frac{N+1}{2N+1} \leq \alpha \leq \frac{N+1}{N}$. Recall that in
this regime, at Receiver 1 the MAC constraint that is active for the
common messages is $R_{1c}+\cdots+R_{(N+1)c}$. Notice that the
received signal $\mathbf{Y}_{1}^*$ at Receiver 1 roughly contains
common information from Transmitter 1 to $N+1$ and private message
from Transmitter 1. Since the power of private messages from
Transmitter 2 to $N+1$ is at noise floor at Receiver 1, they do not
reveal to Receiver 1. Therefore, we can think that
$R_1+R_{2c}+\cdots+R_{(N+1)c} \leq
h(\mathbf{Y}_{1}^*)-h(\mathbf{Z}_1)$. On the other hand,
$\underline{\mathbf{Y}}_{2}^*|\mathbf{S}^*$ roughly contains the
private information from transmitter 2 to $N+1$, since the
interfering signal $\mathbf{S}^*$ roughly contains the common
information of transmitter 2 to $N+1$. Thus, we can think that
$R_{2p}+\cdots+R_{(N+1)p} \leq
h(\underline{\mathbf{Y}}_{2}^*|\mathbf{S}^*)-h(\underline{\mathbf{Z}}_2|\mathbf{Z}_1)$.
Adding this one to the previous constraint, we have the outer bound
for $R_1+\cdots+R_{N+1}$.

{\it Remark}: Note that the corresponding outer bound for the
deterministic channel is $R_1+R_2+R_3\leq
H(Y_1|V_1V_2V_3)+H(Y_2|V_1V_2V_3)+H(Y_3)$. As we can see, they are
very similar, except that there is no noise term for the
deterministic case.

\begin{lemma}
For $\frac{N+1}{2N+1} \leq \alpha \leq \frac{N+1}{N}$, the symmetric
generalized degrees of freedom of the $N+1$ user $1\times N$ SIMO
Gaussian interference channel are bounded above as
\begin{eqnarray*}
d_{\text{sym}}(\alpha)\leq\left\{\begin{array}{ccc}\frac{N}{N+1}\alpha
&~& \alpha \geq 1\\1-\frac{\alpha}{N+1} &~& \alpha \leq
1\end{array}\right.
\end{eqnarray*}
\end{lemma}

\begin{proof}
Applying Lemma \ref{lemma:mtoo} to the $N+1$ user symmetric case
defined in \eqref{channelmodel}, we have
\begin{eqnarray}
R_{\text{sum}} &\leq&
\log|\mathbf{I}+(\mathbf{I}+\rho^{\alpha}\underline{\mathbf{H}}_{12}\underline{\mathbf{H}}_{12}^{\dag})^{-1}
\rho\mathbf{H}_{11}\mathbf{H}_{11}^{\dag}|+N\log(1+\rho)+\log|\mathbf{I}+\frac{\rho^{\alpha}}{1+\rho}\underline{\mathbf{H}}_{12}\underline{\mathbf{H}}_{12}^{\dag}|
\end{eqnarray}
where
$\underline{\mathbf{H}}_{12}=[\mathbf{H}_{12}~\cdots~\mathbf{H}_{1N+1}]$.
Next we calculate the degrees of freedom associated with each term
in the above equation.
\begin{eqnarray}
\log|\mathbf{I}+(\mathbf{I}+\rho^{\alpha}\underline{\mathbf{H}}_{12}\underline{\mathbf{H}}_{12}^{\dag})^{-1}
\rho\mathbf{H}_{11}\mathbf{H}_{11}^{\dag}|
=\left\{ \begin{array}{ccc} \mathcal{O}(1) &~& \alpha \geq 1\\
(1-\alpha)\log\rho+\mathcal{O}(1) &~& \alpha \leq
1\end{array}\right.\label{firstterm}
\end{eqnarray}
and
\begin{equation}
N\log(1+\rho)=N\log \rho +\mathcal{O} (1) \label{secondterm}
\end{equation}
and
\begin{eqnarray}
\log|\mathbf{I}+\frac{\rho^{\alpha}}{1+\rho}\underline{\mathbf{H}}_{12}\underline{\mathbf{H}}_{12}^{\dag}|=\left\{\begin{array}{ccc}N(\alpha-1)\log\rho+\mathcal{O}(1)&~&\alpha
\geq 1\\ \mathcal{O}(1)&~&\alpha \leq
1\end{array}\right.\label{thirdterm}
\end{eqnarray}
From \eqref{firstterm}, \eqref{secondterm}, \eqref{thirdterm}, we
have
\begin{eqnarray}
R_{\text{sum}} \leq \left\{\begin{array}{ccc}N\alpha
\log\rho+\mathcal{O}(1) &~& \alpha \geq 1\\(N+1-\alpha)\log\rho +
\mathcal{O}(1) &~& \alpha \leq 1\end{array}\right.
\end{eqnarray}
Therefore, the symmetric GDOF are bounded above as
\begin{equation*}
d_{\text{sym}}(\alpha)\leq\left\{\begin{array}{ccc}\frac{N}{N+1}\alpha
&~& \alpha \geq 1\\1-\frac{\alpha}{N+1} &~& \alpha \leq
1\end{array}\right.
\end{equation*}
\end{proof}
As shown in Fig. \ref{fig:outerbound}, this bound gives a tight
outer bound for the second ``V'' part of the W curve.

\subsubsection{A new outer bound}
Again, we derive an outer bound for the general SIMO Gaussian
interference channel. Then we apply this bound to the symmetric case
and show that it is tight in terms of GDOF in the regime where
$0\leq \alpha\leq \frac{N+1}{2N+1}$.
\begin{lemma}\label{lemma:Kusernewouterbound}
 For the $K>2$ user
SIMO Gaussian interference channel with $N$ antennas at each
receiver, we have the following bound:
\begin{eqnarray*}
&&R_1+2(R_2+\cdots+R_{K-1})+R_K \notag \\& \leq&
\log|\mathbf{I}+\sum_{i=2}^K P_i
\mathbf{H}_{1i}\mathbf{H}_{1i}^{\dagger}+\frac{P_1}{1+\|\mathbf{H}_{K1}\|^2P_1}\mathbf{H}_{11}\mathbf{H}_{11}^{\dagger}|\\
&&+\log|\mathbf{I}+\sum_{i=1}^{K-1} P_i
\mathbf{H}_{Ki}\mathbf{H}_{Ki}^{\dagger}+\frac{P_K}{1+\|\mathbf{H}_{1K}\|^2P_K}\mathbf{H}_{KK}\mathbf{H}_{KK}^{\dagger}|\\
&&+\sum_{i=2}^{K-1}\big(\log|\mathbf{I}+\underline{\mathbf{H}}_{i1}(\underline{\mathbf{P}}_{i1}^{-1}+\mathbf{\underline{H}}_{K1}^{\dagger}\mathbf{\underline{H}}_{K1})^{-1}\mathbf{\underline{H}}_{i1}^{\dagger}+\underline{\mathbf{H}}_{i2}(\underline{\mathbf{P}}_{i2}^{-1}+\mathbf{\underline{H}}_{12}^{\dagger}\mathbf{\underline{H}}_{12})^{-1}\mathbf{\underline{H}}_{i2}^{\dagger}|\\
&&+
\log|\mathbf{I}+\underline{\mathbf{H}}_{i3}(\underline{\mathbf{P}}_{i3}^{-1}+\mathbf{\underline{H}}_{13}^{\dagger}\mathbf{\underline{H}}_{13})^{-1}\mathbf{\underline{H}}_{i3}^{\dagger}+\underline{\mathbf{H}}_{i4}(\underline{\mathbf{P}}_{i4}^{-1}+\mathbf{\underline{H}}_{K4}^{\dagger}\mathbf{\underline{H}}_{K4})^{-1}\mathbf{\underline{H}}_{i4}^{\dagger}|\big)
\end{eqnarray*}
where $ \mathbf{\underline{H}}_{i1}=[\mathbf{H}_{i1} \cdots
\mathbf{H}_{ii}],~ \mathbf{\underline{H}}_{i2}=[\mathbf{H}_{ii+1}
\cdots \mathbf{H}_{iK}] ~
\mathbf{\underline{H}}_{i3}=[\mathbf{H}_{iK} ~\mathbf{H}_{i2}\cdots
\mathbf{H}_{ii}],~ \mathbf{\underline{H}}_{i4}=[\mathbf{H}_{i1}~
\mathbf{H}_{ii+1} \cdots \mathbf{H}_{iK-1}]$ and
$\mathbf{\underline{H}}_{K1}=[\mathbf{H}_{K1} \cdots
\mathbf{H}_{Ki}],~ \mathbf{\underline{H}}_{12}=[\mathbf{H}_{1i+1}
\cdots \mathbf{H}_{1K}]
,~\mathbf{\underline{H}}_{13}=[\mathbf{H}_{1K}~\mathbf{H}_{12}
\cdots \mathbf{H}_{1i}],~
\mathbf{\underline{H}}_{K4}=[\mathbf{H}_{K1}~\mathbf{H}_{Ki+1}
\cdots \mathbf{H}_{KK-1}]$. And
\begin{eqnarray*}
\underline{\mathbf{P}}_{i1}=\left[\begin{array}{ccc}P_1 & \cdots & 0 \\
\vdots & \ddots & \vdots\\ 0 & \cdots & P_i \end{array}\right]~~
\underline{\mathbf{P}}_{i2}=\left[\begin{array}{ccc}P_{i+1} & \cdots & 0 \\
\vdots & \ddots & \vdots\\ 0 & \cdots & P_K \end{array}\right]
\end{eqnarray*}
\begin{eqnarray*}
\underline{\mathbf{P}}_{i3}=\left[\begin{array}{cccc}P_K & 0 & \cdots & 0 \\
0& P_2 & \cdots & 0 \\ \vdots & \vdots & \ddots& \vdots\\ 0 & 0&
\cdots &P_i
\end{array}\right]~~
\underline{\mathbf{P}}_{i4}=\left[\begin{array}{cccc}P_1 & 0 & \cdots & 0 \\
0& P_{i+1} & \cdots & 0 \\ \vdots & \vdots & \ddots& \vdots\\ 0 & 0&
\cdots & P_{K-1}
\end{array}\right]
\end{eqnarray*}
\end{lemma}

Note that this bound is the counterpart of the ETW bound derived for
the 2 user Gaussian interference channel in \cite{onebit}. However,
the nature of this bound is significantly different from the two
user case. In the two user case, we simply have a sum rate bound,
but as seen here, with more than 2 users this is not a sum-rate
bound.

\begin{proof}
The outer bound is obtained by providing side information to
receivers such that unwanted terms can be canceled. Let
$\mathbf{S}_{j,\mathcal{B}}=\sum_{i \in \mathcal{B}}
\mathbf{H}_{ji}x_i+\mathbf{Z}_j$ where $\mathcal{B}$ is a set of
transmitters. And let $\mathcal{A}$ denote the set of all
transmitters, i.e., $\mathcal{A}=\{1,2,\cdots, K\}$. The notation
$\mathcal{A}\backslash \mathcal{B}$ means the complement of
$\mathcal{B}$ in $\mathcal{A}$. Then, the outer bound is
\begin{eqnarray}
&&R_1+2(R_2+\cdots+R_{K-1})+R_K \notag\\
&\leq&
h(\mathbf{Y}_1^*|\mathbf{S}_{K,1}^*)+h(\mathbf{Y}_K^*|\mathbf{S}_{1,K}^*)+\sum_{i=2}^{K-1}\big(h(\mathbf{Y}_i^*|\mathbf{S}_{K,\mathcal{A}\backslash
\{K,2,\ldots, i\}}^*,\mathbf{S}_{1,\{K,2,\ldots,
i\}}^*)+h(\mathbf{Y}_i^*|\mathbf{S}_{1,\mathcal{A}\backslash
\{1,2,\ldots, i\}}^*,\mathbf{S}_{K,\{1,2,\ldots, i\}}^*)
\big)\notag\\
&&-2\sum_{i=2}^{K-1}h(\mathbf{Z}_i)-h(\mathbf{Z}_1)-h(\mathbf{Z}_K)\label{newouterbound}
\end{eqnarray}
where * denotes the inputs are i.i.d Gaussian with maximum power,
i.e. $x_i^* \sim \mathcal{CN}(0,P_i)$ and $\mathbf{Y}_{i}^*$ and
$\mathbf{S}^*$ are the corresponding signals.

Here we only provide a proof for the case when $K=3$. The proof for
arbitrary $K>3$ is provided in the Appendix. The outer bound is
derived by giving side information to receivers. In general, it is
very difficult to guess what genie information should be given to
receivers. Here, we can easily figure out the appropriate genie
information by using the hints provided by the deterministic
channel. Notice that for the deterministic channel, the outer bounds
are in terms of $V_1$, $V_2$ and $V_3$. We first determine the
counterparts to $V_1, V_2, V_3$ in the Gaussian case. Let
$\mathbf{S}_{j,\mathcal{B}}=\sum_{i \in \mathcal{B}}
\mathbf{H}_{ji}x_i+\mathbf{Z}_j$ where $\mathcal{B} \subseteq
\{1,2,3\}$ is a set of transmitters. Then the counterparts to $V_1$,
$V_2$  and $V_3$ should be $\mathbf{S}_{3,1}$ or $\mathbf{S}_{2,1}$,
$\mathbf{S}_{1,2}$ or $\mathbf{S}_{3,2}$ and $\mathbf{S}_{1,3}$ or
$\mathbf{S}_{2,3}$, respectively. Replacing $V_i$ in the
deterministic outer bounds with the Gaussian counterparts and
roughly calculating the generalized degrees of freedom of the outer
bounds, we identify the following bound is tight in terms of GDOF
for the very weak interference regime, i.e., $0\leq \alpha \leq
\frac{3}{5}$.
\begin{eqnarray}
R_1+2R_2+R_3 \leq H(Y_1|V_1)+
2H(Y_2|V_1V_2V_3)+H(Y_3|V_3)\label{deterministicbound2}
\end{eqnarray}
Thus we expect the outer bound for the Gaussian case will be similar
to this bound. Consider the first term in
\eqref{deterministicbound2}. The counterpart of $H(Y_1|V_1)$ should
be $h(\mathbf{Y}_1|\mathbf{S}_{3,1})$ or
$h(\mathbf{Y}_1|\mathbf{S}_{2,1})$. Now we choose
$h(\mathbf{Y}_1|\mathbf{S}_{3,1})$. In order to get
$h(\mathbf{Y}_1|\mathbf{S}_{3,1})$, we need to provide side
information $\mathbf{S}_{3,1}$ to Receiver 1. Then, we have
\begin{eqnarray}
I(x_1^n;\mathbf{Y}_1^n)
&\leq& I(x_1^n;\mathbf{Y}_1^n,\mathbf{S}_{3,1}^n)\notag\\
&=&I(x_1^n;\mathbf{S}_{3,1}^n)+I(x_1^n;\mathbf{Y}_1^n|\mathbf{S}_{3,1}^n)\notag\\
&=&h(\mathbf{S}_{3,1}^n)-h(\mathbf{S}_{3,1}^n|x_1^n)+h(\mathbf{Y}_1^n|\mathbf{S}_{3,1}^n)-h(\mathbf{Y}_1^n|\mathbf{S}_{3,1}^n,x_1^n)\notag\\
&=&h(\mathbf{S}_{3,1}^n)-h(\mathbf{Z}_{3}^n)+h(\mathbf{Y}_1^n|\mathbf{S}_{3,1}^n)-h(\mathbf{S}_{1,\{2,3\}}^n)\label{newboundterm1}
\end{eqnarray}
Comparing \eqref{newboundterm1} with \eqref{deterministicbound2}, we
can see that $h(\mathbf{S}_{3,1}^n)$ and
$h(\mathbf{S}_{1,\{2,3\}}^n)$ are unwanted terms. So we would like
to give appropriate side information to other receivers such that
they can be canceled. On the other hand, in order to have terms
similar to $H(Y_2|V_1V_2V_3)$, we should provide the counterparts of
$V_1V_2V_3$ to Receiver 2. Based on these two considerations, we
give $\mathbf{S}_{3,\{1,2\}},\mathbf{S}_{1,3}$ to Receiver 2. Then,
we have
\begin{eqnarray}
I(x_2^n;\mathbf{Y}_2^n)
&\leq&I(x_2^n;\mathbf{Y}_2^n,\mathbf{S}_{3,\{1,2\}}^n,\mathbf{S}_{1,3}^n)\notag\\
&=& I(x_2^n;\mathbf{S}_{3,\{1,2\}}^n)+I(x_2^n;\mathbf{Y}_2^n,\mathbf{S}_{1,3}^n|\mathbf{S}_{3,\{1,2\}}^n)\notag\\
&=&h(\mathbf{S}_{3,\{1,2\}}^n)-h(\mathbf{S}_{3,\{1,2\}}^n|x_2^n)+I(x_2^n;\mathbf{S}_{1,3}^n|\mathbf{S}_{3,\{1,2\}}^n)+I(x_2^n;\mathbf{Y}_2^n|\mathbf{S}_{1,3}^n,\mathbf{S}_{3,\{1,2\}}^n)\notag\\
&=&h(\mathbf{S}_{3,\{1,2\}}^n)-h(\mathbf{S}_{3,1}^n)+I(x_2^n;\mathbf{Y}_2^n|\mathbf{S}_{1,3}^n,\mathbf{S}_{3,\{1,2\}}^n)\notag\\
&=&h(\mathbf{S}_{3,\{1,2\}}^n)-h(\mathbf{S}_{3,1}^n)+h(\mathbf{Y}_2^n|\mathbf{S}_{1,3}^n,\mathbf{S}_{3,\{1,2\}}^n)-h(\mathbf{Y}_2^n|\mathbf{S}_{1,3}^n,\mathbf{S}_{3,\{1,2\}}^n,x_2^n)\notag\\
&\leq&h(\mathbf{S}_{3,\{1,2\}}^n)-h(\mathbf{S}_{3,1}^n)+h(\mathbf{Y}_2^n|\mathbf{S}_{1,3}^n,\mathbf{S}_{3,\{1,2\}}^n)-h(\mathbf{Y}_2^n|\mathbf{S}_{1,3}^n,\mathbf{S}_{3,\{1,2\}}^n,x_1^n, x_2^n, x_3^n)\notag\\
&=&h(\mathbf{S}_{3,\{1,2\}}^n)-h(\mathbf{S}_{3,1}^n)+h(\mathbf{Y}_2^n|\mathbf{S}_{1,3}^n,\mathbf{S}_{3,\{1,2\}}^n)-h(\mathbf{Z}_2^n)\label{newboundterm2}
\end{eqnarray}
Adding \eqref{newboundterm1} and \eqref{newboundterm2}, we have
\begin{eqnarray*}
&&I(x_1^n;\mathbf{Y}_1^n)+I(x_2^n;\mathbf{Y}_2^n)\\
&\leq&
h(\mathbf{S}_{3,1}^n)-h(\mathbf{Z}_{3}^n)+h(\mathbf{Y}_1^n|\mathbf{S}_{3,1}^n)-h(\mathbf{S}_{1,\{2,3\}}^n) +h(\mathbf{S}_{3,\{1,2\}}^n)-h(\mathbf{S}_{3,1}^n)+h(\mathbf{Y}_2^n|\mathbf{S}_{1,3}^n,\mathbf{S}_{3,\{1,2\}}^n)-h(\mathbf{Z}_2^n)\\
&=&h(\mathbf{Y}_1^n|\mathbf{S}_{3,1}^n)+h(\mathbf{Y}_2^n|\mathbf{S}_{1,3}^n,\mathbf{S}_{3,\{1,2\}}^n)+h(\mathbf{S}_{3,\{1,2\}}^n)-h(\mathbf{S}_{1,\{2,3\}}^n)
-h(\mathbf{Z}_2^n)-h(\mathbf{Z}_{3}^n)
\end{eqnarray*}
Note that
$h(\mathbf{Y}_2^n|\mathbf{S}_{1,3}^n,\mathbf{S}_{3,\{1,2\}}^n)$ is
the counterpart to $H(Y_2|V_1V_2V_3)$ and the unwanted term
$h(\mathbf{S}_{3,1}^n)$ is canceled. Similarly, we have
\begin{eqnarray*}
I(x_3^n;\mathbf{Y}_3^n)+I(x_2^n;\mathbf{Y}_2^n) \leq
h(\mathbf{Y}_3^n|\mathbf{S}_{1,3}^n)+h(\mathbf{Y}_2^n|\mathbf{S}_{3,1}^n,\mathbf{S}_{1,\{3,2\}}^n)+h(\mathbf{S}_{1,\{3,2\}}^n)-h(\mathbf{S}_{3,\{2,1\}}^n)
-h(\mathbf{Z}_2^n)-h(\mathbf{Z}_{1}^n)
\end{eqnarray*}
Thus, from Fano's inequality, we have
\begin{eqnarray}
&&n(R_1+2R_2+R_3-\epsilon_n)\notag\\
&\leq&I(x_1^n;\mathbf{Y}_1^n)+2I(x_2^n;\mathbf{Y}_2^n)+I(x_3^n;\mathbf{Y}_3^n)\notag\\
&\leq&
h(\mathbf{Y}_1^n|\mathbf{S}_{3,1}^n)+h(\mathbf{Y}_2^n|\mathbf{S}_{1,3}^n,\mathbf{S}_{3,\{1,2\}}^n)+h(\mathbf{Y}_3^n|\mathbf{S}_{1,3}^n)+h(\mathbf{Y}_2^n|\mathbf{S}_{3,1}^n,\mathbf{S}_{1,\{3,2\}}^n)-h(\mathbf{Z}_{1}^n)-2h(\mathbf{Z}_2^n)-h(\mathbf{Z}_{3}^n)\notag\\
&\leq&n\big(h(\mathbf{Y}_1^*|\mathbf{S}_{3,1}^*)+h(\mathbf{Y}_2^*|\mathbf{S}_{1,3}^*,\mathbf{S}_{3,\{1,2\}}^*)+h(\mathbf{Y}_3^*|\mathbf{S}_{1,3}^*)+h(\mathbf{Y}_2^*|\mathbf{S}_{3,1}^*,\mathbf{S}_{1,\{3,2\}}^*)-h(\mathbf{Z}_{1})-2h(\mathbf{Z}_2)-h(\mathbf{Z}_{3})\big)\label{newouterbound3user}
\end{eqnarray}
where $\mathbf{Y}^*$ and $\mathbf{S}^*$ are corresponding signal
when $x_i \sim \mathcal{CN}(0,P_i)$. This follows from Lemma 1 in
\cite{Sreekanth_Veeravalli}.
\end{proof}
Also, let us see intuitively why this bound is good for $0\leq
\alpha\leq \frac{N+1}{2N+1}$. Recall that in this regime, at
Receiver 1 the MAC constraint that is active for the common messages
is the sum rate of all common messages from interfering
transmitters, i.e., $R_{2c}+\cdots+R_{(N+1)c}$. Notice that
$\mathbf{Y}_1^*|\mathbf{S}_{N+1,1}^*$ roughly contains common
information from Transmitter 2 to $N+1$ and private message from
Transmitter 1, since $\mathbf{S}_{N+1,1}^*$ roughly contains
Transmitter 1's common message. Thus, we can think that
$R_{1p}+R_{2c}+\cdots+R_{(N+1)c} \leq
h(\mathbf{Y}_1^*|\mathbf{S}_{N+1,1}^*)-h(\mathbf{Z}_1)$. On the
other hand, $\mathbf{Y}_i^*|\mathbf{S}_{N+1,\mathcal{A}\backslash
\{N+1,2,\ldots, i\}}^*,\mathbf{S}_{1,\{N+1,2,\ldots, i\}}^*$ roughly
contains the private messages of Transmitter $i~\forall i=2,\cdots,
N$, since $\mathbf{S}_{N+1,\mathcal{A}\backslash \{N+1,2,\ldots,
i\}}^*\text{and} ~\mathbf{S}_{1,\{N+1,2,\ldots, i\}}^*$ roughly
contain common information of all transmitters. Thus, we can think
that $R_{2p}+\cdots+R_{Np}\leq
\sum_{i=2}^{N}(h(\mathbf{Y}_i^*|\mathbf{S}_{N+1,\mathcal{A}\backslash
\{N+1,2,\ldots, i\}}^*,\mathbf{S}_{1,\{N+1,2,\ldots,
i\}}^*)-h(\mathbf{Z}_i))$. Adding this to the previous one, we have
\begin{eqnarray*}
R_{1p}+R_2+\cdots+R_N+R_{(N+1)c}\leq
h(\mathbf{Y}_1^*|\mathbf{S}_{N+1,1}^*)-h(\mathbf{Z}_1)+\sum_{i=2}^{N}(h(\mathbf{Y}_i^*|\mathbf{S}_{N+1,\mathcal{A}\backslash
\{N+1,2,\ldots, i\}}^*,\mathbf{S}_{1,\{N+1,2,\ldots,
i\}}^*)-h(\mathbf{Z}_i))
\end{eqnarray*}
Similarly, we have
\begin{eqnarray*}
R_{1c}+R_2+\cdots+R_N+R_{(N+1)p}\leq
h(\mathbf{Y}_{N+1}^*|\mathbf{S}_{1,N+1}^*)-h(\mathbf{Z}_{N+1})+\sum_{i=2}^{N}(h(\mathbf{Y}_i^*|\mathbf{S}_{1,\mathcal{A}\backslash
\{1,2,\ldots, i\}}^*,\mathbf{S}_{N+1,\{1,2,\ldots,
i\}}^*)-h(\mathbf{Z}_i))
\end{eqnarray*}
Adding these two bounds, we have the outer bound for
$R_1+2(R_2+\cdots+R_N)+R_{N+1}$.

\begin{lemma}\label{Kusernewgdof}
For $0\leq \alpha\leq \frac{N+1}{2N+1}$, the symmetric generalized
degrees of freedom of the $N+1$ user $1 \times N$ SIMO Gaussian
interference channel are bounded above as
\begin{eqnarray*}
d_{\text{sym}}(\alpha)&\leq& \max\{1-\frac{\alpha}{N},
\frac{N-1}{N}+\frac{\alpha}{N}\}
\end{eqnarray*}
\end{lemma}
\begin{proof}
See the Appendix.
\end{proof}
As shown in Fig. \ref{fig:outerbound}, this bound gives a tight
outer bound for the first ``V'' part of the W curve.

\section{The Symmetric Capacity within $\mathcal{O}(1)$}
\subsection{The $\mathcal{O}(1)$ Capacity Approximation}
\begin{theorem}
For the $N+1$ user symmetric SIMO Gaussian interference channel with
$N$ antennas at each receiver, the symmetric capacity is
approximated within $\mathcal{O}(1)$ as
\begin{eqnarray*}
C_{\text{sym}}\approx\left\{\begin{array}{cc} \log\text{SNR} & \log{\text{INR}} < 0\\
\log\text{SNR}-\frac{1}{N}\log\text{INR} & 0 < \log{\text{INR}}<\frac{1}{2}\log{\text{SNR}}\\
\frac{N-1}{N}\log\text{SNR}+\frac{1}{N}\log\text{INR}& \frac{1}{2}\log{\text{SNR}} \leq \log{\text{INR}} \leq \frac{N+1}{2N+1}\log\text{SNR}\\
\log\text{SNR}-\frac{1}{N+1}\log\text{INR} & \frac{N+1}{2N+1}\log{\text{SNR}} \leq \log{\text{INR}} \leq \log{\text{SNR}} \\
\frac{N}{N+1}\log\text{INR}&  \log{\text{SNR}} \leq \log{\text{INR}} \leq \frac{N+1}{N}\log{\text{SNR}}\\
\log{\text{SNR}}& \log{\text{INR}} \geq
\frac{N+1}{N}\log{\text{SNR}}
\end{array}\right.
\end{eqnarray*}
\end{theorem}
\begin{proof}
If $\log{\text{INR}} < 0$, by treating interference as noise, it can
be easily seen that the symmetric capacity can be approximated as
$\log\text{SNR}+\mathcal{O}(1)$. If $\log{\text{INR}} > 0$, since
both the outer bounds and inner bounds for GDOF are $\mathcal{O}(1)$
characterizations, it directly leads to the $\mathcal{O}(1)$
approximation of the capacity, i.e.
$C_{\text{sym}}=d_{\text{sym}}(\alpha)\log\text{SNR}+\mathcal{O}(1)$.
%For $\text{INR}\leq 1$, it can be easily seen that by treating
%interference as noise, the rate $\log\text{SNR}+\mathcal{O}(1)$ is
%achievable. Also, this rate is the upper bound since it is the
%interference free performance. Thus, for $\log\text{INR}\leq 0$,
%$C_{\text{sym}}=\log\text{SNR}+\mathcal{O}(1)$.
\end{proof}
The $\mathcal{O}(1)$ approximation provides a capacity approximation
whose gap to the accurate capacity is a constant. This constant is
{\em independent} of SNR and INR. However, the gap depends on other
channel parameters. In this case, the gap depends on the
correlations between channel vectors at each receiver:
\begin{eqnarray}
c_{jk,ji}=\frac{\mathbf{H}_{jk}^{\dagger}\mathbf{H}_{ji}}{\|\mathbf{H}_{jk}\|\|\mathbf{H}_{ji}\|}~\forall
j,k,i=1,\cdots,N+1 ~~ k\neq i
\end{eqnarray}
In the following part of this section, we would like to explore
further how the gap depends on the channel parameters.

\subsection{Gap between the inner bound and outer bound}
For simplicity, we consider the case when $N=2$, i.e., 3 user
symmetric
 SIMO interference channel with 2 antennas at each receiver. In
addition to the assumptions we make for the channel model in Section
\ref{section:model} , we further assume
\begin{eqnarray}
c=\mathbf{H}_{12}^{\dagger}\mathbf{H}_{13}=\mathbf{H}_{21}^{\dagger}\mathbf{H}_{23}=\mathbf{H}_{32}^{\dagger}\mathbf{H}_{31}\label{c1}\\
c_1=\mathbf{H}_{12}^{\dagger}\mathbf{H}_{11}=\mathbf{H}_{21}^{\dagger}\mathbf{H}_{22}=\mathbf{H}_{32}^{\dagger}\mathbf{H}_{33}\label{c2}\\
c_2=\mathbf{H}_{13}^{\dagger}\mathbf{H}_{11}=\mathbf{H}_{23}^{\dagger}\mathbf{H}_{22}=\mathbf{H}_{31}^{\dagger}\mathbf{H}_{33}\label{c3}
\end{eqnarray}
This assumption means that the relative orientations of the desired
signal and interference vectors are identical at each receiver. The
channel is unaffected by relabeling the users. Now the capacity
depends on parameters $\text{SNR}$, $\text{INR}$, $c,c_1$ and $c_2$.
\begin{theorem}\label{thm:gap}
For the 3 user symmetric Gaussian interference channel with 2
antennas at each receiver defined above, the achievable scheme
proposed in Section \ref{section:gdof} achieves the symmetric
capacity within $\max\{3, \frac{8}{3}-\frac{1}{3}\log(1-|c|^2)\}$
bits/channel use.
\end{theorem}
This result can be obtained by directly calculating the gap between
the inner bound and the outer bound. The proof is provided in the
Appendix. Note that the gap only depends on $c$ which is the
correlation between two interfering vectors. If $|c|$ is small, then
the gap is small. But if $|c|$ is large, the gap is large. In fact,
this indicates that Han-Kobayashi type scheme with Gaussian
codebooks is not good enough and interference alignment may be
needed when $|c|$ is large. Similar observation is made by Wang and
Tse \cite{Wang_Tse} for the three-to-one Gaussian interference
channel where the interfered receiver has two antennas.

\section{Capacity region for the strong interference regime}
In this section, we derive an outer bound for the 3 user Gaussian
interference channel (not necessarily symmetric) with 2 antennas at
each receiver. This outer bound directly leads to the capacity
region of this channel if the channel vectors satisfy certain
constraints. The channel's input-output relationship is described as
\begin{equation}
\mathbf{Y}_k= \sum_{j=1}^3 \mathbf{H}_{kj}x_j+\mathbf{Z}_k
~~~~~~\forall k\in\{1,2,3\}
\end{equation}
where $\mathbf{Y}_{k}$ is the $2 \times 1$ received signal vector at
receiver $k$, $\mathbf{H}_{kj}$ is the $2 \times 1$  channel vector
from transmitter $j$ to receiver $k$, $x_j$ is the input signal
which satisfies the average power constraint
$\mathbf{E}[|x_j|^2]\leq P_j$ and $\mathbf{Z}_k$ is the additive
circularly symmetric white complex Gaussian noise  vector with zero
mean and identity covariance matrix, i.e., $\mathbf{Z}_k\sim
\mathcal{CN}(\mathbf{0},\mathbf{I})$. Without loss of generality, we
assume the norm of each direct channel vector is equal to 1, i.e.,
$\|\mathbf{H}_{jj}\|=1, \forall j=1,2,3$.

\subsection{Outer bound on the Capacity Region}
Let us define the correlation coefficients between two channel
vectors as
\begin{eqnarray*}
c_{12,13}=\frac{<\mathbf{H}_{12},\mathbf{H}_{13}>}{\|\mathbf{H}_{12}\|\|\mathbf{H}_{13}\|}\\
c_{21,23}=\frac{<\mathbf{H}_{21},\mathbf{H}_{23}>}{\|\mathbf{H}_{21}\|\|\mathbf{H}_{23}\|}\\
c_{31,32}=\frac{<\mathbf{H}_{31},\mathbf{H}_{32}>}{\|\mathbf{H}_{31}\|\|\mathbf{H}_{32}\|}
\end{eqnarray*}
where $<\cdot,\cdot>$ is the inner product of two vectors. The outer
bound of the capacity region of the 3 user SIMO Gaussian
interference channel is presented in the following theorem.

\begin{theorem}\label{thm:strongregion}
The capacity region of the 3 user SIMO Gaussian interference channel
where each receiver has two antennas is bounded above by the
intersection of the capacity regions of the 3 multiple access
channels (MACs), one for each receiver with the additive noise
$\mathbf{Z}_k \sim \mathcal{CN}(\mathbf{0},\mathbf{I})$ modified to
$\mathbf{Z'}_k \sim \mathcal{CN}(\mathbf{0},a^2_k\mathbf{I})
~\forall k=1,2,3$ where
\begin{eqnarray*}
a_1 &=&
\min\big(1,\sqrt{1-|c_{12,13}|^2}\max(\|\mathbf{H}_{12}\|,\|\mathbf{H}_{13}\|),\min(\|\mathbf{H}_{12}\|,\|\mathbf{H}_{13}\|)\big)\\
a_2 &=&
\min\big(1,\sqrt{1-|c_{21,23}|^2}\max(\|\mathbf{H}_{21}\|,\|\mathbf{H}_{23}\|),\min(\|\mathbf{H}_{21}\|,\|\mathbf{H}_{23}\|)\big)\\
a_3 &=&
\min\big(1,\sqrt{1-|c_{31,32}|^2}\max(\|\mathbf{H}_{31}\|,\|\mathbf{H}_{32}\|),\min(\|\mathbf{H}_{31}\|,\|\mathbf{H}_{32}\|)\big)
\end{eqnarray*}
\end{theorem}
\begin{proof}
\begin{figure*}
\centerline{\subfigure[]{\includegraphics[width=2.6in]{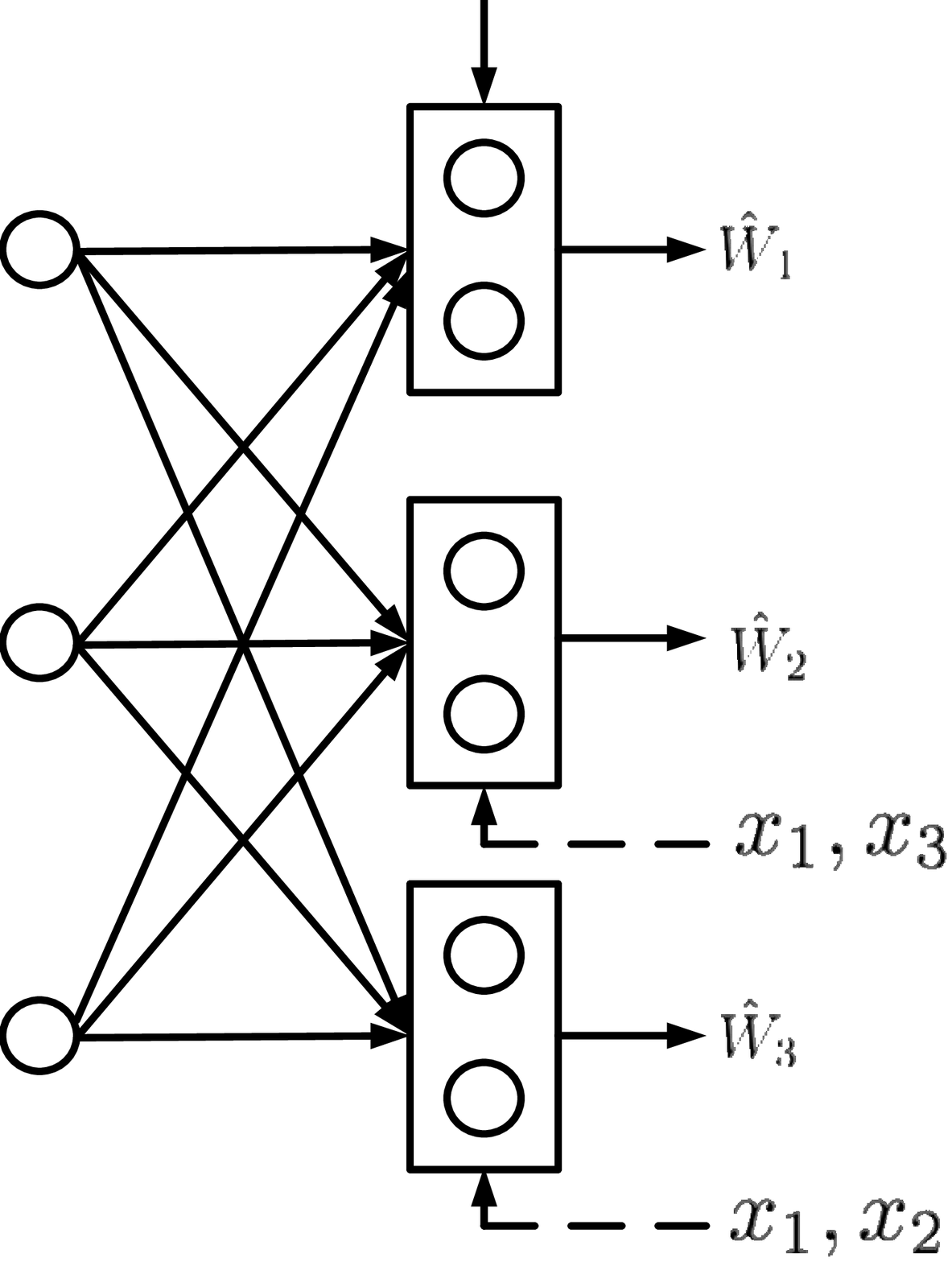}
\label{genie}} \hfil
\subfigure[]{\includegraphics[width=2.6in]{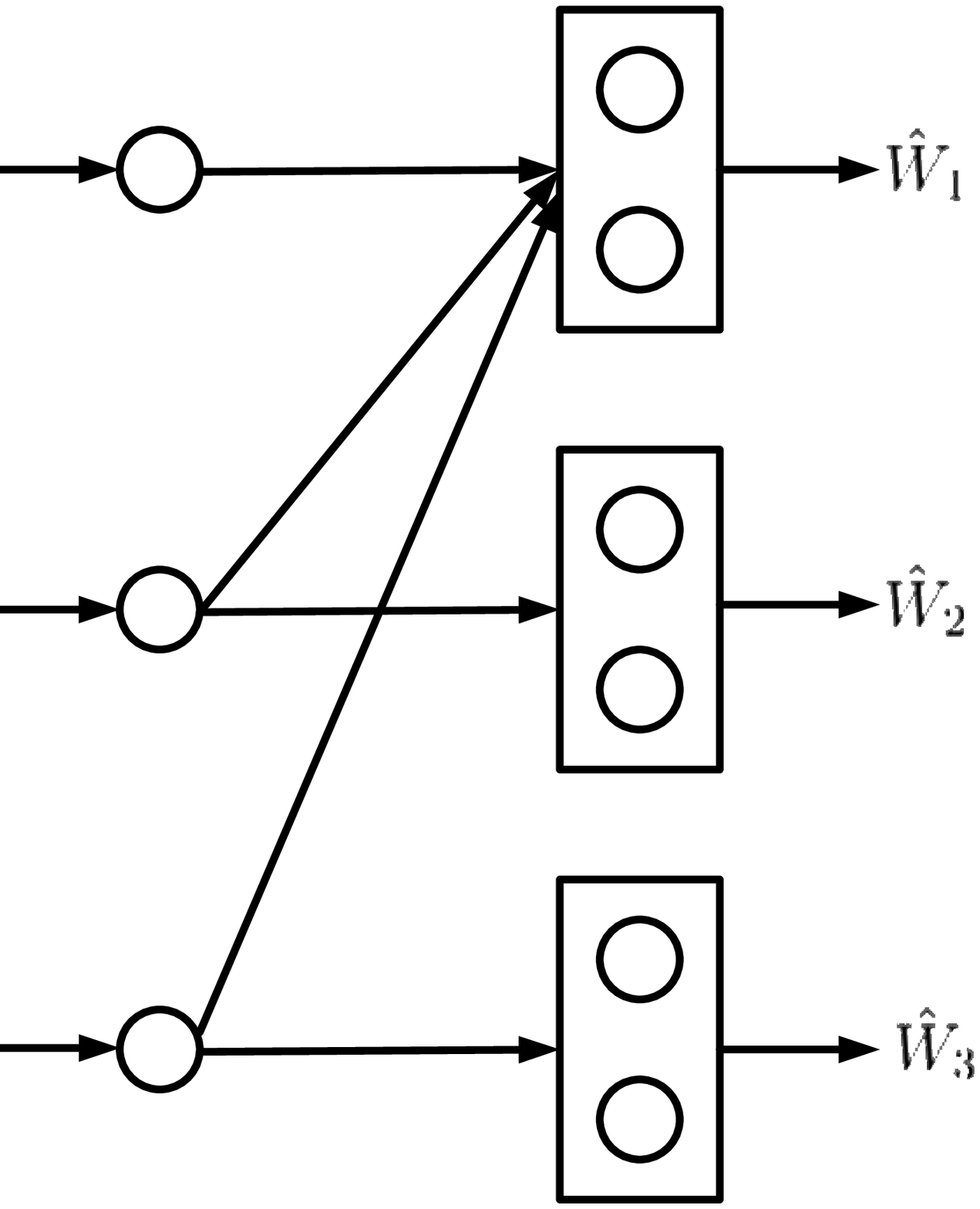}\label{genieaided}
}} \caption{Genie-aided channel with noise reduction}
\end{figure*}
Let a genie provide Receiver 2 with the side information containing
the entire codewords $x_1$ and $x_3$. Then Receiver 2 can simply
subtract out the interference from Transmitter 1 and 3 from its
received signal. Similarly, $x_1$ and $x_2$ are given to Receiver 3
by a genie, so Receiver 3 can subtract out the interference from its
received signal. Then we replace the original noise $\mathbf{Z}_1$
at Receiver 1 with $\mathbf{Z'}_1$. These operations are summarized
in Fig. \ref{genie}. We end up with a many-to-one interference
channel as shown in Fig. \ref{genieaided}. It is obvious that the
capacity region of this channel is an outer bound of the capacity
region of the original interference channel. Now we will argue that
on the genie-aided channel with noise reduction, Receiver 1 is able
to decode all messages, hence giving the MAC bound. Consider any
rate point in the capacity region of the genie-aided channel, so
that each receiver can decode its message reliably. Receiver 1 can
subtract $x_1$ from its own received signal after decoding it. The
resulting channel is given by
\begin{eqnarray}
\mathbf{Y'}_1 &=& \mathbf{H}_{12}x_2+\mathbf{H}_{13}x_3+\mathbf{Z}'_1\\
\mathbf{Y'}_2 &=& \mathbf{H}_{22}x_2+\mathbf{Z}_2\label{r2} \\
\mathbf{Y'}_3 &=& \mathbf{H}_{33}x_3+\mathbf{Z}_3\label{r3}
\end{eqnarray}
Without loss of generality, we assume that $\|\mathbf{H}_{12}\| \geq
\|\mathbf{H}_{13}\|$. Now we argue that Receiver 1 is able to decode
$x_2$ by zero forcing $x_3$. It projects the received signal onto
the direction that is orthogonal to $\mathbf{H}_{13}$. The resulting
signal is
\begin{equation}
y_1=\sqrt{1-|c_{12,13}|^2}\|\mathbf{H}_{12}\|x_2+z_1 \label{r1de2}
\end{equation}
where $z_1 \sim \mathcal{CN}(0,a_1^2)$. \eqref{r1de2} is equivalent
to
\begin{equation}
y'_1=x_2+z'_1 \label{r12}
\end{equation}
where $y'_1= \frac{y_1}{\sqrt{1-|c_{12,13}|^2}\|\mathbf{H}_{12}\|}$
and $z'_1=\frac{z_1}{\sqrt{1-|c_{12,13}|^2}\|\mathbf{H}_{12}\|} \sim
\mathcal{CN}(0,\frac{a^2_1}{(1-|c_{12,13}|^2)\|\mathbf{H}_{12}\|^2})$.
Note that at Receiver 2 the received signal \eqref{r2} is equivalent
to
\begin{equation}
y_2=x_2+z_2 \label{rsiso2}
\end{equation}
where $z_2 \sim \mathcal{CN}(0,1)$. Since
$\frac{a_1}{\sqrt{1-|c_{12,13}|^2}\|\mathbf{H}_{12}\|} \leq 1$,
\eqref{rsiso2} is more noisy than \eqref{r12}. This implies that
Receiver 1 is able to decode $x_2$ since Receiver 2 can decode
$x_2$. After decoding $x_2$, Receiver 1 can subtract it resulting in
a clean channel from Transmitter 3:
\begin{eqnarray}
y''_1=\|\mathbf{H}_{13}\|x_3+z''_1\label{rx1decodex3}
\end{eqnarray}
where $z''_1 \sim \mathcal{CN}(0,a_1^2)$. \eqref{rx1decodex3} is
equivalent to
\begin{eqnarray}
\frac{y''_1}{\|\mathbf{H}_{13}\|}=x_3+\frac{z''_1}{\|\mathbf{H}_{13}\|}\label{r13}
\end{eqnarray}
where $\frac{z''_1}{\|\mathbf{H}_{13}\|} \sim
\mathcal{CN}(0,\frac{a^2_1}{\|\mathbf{H}_{13}\|^2})$. Note that
\eqref{r3} is equivalent to
\begin{equation}
y_3=x_3+z_3 \label{rsiso3}
\end{equation}
where $z_3 \sim \mathcal{CN}(0,1)$. Since
$\frac{a^2_1}{\|\mathbf{H}_{13}\|^2}\leq 1$, \eqref{rsiso3} is more
noisy than \eqref{r13}. This implies that Receiver 1 is able to
decode $x_3$ since Receiver 3 can decode $x_3$. As a result,
Receiver 1 is able to decode both $x_2$ and $x_3$, hence giving the
MAC bound.

By the same arguments, Receiver 2 and Receiver 3 can decode all
messages if the noise at Receiver 2 and Receiver 3 are modified to
$\mathbf{Z}_2'$ and $\mathbf{Z}_3'$, respectively. Therefore, the
capacity region of the original interference channel is bounded
above by the intersection of the capacity regions of the 3 MACs.
\end{proof}

A direct application of this outer bound is to establish the
capacity region of the 3 user SIMO interference channel in strong
interference regime, where each receiver can decode all messages.
\begin{corollary}\label{col:region}
The capacity region of the 3 user SIMO Gaussian interference channel
is the intersection of the MAC capacity regions of each receiver, if
the following conditions are true:
\begin{eqnarray}
  (1-|c_{12,13}|^2)\geq
\min(\frac{1}{\|\mathbf{H}_{12}\|^2}, \frac{1}{\|\mathbf{H}_{13}\|^2}), ~~\min(\|\mathbf{H}_{12}\|,\|\mathbf{H}_{13}\|) \geq 1\label{c1}\\
 (1-|c_{21,23}|^2)\geq
\min(\frac{1}{\|\mathbf{H}_{21}\|^2}, \frac{1}{\|\mathbf{H}_{23}\|^2}), ~~\min(\|\mathbf{H}_{21}\|,\|\mathbf{H}_{23}\|) \geq 1\label{c2}\\
(1-|c_{31,32}|^2)\geq \min(\frac{1}{\|\mathbf{H}_{31}\|^2},
\frac{1}{\|\mathbf{H}_{32}\|^2}),~~
\min(\|\mathbf{H}_{31}\|,\|\mathbf{H}_{32}\|) \geq 1\label{c3}
\end{eqnarray}
\end{corollary}
\begin{proof}
If \eqref{c1}, \eqref{c2}, \eqref{c3} are satisfied, then
$a_1=a_2=a_3=1$ in Theorem \ref{thm:strongregion}. The outer bound
is the intersection of the MAC capacity regions of each receiver.
Since this region is achieved by decoding all messages at each
receiver, this region is the capacity region of the 3 user SIMO
interference channel.
\end{proof}
{\it Remark}: The capacity region in Corollary \ref{col:region} is
also the capacity region of the 3 user SIMO Gaussian multicast
channel, where each transmitter has a common message for all
receivers.

\section{Conclusion}
We characterize the capacity of a class of symmetric SIMO Gaussian
interference channels within $\mathcal{O}(1)$. To get this result,
we first generalize the El Gamal and Costa deterministic
interference channel model to study a three user interference
channel with multiple antenna nodes and find the capacity region of
the proposed deterministic interference channel. Based on the
insights provided by the deterministic channel, we characterize the
generalized degrees of freedom of the $N+1$ user symmetric SIMO
Gaussian interference channels with $N$ antennas at each receiver,
which leads to the $\mathcal{O}(1)$ capacity characterization.

This work follows the idea of successive approximations of capacity
of interference networks that have emerged throughout recent
research on interference channel. Starting from degrees of freedom
which is obtained in \cite{Gou_Jafar_dofmimo}, we solve the GDOF
problem and find the $\mathcal{O}(1)$ characterization and the exact
capacity region for the strong interference regime. While we
consider only the symmetric case in this paper, we believe the key
ideas that emerge from this study can be used to solve the
asymmetric case (as in the 2 user interference channel) as well  in
a similar manner. We suspect the extension will be extremely
cumbersome due to the explosive growth in the number of parameters,
but may be a useful exercise for a system designer to develop
detailed insights into the problem of interference management for
interference networks.

Finally, the SIMO setting is especially of interest for commonly
occurring asymmetric communication scenarios where one end of the
communication link, e.g. the base station is equipped with more
antennas than the other, e.g. the user terminals. This work may
provide useful insights for interference management schemes for such
systems from an information theoretic perspective.

\appendices
\section{Analysis of error probability for the 3 user deterministic interference channel}

We only consider Receiver 1. The same analysis is applied to
Receiver 2 and 3. Due to the symmetry of the code generation, the
average probability of error averaged over all codes does not depend
on the particular index that was sent. Thus, without loss of
generality, we assume the message indexed by (1,1) is sent at
Receiver 1, 2 and 3, i.e., $(j_1,l_1)=(j_2,l_2)=(j_3,l_3)=(1,1)$.

An error occurs if the correct codewords are not jointly typical
with the received sequence, i.e.,
\begin{eqnarray*}
(X_1^n(1,1), V_1^n(1), V_2^n(1),V_3^n(1),Y_1^n) \notin
A_n^{\epsilon}.
\end{eqnarray*}
Also, we have an error if the incorrect codewords from Transmitter 1
are jointly typical with the received sequence, i.e.,
$(X_1^n(\hat{j}_1,\hat{l}_1), V_1^n(\hat{j}_1),
V_2^n(\hat{j}_2),V_3^n(\hat{j}_3),Y_1^n)\in A_n^{\epsilon}$ if
$(\hat{j}_1,\hat{l}_1) \neq (1,1)$. Note that Receiver 1 is only
interested in Transmitter 1's message, so there is no error if the
message from Transmitter 1 can be decoded correctly even the common
messages from Transmitter 2 and 3 are decoded wrongly. Therefore, no
error is declared if $(X_1^n(1,1), V_1^n(1),
V_2^n(\hat{j}_2),V_3^n(\hat{j}_3),Y_1^n)$ are jointly typical even
$(\hat{j}_2,\hat{j}_3) \neq (1,1)$. Define the events
\begin{eqnarray*}
E_{j_1l_1j_2j_3}=\{(X_1^n(j_1,l_1), V_1^n(j_1),
V_2^n(j_2),V_3^n(j_3),Y_1^n)\in A_n^{\epsilon}\}
\end{eqnarray*}
Then the probability of error is
\begin{eqnarray*}
P_e&=&P\big(E_{1111}^c \bigcup \cup_{(j_1,l_1)\neq
(1,1)}E_{j_1l_1j_2j_3}\big)\\
&<&P(E_{1111}^c)+\sum_{j_1\neq 1, l_1 \neq 1, j_2 \neq 1, j_3 \neq
1}P(E_{j_1l_1j_2j_3})+\sum_{j_1\neq 1, l_1 \neq 1, j_2 \neq 1, j_3 =
1}P(E_{j_1l_1j_21})+\sum_{j_1\neq 1, l_1 \neq 1, j_2 = 1, j_3 \neq
1}P(E_{j_1l_11j_3})\\
&&+\sum_{j_1\neq 1, l_1 \neq 1, j_2 =1, j_3= 1}P(E_{j_1 l_1 1
1})+\sum_{j_1\neq 1, l_1 = 1, j_2 \neq 1, j_3 \neq 1}P(E_{j_1 1 j_2
j_3})+\sum_{j_1\neq 1, l_1 = 1, j_2 \neq 1, j_3= 1}P(E_{j_1 1 j_2
1})\\
&&+\sum_{j_1\neq 1, l_1 = 1, j_2 =1, j_3\neq 1}P(E_{j_1 1 1
j_3})+\sum_{j_1\neq 1, l_1 = 1, j_2 = 1, j_3= 1}P(E_{j_1 1 1
1})+\sum_{j_1= 1, l_1 \neq 1, j_2 \neq 1, j_3 \neq 1}P(E_{1
l_1j_2j_3})\\
&&+\sum_{j_1= 1, l_1 \neq 1, j_2 \neq 1, j_3 = 1}P(E_{1 l_1 j_2
1})+\sum_{j_1= 1, l_1 \neq 1, j_2 = 1, j_3 \neq 1}P(E_{1 l_1 1
j_3})+\sum_{j_1= 1, l_1 \neq 1, j_2=1, j_3 = 1}P(E_{1 l_1 1 1})
\end{eqnarray*}
The first term is $\epsilon$ as $n \rightarrow \infty$. All other
terms go to zero as $n \rightarrow \infty$, if the following
constraints are satisfied:
\begin{eqnarray*}
R_{1c}+R_{1p}+R_{2c}+R_{3c} &\leq& I(X_1,V_2,V_3;Y_1)=H(Y_1)\\
R_{1c}+R_{1p}+R_{2c} &\leq& I(X_1,V_2;Y_1|V_3)=H(Y_1|V_3)\\
R_{1c}+R_{1p}+R_{3c} &\leq& I(X_1,V_3;Y_1|V_2)=H(Y_1|V_2)\\
R_{1c}+R_{1p}&\leq& I(X_1;Y_1|V_2V_3)=H(Y_1|V_2V_3)\\
R_{1c}+R_{2c}+R_{3c} &\leq& I(X_1,V_2,V_3;Y_1)=H(Y_1)\\
R_{1c}+R_{2c} &\leq& I(X_1,V_2;Y_1|V_3)=H(Y_1|V_3)\\
R_{1c}+R_{3c} &\leq& I(X_1,V_3;Y_1|V_2)=H(Y_1|V_2)\\
R_{1c} &\leq& I(X_1;Y_1|V_2V_3)=H(Y_1|V_2V_3)\\
R_{1p}+R_{2c}+R_{3c} &\leq& I(X_1,V_2,V_3;Y_1|V_1)=H(Y_1|V_1)\\
R_{1p}+R_{2c} &\leq& I(X_1,V_2;Y_1|V_1V_3)=H(Y_1|V_1V_3)\\
R_{1p}+R_{3c} &\leq& I(X_1,V_3;Y_1|V_1V_2)=H(Y_1|V_1V_2)\\
R_{1p}&\leq& I(X_1;Y_1|V_1V_2V_3)=H(Y_1|V_1V_2V_3)\\
R_{ip},R_{ic} &\ge& 0 ~\forall i=1,2,3
\end{eqnarray*}
Reducing the redundant ones, we have
\begin{eqnarray*}
R_{1p}+R_{1c}+R_{2c}+R_{3c} &\leq& H(Y_1)\\
R_{1c}+R_{1p}+R_{2c} &\leq& H(Y_1|V_3)\\
R_{1c}+R_{1p}+R_{3c} &\leq& H(Y_1|V_2)\\
R_{1p}+R_{2c}+R_{3c} &\leq&H(Y_1|V_1)\\
R_{1p}+R_{2c} &\leq& H(Y_1|V_1V_3)\\
R_{1p}+R_{3c} &\leq& H(Y_1|V_1V_2)\\
R_{1c}+R_{1p}&\leq& H(Y_1|V_2V_3)\\
R_{1p}&\leq& H(Y_1|V_1V_2V_3)\\
R_{ip},R_{ic} &\ge& 0 ~\forall i=1,2,3
\end{eqnarray*}
Receiver 2 (3) has similar constraints by swapping the indices 1 and
2 (3). All these conditions specify the achievable rate region for
rate vector $(R_{1p}, R_{2p}, R_{3p}, R_{1c}, R_{2c}, R_{3c})$.

\section{Proof for Lemma \ref{lemma:Kusernewouterbound}}
\begin{proof}
Let $\mathbf{S}_{j,\mathcal{B}}=\sum_{i \in \mathcal{B}}
\mathbf{H}_{ji}x_i+\mathbf{Z}_j$ where $\mathcal{B}$ is a set of
transmitters. And let $\mathcal{A}$ denote the set of all
transmitters, i.e. $\mathcal{A}=\{1,2,\cdots, K\}$. The notation
$\mathcal{A}\backslash \mathcal{B}$ means the complement of
$\mathcal{B}$ in $\mathcal{A}$. First, let us consider
$I(x_1^n;\mathbf{Y}_1^n)+I(x_2^n;\mathbf{Y}_2^n)$+$\cdots$+$I(x_{K-1}^n;\mathbf{Y}_{K-1}^n)$.
By providing side information to receivers, we have
\begin{eqnarray}
I(x_1^n;\mathbf{Y}_1^n)
&\leq& I(x_1^n;\mathbf{Y}_1^n,\mathbf{S}_{K,1}^n)\notag\\
&=&I(x_1^n;\mathbf{S}_{K,1}^n)+I(x_1^n;\mathbf{Y}_1^n|\mathbf{S}_{K,1}^n)\notag\\
&=&h(\mathbf{S}_{K,1}^n)-h(\mathbf{S}_{K,1}^n|x_1^n)+h(\mathbf{Y}_1^n|\mathbf{S}_{K,1}^n)-h(\mathbf{Y}_1^n|\mathbf{S}_{K,1}^n,x_1^n)\notag\\
&=&h(\mathbf{S}_{K,1}^n)-h(\mathbf{Z}_{K}^n)+h(\mathbf{Y}_1^n|\mathbf{S}_{K,1}^n)-h(\mathbf{S}_{1,\{2,3,\cdots,
K\}}^n)\label{e1} \end{eqnarray} For any $i=2,\ldots, K-1$,
\begin{eqnarray}
I(x_i^n;\mathbf{Y}_i^n)
&\leq&I(x_i^n;\mathbf{Y}_i^n,\mathbf{S}_{K,\{1,2,\ldots, i\}}^n,\mathbf{S}_{1,\mathcal{A}\backslash \{1,2,\ldots, i\}}^n)\notag\\
&=& I(x_i^n;\mathbf{S}_{K,\{1,2,\ldots, i\}}^n)+I(x_i^n;\mathbf{Y}_i^n,\mathbf{S}_{1,\mathcal{A}\backslash \{1,2,\ldots, i\}}^n|\mathbf{S}_{K,\{1,2,\ldots, i\}}^n)\notag\\
&=&h(\mathbf{S}_{K,\{1,2,\ldots, i\}}^n)-h(\mathbf{S}_{K,\{1,2,\ldots, i\}}^n|x_i^n)+I(x_i^n;\mathbf{S}_{1,\mathcal{A}\backslash \{1,2,\ldots, i\}}^n|\mathbf{S}_{K,\{1,2,\ldots, i\}}^n)\notag\\&&+I(x_i^n;\mathbf{Y}_i^n|\mathbf{S}_{1,\mathcal{A}\backslash \{1,2,\ldots, i\}}^n,\mathbf{S}_{K,\{1,2,\ldots, i\}}^n)\notag\\
&=&h(\mathbf{S}_{K,\{1,2,\ldots, i\}}^n)-h(\mathbf{S}_{K,\{1,2,\ldots, i-1\}}^n)+I(x_i^n;\mathbf{Y}_i^n|\mathbf{S}_{1,\mathcal{A}\backslash \{1,2,\ldots, i\}}^n,\mathbf{S}_{K,\{1,2,\ldots, i\}}^n)\notag\\
&=&h(\mathbf{S}_{K,\{1,2,\ldots, i\}}^n)-h(\mathbf{S}_{K,\{1,2,\ldots, i-1\}}^n)+h(\mathbf{Y}_i^n|\mathbf{S}_{1,\mathcal{A}\backslash \{1,2,\ldots, i\}}^n,\mathbf{S}_{K,\{1,2,\ldots, i\}}^n)\notag\\&&-h(\mathbf{Y}_i^n|\mathbf{S}_{1,\mathcal{A}\backslash \{1,2,\ldots, i\}}^n,\mathbf{S}_{K,\{1,2,\ldots, i\}}^n,x_i^n)\notag\\
&\leq&h(\mathbf{S}_{K,\{1,2,\ldots, i\}}^n)-h(\mathbf{S}_{K,\{1,2,\ldots, i-1\}}^n)+h(\mathbf{Y}_i^n|\mathbf{S}_{1,\mathcal{A}\backslash \{1,2,\ldots, i\}}^n,\mathbf{S}_{K,\{1,2,\ldots, i\}}^n)\notag\\&&-h(\mathbf{Y}_i^n|\mathbf{S}_{1,\mathcal{A}\backslash \{1,2,\ldots, i\}}^n,\mathbf{S}_{K,\{1,2,\ldots, i\}}^n,x_\mathcal{A}^n)\notag\\
&=&h(\mathbf{S}_{K,\{1,2,\ldots,
i\}}^n)-h(\mathbf{S}_{K,\{1,2,\ldots,
i-1\}}^n)+h(\mathbf{Y}_i^n|\mathbf{S}_{1,\mathcal{A}\backslash
\{1,2,\ldots, i\}}^n,\mathbf{S}_{K,\{1,2,\ldots,
i\}}^n)-h(\mathbf{Z}_i^n) \label{e2}
\end{eqnarray}
Therefore,
\begin{eqnarray*}
&&I(x_1^n;\mathbf{Y}_1^n)+I(x_2^n;\mathbf{Y}_2^n)+\cdots+I(x_{K-1}^n;\mathbf{Y}_{K-1}^n) \\
&\leq&
h(\mathbf{S}_{K,1}^n)-h(\mathbf{Z}_{K}^n)+h(\mathbf{Y}_1^n|\mathbf{S}_{K,1}^n)-h(\mathbf{S}_{1,\{2,3,\cdots,
K\}}^n)\\&&+ \sum_{i=2}^{K-1}\big(h(\mathbf{S}_{K,\{1,2,\ldots,
i\}}^n)-h(\mathbf{S}_{K,\{1,2,\ldots,
i-1\}}^n)+h(\mathbf{Y}_i^n|\mathbf{S}_{1,\mathcal{A}\backslash
\{1,2,\ldots, i\}}^n,\mathbf{S}_{K,\{1,2,\ldots, i\}}^n)-h(\mathbf{Z}_i^n)\big)\\
&=&h(\mathbf{Y}_1^n|\mathbf{S}_{K,1}^n)-h(\mathbf{S}_{1,\{2,3,\cdots,
K\}}^n)+\sum_{i=2}^{K-1}h(\mathbf{Y}_i^n|\mathbf{S}_{1,\mathcal{A}\backslash
\{1,2,\ldots, i\}}^n,\mathbf{S}_{K,\{1,2,\ldots,
i\}}^n)-\sum_{i=2}^{K-1}h(\mathbf{Z}_i^n)-h(\mathbf{Z}_{K}^n)\\
&&+h(\mathbf{S}_{K,1}^n)+\sum_{i=2}^{K-1}\big(h(\mathbf{S}_{K,\{1,2,\ldots,
i\}}^n)-h(\mathbf{S}_{K,\{1,2,\ldots, i-1\}}^n)\big)\\
&=&
h(\mathbf{Y}_1^n|\mathbf{S}_{K,1}^n)+\sum_{i=2}^{K-1}h(\mathbf{Y}_i^n|\mathbf{S}_{1,\mathcal{A}\backslash
\{1,2,\ldots, i\}}^n,\mathbf{S}_{K,\{1,2,\ldots,
i\}}^n)+h(\mathbf{S}_{K,\{1,2,\ldots,
K-1\}}^n)-h(\mathbf{S}_{1,\{2,3,\cdots,
K\}}^n)-\sum_{i=2}^{K}h(\mathbf{Z}_i^n)
\end{eqnarray*}
Similarly,
\begin{eqnarray*}
&&I(x_K^n;\mathbf{Y}_K^n)+I(x_2^n;\mathbf{Y}_2^n)+\cdots+I(x_{K-1}^n;\mathbf{Y}_{K-1}^n) \\
&\leq&h(\mathbf{Y}_K^n|\mathbf{S}_{1,K}^n)+\sum_{i=2}^{K-1}h(\mathbf{Y}_i^n|\mathbf{S}_{K,\mathcal{A}\backslash
\{K,2,\ldots, i\}}^n,\mathbf{S}_{1,\{K,2,\ldots,
i\}}^n)+h(\mathbf{S}_{1,\{2,\ldots, K-1,
K\}}^n)-h(\mathbf{S}_{K,\{1,2,3,\cdots,
K-1\}}^n)-\sum_{i=1}^{K-1}h(\mathbf{Z}_i^n)
\end{eqnarray*}
Hence,
\begin{eqnarray*}
&&I(x_1^n;\mathbf{Y}_1^n)+2(I(x_2^n;\mathbf{Y}_2^n)+\cdots+I(x_{K-1}^n;\mathbf{Y}_{K-1}^n))+I(x_K^n;\mathbf{Y}_K^n)\\
&\leq&
h(\mathbf{Y}_1^n|\mathbf{S}_{K,1}^n)+h(\mathbf{Y}_K^n|\mathbf{S}_{1,K}^n)+\sum_{i=2}^{K-1}h(\mathbf{Y}_i^n|\mathbf{S}_{K,\mathcal{A}\backslash
\{K,2,\ldots, i\}}^n,\mathbf{S}_{1,\{K,2,\ldots,
i\}}^n)+\sum_{i=2}^{K-1}h(\mathbf{Y}_i^n|\mathbf{S}_{1,\mathcal{A}\backslash
\{1,2,\ldots,
i\}}^n,\mathbf{S}_{K,\{1,2,\ldots, i\}}^n)\\
&&-\sum_{i=1}^{K-1}h(\mathbf{Z}_i^n)-\sum_{i=2}^{K}h(\mathbf{Z}_i^n)\\
&\leq&n\big(h(\mathbf{Y}_1^*|\mathbf{S}_{K,1}^*)+h(\mathbf{Y}_K^*|\mathbf{S}_{1,K}^*)+\sum_{i=2}^{K-1}h(\mathbf{Y}_i^*|\mathbf{S}_{K,\mathcal{A}\backslash
\{K,2,\ldots, i\}}^*,\mathbf{S}_{1,\{K,2,\ldots,
i\}}^*)+\sum_{i=2}^{K-1}h(\mathbf{Y}_i^*|\mathbf{S}_{1,\mathcal{A}\backslash
\{1,2,\ldots, i\}}^*,\mathbf{S}_{K,\{1,2,\ldots, i\}}^*)
\\&&-\sum_{i=1}^{K-1}h(\mathbf{Z}_i)-\sum_{i=2}^{K}h(\mathbf{Z}_i)\big)
\end{eqnarray*}
where $\mathbf{Y}^*$ and $\mathbf{S}^*$ are the corresponding
signals when $x_i \sim \mathcal{CN}(0,P_i)$. This follows from Lemma
1 in \cite{Sreekanth_Veeravalli}. Therefore, from Fano's inequality,
we have
\begin{eqnarray*}
&&R_1+2(R_2+\cdots+R_{K-1})+R_K-\epsilon_n \\
&\leq& \frac{1}{n}(I(x_1^n;\mathbf{Y}_1^n)+2(I(x_2^n;\mathbf{Y}_2^n)+\cdots+I(x_{K-1}^n;\mathbf{Y}_{K-1}^n))+I(x_K^n;\mathbf{Y}_K^n))\\
&\leq&
h(\mathbf{Y}_1^*|\mathbf{S}_{K,1}^*)+h(\mathbf{Y}_K^*|\mathbf{S}_{1,K}^*)+\sum_{i=2}^{K-1}\big(h(\mathbf{Y}_i^*|\mathbf{S}_{K,\mathcal{A}\backslash
\{K,2,\ldots, i\}}^*,\mathbf{S}_{1,\{K,2,\ldots,
i\}}^*)+h(\mathbf{Y}_i^*|\mathbf{S}_{1,\mathcal{A}\backslash
\{1,2,\ldots, i\}}^*,\mathbf{S}_{K,\{1,2,\ldots, i\}}^*)
\big)\\&&-2\sum_{i=2}^{K-1}h(\mathbf{Z}_i)-h(\mathbf{Z}_1)-h(\mathbf{Z}_K)
\end{eqnarray*}

Now we calculate each term in the above expression. First consider
$h(\mathbf{Y}_1^*|\mathbf{S}_{K,1}^*)$. Let
$\Sigma_{\mathbf{Y}_1^*|\mathbf{S}_{K,1}^*}$ be the covariance
matrix of $\mathbf{Y}_1^*|\mathbf{S}_{K,1}^*$. Then
\begin{equation*}
h(\mathbf{Y}_1^*|\mathbf{S}_{K,1}^*) =\log|\pi e
\Sigma_{\mathbf{Y}_1^*|\mathbf{S}_{K,1}^*}|
\end{equation*}
where
\begin{eqnarray}
\Sigma_{\mathbf{Y}_1^*|\mathbf{S}_{K,1}^*}&=&E[\mathbf{Y}_1^*
\mathbf{Y}_1^{* \dagger}]-E[\mathbf{Y}_1^* \mathbf{S}_{K,1}^{*
\dagger}]E[\mathbf{S}_{K,1}^* \mathbf{S}_{K,1}^{
*\dagger}]^{-1}E[\mathbf{S}_{K,1}^*\mathbf{Y}_1^{* \dagger}]\\
&=&\mathbf{I}+
\mathbf{H}_{11}P_1\mathbf{H}_{11}^{\dagger}+\sum_{i=2}^{K}\mathbf{H}_{1i}P_i\mathbf{H}_{1i}^{\dagger}-\mathbf{H}_{11}P_1\mathbf{H}_{K1}^{\dagger}(\mathbf{I}+\mathbf{H}_{K1}P_1\mathbf{H}_{K1}^{\dagger})^{-1}\mathbf{H}_{K1}P_1\mathbf{H}_{11}^{\dagger}\\
&=&\mathbf{I}+\sum_{i=2}^{K}\mathbf{H}_{1i}P_i\mathbf{H}_{1i}^{\dagger}+
\mathbf{H}_{11}(P_1-P_1\mathbf{H}_{K1}^{\dagger}(\mathbf{I}+\mathbf{H}_{K1}P_1\mathbf{H}_{K1}^{\dagger})^{-1}\mathbf{H}_{K1}P_1)\mathbf{H}_{11}^{\dagger}\\
&\stackrel{(a)}{=}&\mathbf{I}+\sum_{i=2}^{K}\mathbf{H}_{1i}P_i\mathbf{H}_{1i}^{\dagger}+
\mathbf{H}_{11}(\frac{1}{P_1}+\|\mathbf{H}_{K1}\|^2)^{-1}\mathbf{H}_{11}^{\dagger}
\end{eqnarray}
where (a) follows from Woodbury matrix identity \cite{woodbury}.
Therefore,
\begin{eqnarray}
&&h(\mathbf{Y}_1^*|\mathbf{S}_{K,1}^*)-h(\mathbf{Z}_1)\\
&=&\log|\pi e
(\mathbf{I}+\sum_{i=2}^{K}\mathbf{H}_{1i}P_i\mathbf{H}_{1i}^{\dagger}+
\mathbf{H}_{11}(\frac{1}{P_1}+\|\mathbf{H}_{K1}\|^2)^{-1}\mathbf{H}_{11}^{\dagger})|-N\log(\pi
e )\\
&=&\log|\mathbf{I}+\sum_{i=2}^{K}\mathbf{H}_{1i}P_i\mathbf{H}_{1i}^{\dagger}+
\frac{P_1}{1+P_1\|\mathbf{H}_{K1}\|^2}\mathbf{H}_{11}\mathbf{H}_{11}^{\dagger}|
\end{eqnarray}
Similarly,
\begin{eqnarray}
h(\mathbf{Y}_K^*|\mathbf{S}_{1,K}^*)-h(\mathbf{Z}_K)=\log|\mathbf{I}+\sum_{i=1}^{K-1}\mathbf{H}_{Ki}P_i\mathbf{H}_{Ki}^{\dagger}+
\frac{P_K}{1+P_K\|\mathbf{H}_{1K}\|^2}\mathbf{H}_{KK}\mathbf{H}_{KK}^{\dagger}|
\end{eqnarray}
Next, consider $h(\mathbf{Y}_i^*|\mathbf{S}_{1,\mathcal{A}\backslash
\{1,2,\ldots, i\}}^*,\mathbf{S}_{K,\{1,2,\ldots, i\}}^*)$. We have
\begin{eqnarray}
\mathbf{Y}_i&=&\underline{\mathbf{H}}_{i1}\underline{\mathbf{X}}_1+\underline{\mathbf{H}}_{i2}\underline{\mathbf{X}}_2+\mathbf{\mathbf{Z}}_i\\
\mathbf{S}_{K,\{1,2,\ldots, i\}}^*&=& \underline{\mathbf{H}}_{K1}\underline{\mathbf{X}}_1+\mathbf{\mathbf{Z}}_{K}\\
\mathbf{S}_{1,\mathcal{A}\backslash \{1,2,\ldots,
i\}}^*&=&\underline{\mathbf{H}}_{12}\underline{\mathbf{X}}_2+\mathbf{\mathbf{Z}}_1
\end{eqnarray}
where
\begin{eqnarray*}
\underline{\mathbf{X}}_1&=&[x_{1}~x_{2} \cdots x_i]^T\\
\underline{\mathbf{X}}_2&=&[x_{i+1}~x_{i+2} \cdots x_K]^T\\
\underline{\mathbf{H}}_{i1}&=&[\mathbf{H}_{i1}~\mathbf{H}_{i2} \cdots \mathbf{H}_{ii}]\\
\underline{\mathbf{H}}_{i2}&=&[\mathbf{H}_{ii+1}~\mathbf{H}_{ii+2} \cdots \mathbf{H}_{iK}]\\
\underline{\mathbf{H}}_{K1}&=&[\mathbf{H}_{K1}~\mathbf{H}_{K2} \cdots \mathbf{H}_{Ki}]\\
\underline{\mathbf{H}}_{12}&=&[\mathbf{H}_{1i+1}~\mathbf{H}_{1i+2}
\cdots \mathbf{H}_{1K}]
\end{eqnarray*}
Hence,
\begin{eqnarray}
h(\mathbf{Y}_i^*|\mathbf{S}_{K,\{1,2,\ldots,
i\}}^*,\mathbf{S}_{1,\mathcal{A}\backslash \{1,2,\ldots,
i\}}^*)=h(\underline{\mathbf{H}}_{i1}\underline{\mathbf{X}}_1+\underline{\mathbf{H}}_{i2}\underline{\mathbf{X}}_2+\mathbf{\mathbf{Z}}_i|\underline{\mathbf{H}}_{K1}\underline{\mathbf{X}}_1+\mathbf{\mathbf{Z}}_{K},
\underline{\mathbf{H}}_{12}\underline{\mathbf{X}}_2+\mathbf{\mathbf{Z}}_{1}
)
\end{eqnarray}
Let $\Sigma$ be the covariance matrix of
$\mathbf{Y}_i^*|\mathbf{S}_{K,\{1,2,\ldots, i\}}^*,
\mathbf{S}_{1,\mathcal{A}\backslash \{1,2,\ldots, i\}}^*$. We have
\begin{eqnarray}
\Sigma
&=&\mathbf{I}+\underline{\mathbf{H}}_{i1}\mathbf{\underline{P}}_{i1}\underline{\mathbf{H}}_{i1}^{\dagger}+\underline{\mathbf{H}}_{i2}\mathbf{\underline{P}}_{i2}\underline{\mathbf{H}}_{i2}^{\dagger}\\&&-\left[\underline{\mathbf{H}}_{i1}\mathbf{\underline{P}}_{i1}\underline{\mathbf{H}}_{K1}^{\dagger}~\underline{\mathbf{H}}_{i2}\mathbf{\underline{P}}_{i2}\underline{\mathbf{H}}_{12}^{\dagger}
\right]\left[\begin{array}{ccc}
\mathbf{I}+\underline{\mathbf{H}}_{K1}\mathbf{\underline{P}}_{i1}\underline{\mathbf{H}}_{K1}^{\dagger}&~& \mathbf{0}\\
\mathbf{0} &~&
\mathbf{I}+\underline{\mathbf{H}}_{12}\mathbf{\underline{P}}_{i2}\underline{\mathbf{H}}_{12}^{\dagger}
\end{array}\right]^{-1}\left[\begin{array}{c}\underline{\mathbf{H}}_{K1}\mathbf{\underline{P}}_{i1}\underline{\mathbf{H}}_{i1}^{\dagger}\\ \underline{\mathbf{H}}_{12}\mathbf{\underline{P}}_{i2}\underline{\mathbf{H}}_{i2}^{\dagger}\end{array}\right]\\
&=&
\mathbf{I}+\underline{\mathbf{H}}_{i1}\mathbf{\underline{P}}_{i1}\underline{\mathbf{H}}_{i1}^{\dagger}-\underline{\mathbf{H}}_{i1}\mathbf{\underline{P}}_{i1}\underline{\mathbf{H}}_{K1}^{\dagger}(\mathbf{I}+\underline{\mathbf{H}}_{K1}\mathbf{\underline{P}}_{i1}\underline{\mathbf{H}}_{K1}^{\dagger})^{-1}\underline{\mathbf{H}}_{K1}\mathbf{\underline{P}}_{i1}\underline{\mathbf{H}}_{i1}^{\dagger}\\&&
+\underline{\mathbf{H}}_{i2}\mathbf{\underline{P}}_{i2}\underline{\mathbf{H}}_{i2}^{\dagger}-\underline{\mathbf{H}}_{i2}\mathbf{\underline{P}}_{i2}\underline{\mathbf{H}}_{12}^{\dagger}(\mathbf{I}+\underline{\mathbf{H}}_{12}\mathbf{\underline{P}}_{i2}\underline{\mathbf{H}}_{12}^{\dagger})^{-1}\underline{\mathbf{H}}_{12}\mathbf{\underline{P}}_{i2}\underline{\mathbf{H}}_{i2}^{\dagger}\\
&=&
\mathbf{I}+\underline{\mathbf{H}}_{i1}(\mathbf{\underline{P}}_{i1}-\mathbf{\underline{P}}_{i1}\underline{\mathbf{H}}_{K1}^{\dagger}(\mathbf{I}+\underline{\mathbf{H}}_{K1}\mathbf{\underline{P}}_{i1}\underline{\mathbf{H}}_{K1}^{\dagger})^{-1}\underline{\mathbf{H}}_{K1}\mathbf{\underline{P}}_{i1})\underline{\mathbf{H}}_{i1}^{\dagger}\\&&
+\underline{\mathbf{H}}_{i2}(\mathbf{\underline{P}}_{i2}-\mathbf{\underline{P}}_{i2}\underline{\mathbf{H}}_{12}^{\dagger}(\mathbf{I}+\underline{\mathbf{H}}_{12}\mathbf{\underline{P}}_{i2}\underline{\mathbf{H}}_{12}^{\dagger})^{-1}\underline{\mathbf{H}}_{12}\mathbf{\underline{P}}_{i2})\underline{\mathbf{H}}_{i2}^{\dagger}\\
&\stackrel{(a)}{=}&
\mathbf{I}+\underline{\mathbf{H}}_{i1}(\mathbf{\underline{P}}_{i1}^{-1}+\underline{\mathbf{H}}_{K1}^{\dagger}\underline{\mathbf{H}}_{K1})^{-1}\underline{\mathbf{H}}_{i1}^{\dagger}+\underline{\mathbf{H}}_{i2}(\mathbf{\underline{P}}_{i2}^{-1}+\underline{\mathbf{H}}_{12}^{\dagger}\underline{\mathbf{H}}_{12})^{-1}\underline{\mathbf{H}}_{i2}^{\dagger}
\end{eqnarray}
where (a) follows from Woodbury matrix identity \cite{woodbury} and
\begin{eqnarray}
\underline{\mathbf{P}}_{i1}=\left[\begin{array}{ccc}P_1 & \cdots & 0 \\
\vdots & \ddots & \vdots\\ 0 & \cdots & P_i \end{array}\right]~~
\underline{\mathbf{P}}_{i2}=\left[\begin{array}{ccc}P_{i+1} & \cdots & 0 \\
\vdots & \ddots & \vdots\\ 0 & \cdots & P_K \end{array}\right]
\end{eqnarray}
Therefore,
\begin{eqnarray}
&&h(\mathbf{Y}_i^*|\mathbf{S}_{K,\{1,2,\ldots,
i\}}^*,\mathbf{S}_{1,\mathcal{A}\backslash \{1,2,\ldots,
i\}}^*)-h(\mathbf{Z}_i)\\
&=&\log|\pi
e(\mathbf{I}+\underline{\mathbf{H}}_{i1}(\mathbf{\underline{P}}_{i1}^{-1}+\underline{\mathbf{H}}_{K1}^{\dagger}\underline{\mathbf{H}}_{K1})^{-1}\underline{\mathbf{H}}_{i1}^{\dagger}+\underline{\mathbf{H}}_{i2}(\mathbf{\underline{P}}_{i2}^{-1}+\underline{\mathbf{H}}_{12}^{\dagger}\underline{\mathbf{H}}_{12})^{-1}\underline{\mathbf{H}}_{i2}^{\dagger}
)|-N\log(\pi e )\\
&=&\log|\mathbf{I}+\underline{\mathbf{H}}_{i1}(\mathbf{\underline{P}}_{i1}^{-1}+\underline{\mathbf{H}}_{K1}^{\dagger}\underline{\mathbf{H}}_{K1})^{-1}\underline{\mathbf{H}}_{i1}^{\dagger}+\underline{\mathbf{H}}_{i2}(\mathbf{\underline{P}}_{i2}^{-1}+\underline{\mathbf{H}}_{12}^{\dagger}\underline{\mathbf{H}}_{12})^{-1}\underline{\mathbf{H}}_{i2}^{\dagger}|
\end{eqnarray}
Similarly, we have
\begin{eqnarray}
&&h(\mathbf{Y}_i^*|\mathbf{S}_{1,\{K,2,\ldots,
i\}}^*,\mathbf{S}_{K,\mathcal{A}\backslash \{K,2,\ldots,
i\}}^*)-h(\mathbf{Z}_i)\\
&=&\log|\mathbf{I}+\underline{\mathbf{H}}_{i3}(\mathbf{\underline{P}}_{i3}^{-1}+\underline{\mathbf{H}}_{13}^{\dagger}\underline{\mathbf{H}}_{13})^{-1}\underline{\mathbf{H}}_{i3}^{\dagger}+\underline{\mathbf{H}}_{i4}(\mathbf{\underline{P}}_{i4}^{-1}+\underline{\mathbf{H}}_{K4}^{\dagger}\underline{\mathbf{H}}_{K4})^{-1}\underline{\mathbf{H}}_{i4}^{\dagger}|
\end{eqnarray}
where
\begin{eqnarray*}
\underline{\mathbf{H}}_{i3}&=&[\mathbf{H}_{iK}~\mathbf{H}_{i2} \cdots \mathbf{H}_{ii}]\\
\underline{\mathbf{H}}_{i4}&=&[\mathbf{H}_{i1}~\mathbf{H}_{ii+1} \cdots \mathbf{H}_{iK-1}]\\
\underline{\mathbf{H}}_{13}&=&[\mathbf{H}_{1K}~\mathbf{H}_{12} \cdots \mathbf{H}_{1i}]\\
\underline{\mathbf{H}}_{K4}&=&[\mathbf{H}_{K1}~\mathbf{H}_{Ki+1} \cdots \mathbf{H}_{KK-1}]\\
\underline{\mathbf{P}}_{i3}&=&\left[\begin{array}{cccc}P_K & 0 & \cdots & 0 \\
0& P_2 & \cdots & 0 \\ \vdots & \vdots & \ddots& \vdots\\ 0 & 0&
\cdots & P_i \end{array}\right]\\
\underline{\mathbf{P}}_{i4}&=&\left[\begin{array}{cccc}P_1 & 0 & \cdots & 0 \\
0& P_{i+1} & \cdots & 0 \\ \vdots & \vdots & \ddots& \vdots\\ 0 & 0&
\cdots & P_{K-1}\end{array}\right]
\end{eqnarray*}
Therefore,
\begin{eqnarray*}
&&R_1+2(R_2+\cdots+R_{K-1})+R_K \notag \\&& \leq
\log|\mathbf{I}+\sum_{i=2}^K P_i
\mathbf{H}_{1i}\mathbf{H}_{1i}^{\dagger}+\frac{P_1}{1+\|\mathbf{H}_{K1}\|^2P_1}\mathbf{H}_{11}\mathbf{H}_{11}^{\dagger}|\\
&&+\log|\mathbf{I}+\sum_{i=1}^{K-1} P_i
\mathbf{H}_{Ki}\mathbf{H}_{Ki}^{\dagger}+\frac{P_K}{1+\|\mathbf{H}_{1K}\|^2P_K}\mathbf{H}_{KK}\mathbf{H}_{KK}^{\dagger}|\\
&&+\sum_{i=2}^{K-1}(\log|\mathbf{I}+\underline{\mathbf{H}}_{i1}(\underline{\mathbf{P}}_{i1}^{-1}+\mathbf{\underline{H}}_{K1}^{\dagger}\mathbf{\underline{H}}_{K1})^{-1}\mathbf{\underline{H}}_{i1}^{\dagger}+\underline{\mathbf{H}}_{i2}(\underline{\mathbf{P}}_{i2}^{-1}+\mathbf{\underline{H}}_{12}^{\dagger}\mathbf{\underline{H}}_{12})^{-1}\mathbf{\underline{H}}_{i2}^{\dagger}|\\
&&+
\log|\mathbf{I}+\underline{\mathbf{H}}_{i3}(\underline{\mathbf{P}}_{i3}^{-1}+\mathbf{\underline{H}}_{13}^{\dagger}\mathbf{\underline{H}}_{13})^{-1}\mathbf{\underline{H}}_{i3}^{\dagger}+\underline{\mathbf{H}}_{i4}(\underline{\mathbf{P}}_{i4}^{-1}+\mathbf{\underline{H}}_{K4}^{\dagger}\mathbf{\underline{H}}_{K4})^{-1}\mathbf{\underline{H}}_{i4}^{\dagger}|)
\end{eqnarray*}
\end{proof}

\section{proof for lemma \ref{Kusernewgdof} }
\begin{proof}
Now we apply the bound in Lemma \ref{lemma:Kusernewouterbound} to
the $N+1$ user $1\times N$ symmetric SIMO interference channel. By
replacing $\mathbf{H}_{ji}$ with
$\sqrt{\rho^{\alpha}}\mathbf{H}_{ji}$ for $i\neq j$,
$\mathbf{H}_{jj}$ with $\sqrt{\rho}\mathbf{H}_{jj}$ and setting
$P_i=1$ in Lemma \ref{lemma:Kusernewouterbound}, we have
\begin{eqnarray}
&&R_1+2(R_2+\cdots+R_{N})+R_{N+1} \notag \\&& \leq
\log|\mathbf{I}+\sum_{i=2}^{N+1} \rho^{\alpha}
\mathbf{H}_{1i}\mathbf{H}_{1i}^{\dagger}+\frac{\rho}{1+\rho^{\alpha}}\mathbf{H}_{11}\mathbf{H}_{11}^{\dagger}|\\
&&+\log|\mathbf{I}+\sum_{i=1}^{N} \rho^{\alpha}
\mathbf{H}_{N+1i}\mathbf{H}_{N+1i}^{\dagger}+\frac{\rho}{1+\rho^{\alpha}}\mathbf{H}_{N+1N+1}\mathbf{H}_{N+1N+1}^{\dagger}|\\
&&+\sum_{i=2}^{N}(\log|\mathbf{I}+\rho^{\alpha}\underline{\mathbf{H}}_{i1}(\mathbf{I}+\rho^{\alpha}\underline{\mathbf{H}}_{N+11}^{\dagger}\underline{\mathbf{H}}_{N+11})^{-1}\underline{\mathbf{H}}_{i1}^{\dagger}+\rho^{\alpha}\underline{\mathbf{H}}_{i2}(\mathbf{I}+\rho^{\alpha}\underline{\mathbf{H}}_{12}^{\dagger}\underline{\mathbf{H}}_{12})^{-1}\underline{\mathbf{H}}_{i2}^{\dagger}|\\
&&+
\log|\mathbf{I}+\rho^{\alpha}\underline{\mathbf{H}}_{i3}(\mathbf{I}+\rho^{\alpha}\underline{\mathbf{H}}_{13}^{\dagger}\underline{\mathbf{H}}_{13})^{-1}\underline{\mathbf{H}}_{i3}^{\dagger}+\rho^{\alpha}\underline{\mathbf{H}}_{i4}(\mathbf{I}+\rho^{\alpha}\underline{\mathbf{H}}_{N+14}^{\dagger}\underline{\mathbf{H}}_{N+14})^{-1}\underline{\mathbf{H}}_{i4}^{\dagger}|)
\end{eqnarray}
Note that here
$\underline{\mathbf{H}}_{i1}=[\mathbf{H}_{i1}~\mathbf{H}_{i2} \cdots
\sqrt{\rho^{1-\alpha}}\mathbf{H}_{ii}]$ and
$\underline{\mathbf{H}}_{i3}=[\mathbf{H}_{iN+1}~\mathbf{H}_{i2}
\cdots \sqrt{\rho^{1-\alpha}}\mathbf{H}_{ii}]$. Now let us calculate
the degrees of freedom of the outer bound. Consider the first term.
\begin{eqnarray}
&&\log|\mathbf{I}+\sum_{i=2}^{N+1} \rho^{\alpha}
\mathbf{H}_{1i}\mathbf{H}_{1i}^{\dagger}+\frac{\rho}{1+\rho^{\alpha}}\mathbf{H}_{11}\mathbf{H}_{11}^{\dagger}|\label{newbound1}\\
&\stackrel{(a)}{=}&(\max\{1-\alpha, \alpha\}+(N-1)\alpha)\log\rho + \mathcal{O}(1)\notag\\
&=&\max\{1+(N-2)\alpha, N\alpha\}\log\rho+ \mathcal{O}(1)
\label{doft1}
\end{eqnarray}
where $(a)$ follows from Lemma \ref{lemma:O(1)1}. Similarly, we have
\begin{eqnarray}
\log|\mathbf{I}+\sum_{i=1}^{N} \rho^{\alpha}
\mathbf{H}_{N+1i}\mathbf{H}_{N+1i}^{\dagger}+\frac{\rho}{1+\rho^{\alpha}}\mathbf{H}_{N+1N+1}\mathbf{H}_{N+1N+1}^{\dagger}|=\max\{1+(N-2)\alpha,
N\alpha\}\log\rho+ \mathcal{O}(1) \label{doft2}
\end{eqnarray}

Next,
\begin{eqnarray}
&&\log|\mathbf{I}+\rho^{\alpha}\underline{\mathbf{H}}_{i1}(\mathbf{I}+\rho^{\alpha}\underline{\mathbf{H}}_{N+11}^{\dagger}\underline{\mathbf{H}}_{N+11})^{-1}\underline{\mathbf{H}}_{i1}^{\dagger}+\rho^{\alpha}\underline{\mathbf{H}}_{i2}(\mathbf{I}+\rho^{\alpha}\underline{\mathbf{H}}_{12}^{\dagger}\underline{\mathbf{H}}_{12})^{-1}\underline{\mathbf{H}}_{i2}^{\dagger}|\label{newouterbound2term}\\
&=&\log|\mathbf{I}+\rho^{\alpha}\underline{\mathbf{H}}_{i1}(\rho^{\alpha}\underline{\mathbf{H}}_{N+11}^{\dagger}\underline{\mathbf{H}}_{N+11})^{-1}\underline{\mathbf{H}}_{i1}^{\dagger}+\rho^{\alpha}\underline{\mathbf{H}}_{i2}(\rho^{\alpha}\underline{\mathbf{H}}_{12}^{\dagger}\underline{\mathbf{H}}_{12})^{-1}\underline{\mathbf{H}}_{i2}^{\dagger}|+\mathcal{O}(1)\\
&=&\log|\mathbf{I}+\underline{\mathbf{H}}_{i1}(\underline{\mathbf{H}}_{N+11}^{\dagger}\underline{\mathbf{H}}_{N+11})^{-1}\underline{\mathbf{H}}_{i1}^{\dagger}+\underline{\mathbf{H}}_{i2}(\underline{\mathbf{H}}_{12}^{\dagger}\underline{\mathbf{H}}_{12})^{-1}\underline{\mathbf{H}}_{i2}^{\dagger}|+\mathcal{O}(1)\\
&\stackrel{(a)}{=}&\log|\mathbf{I}+\underline{\mathbf{H}}_{i1}\underline{\mathbf{H}}_{i1}^{\dagger}|+\mathcal{O}(1)\\
&=&\log|\mathbf{I}+\mathbf{H}_{i1}\mathbf{H}_{i1}^{\dagger}+\cdots+\mathbf{H}_{ii-1}\mathbf{H}_{ii-1}^{\dagger}+\rho^{1-\alpha}\mathbf{H}_{ii}\mathbf{H}_{ii}^{\dagger}|+\mathcal{O}(1)\\
&=&(1-\alpha) \log \rho+\mathcal{O}(1)\label{doft3}
\end{eqnarray}
where ($a$) follows from the fact that
$\underline{\mathbf{H}}_{i2}(\underline{\mathbf{H}}_{12}^{\dagger}\underline{\mathbf{H}}_{12})^{-1}\underline{\mathbf{H}}_{i2}^{\dagger}$
and
$\underline{\mathbf{H}}_{N+11}^{\dagger}\underline{\mathbf{H}}_{N+11}$
are constant. Similarly, we have
\begin{eqnarray}
&&\log|\mathbf{I}+\rho^{\alpha}\underline{\mathbf{H}}_{i3}(\mathbf{I}+\rho^{\alpha}\underline{\mathbf{H}}_{13}^{\dagger}\underline{\mathbf{H}}_{13})^{-1}\underline{\mathbf{H}}_{i3}^{\dagger}+\rho^{\alpha}\underline{\mathbf{H}}_{i4}(\mathbf{I}+\rho^{\alpha}\underline{\mathbf{H}}_{N+14}^{\dagger}\underline{\mathbf{H}}_{N+14})^{-1}\underline{\mathbf{H}}_{i4}^{\dagger}|\\
&=&(1-\alpha) \log \rho+\mathcal{O}(1)\label{doft4}
\end{eqnarray}
From \eqref{doft1}, \eqref{doft2}, \eqref{doft3} and \eqref{doft4},
we have
\begin{eqnarray}
&&R_1+2(R_2+\cdots+R_{N})+R_{N+1}\\
&\leq& (2\max\{1+(N-2)\alpha,
N\alpha\}+2(N-1)(1-\alpha))\log\rho+\mathcal{O}(1)
\end{eqnarray}
Hence,
\begin{eqnarray}
d_{\text{sym}}(\alpha)&\leq& \frac{1}{2N}(2\max\{1+(N-2)\alpha,
N\alpha\}+2(N-1)(1-\alpha))\\
&=&\max\{1-\frac{\alpha}{N},\frac{N-1}{N}+\frac{\alpha}{N}\}
\end{eqnarray}
\end{proof}

\section{Proof for Theorem \ref{thm:gap}}
\begin{proof}
We prove the theorem by calculating the difference between the inner
bound and outer bound in different regimes. First, we state some
equalities and inequalities which will be used repeatedly later.
Define
$\underline{\mathbf{H}}_{12}=[\mathbf{H}_{12}~\mathbf{H}_{13}]$.
Consider
\begin{eqnarray}
&&\log|\mathbf{I}+\rho^{\alpha}\underline{\mathbf{H}}_{12}\underline{\mathbf{H}}_{12}^{\dag}| \label{equality1}\\
&=&\log|\mathbf{I}+\rho^{\alpha}\underline{\mathbf{H}}_{12}^{\dag}\underline{\mathbf{H}}_{12}|\notag\\
&=&\log|\mathbf{I}+\rho^{\alpha}\left[ \begin{array}{cc} 1& c\\ c^*&1 \end{array} \right]|\notag\\
&=&\log(1+2\rho^{\alpha}+\rho^{2\alpha}-\rho^{2\alpha}|c|^2)\\
&=&\log(1+2\rho^{\alpha}+\rho^{2\alpha}|\bar{c}|^2)
\end{eqnarray}
where $|\bar{c}|^2=1-|c|^2$. If $\alpha=0$ in \eqref{equality1},
then
\begin{eqnarray}
\log|\mathbf{I}+\underline{\mathbf{H}}_{12}\underline{\mathbf{H}}_{12}^{\dag}|
= \log(3+|\bar{c}|^2) \leq 2
\end{eqnarray}
We will use inequality:
\begin{eqnarray}
\log|\mathbf{I}+\rho^{\alpha}\underline{\mathbf{H}}_{12}\underline{\mathbf{H}}_{12}^{\dag}|<\log|\mathbf{I}+\mathbf{H}_{11}\mathbf{H}^{\dag}_{11}+\rho^{\alpha}\underline{\mathbf{H}}_{12}\underline{\mathbf{H}}_{12}^{\dag}|<\log|\mathbf{I}+\mathbf{H}_{11}\mathbf{H}^{\dag}_{11}|+\log|\mathbf{I}+\rho^{\alpha}\underline{\mathbf{H}}_{12}\underline{\mathbf{H}}_{12}^{\dag}|
\end{eqnarray}
Notice that this inequality is GDOF tight, i.e., the 3 parts are
equal in terms of GDOF.

We present the outer bounds which will be used later. The first one
is the single user bound:
\begin{eqnarray}
R &\leq& \log(1+\rho)
\end{eqnarray}

The second bound is the two user bound. We apply Lemma
\ref{lemma:mtoo} to a two user $1\times 2$ SIMO interference channel
. It is obvious that this is also an outer bound for the 3 user SIMO
interference channel. Then, we have
\begin{eqnarray}
 R_1+R_2 &\leq&
\log|\mathbf{I}+\rho^{\alpha}\mathbf{H}_{12}\mathbf{H}_{12}^{\dagger}+\rho\mathbf{H}_{11}\mathbf{H}_{11}^{\dagger}|+\log(1+\frac{\rho}{1+\rho^{\alpha}})\\
\Rightarrow R &\leq&
\frac{1}{2}\log|\mathbf{I}+\rho^{\alpha}\mathbf{H}_{12}\mathbf{H}_{12}^{\dagger}+\rho\mathbf{H}_{11}\mathbf{H}_{11}^{\dagger}|+\frac{1}{2}\log(1+\frac{\rho}{1+\rho^{\alpha}})\label{gapouterbound1}
\end{eqnarray}
where $R$ is the symmetric rate. Similarly,
\begin{eqnarray}
R &\leq&
\frac{1}{2}\log|\mathbf{I}+\rho^{\alpha}\mathbf{H}_{13}\mathbf{H}_{13}^{\dagger}+\rho\mathbf{H}_{11}\mathbf{H}_{11}^{\dagger}|+\frac{1}{2}\log(1+\frac{\rho}{1+\rho^{\alpha}})
\end{eqnarray}

The third bound is obtained by applying Lemma
\ref{lemma:Kusernewouterbound} to the 3 user symmetric SIMO
interference channel. Then, we have
\begin{eqnarray*}
R_1+2R_2+R_3&\leq&
\log|\mathbf{I}+\rho^{\alpha}\mathbf{\underline{H}}_{12}\mathbf{\underline{H}}_{12}^{\dagger}+\frac{\rho}{1+\rho^{\alpha}}\mathbf{H}_{11}\mathbf{H}_{11}^{\dagger}|+\log|\mathbf{I}+\rho^{\alpha}\underline{\mathbf{H}}_{32}\underline{\mathbf{H}}_{32}^{\dagger}+\frac{\rho}{1+\rho^{\alpha}}
\mathbf{H}_{33}\mathbf{H}_{33}^{\dagger}|\\
&&+\log|\mathbf{I}+\rho^{\alpha}\underline{\mathbf{H}}_{21}(\mathbf{I}+\rho^{\alpha}\underline{\mathbf{H}}_{31}^{\dagger}\underline{\mathbf{H}}_{31})^{-1}\underline{\mathbf{H}}_{21}^{\dagger}+\frac{\rho^{\alpha}}{1+\rho^{\alpha}}\mathbf{H}_{23}\mathbf{H}_{23}^{\dagger}|\\
&&+\log|\mathbf{I}+\rho^{\alpha}\underline{\mathbf{H}}_{23}(\mathbf{I}+\rho^{\alpha}\underline{\mathbf{H}}_{13}^{\dagger}\underline{\mathbf{H}}_{13})^{-1}\underline{\mathbf{H}}_{23}^{\dagger}+\frac{\rho^{\alpha}}{1+\rho^{\alpha}}\mathbf{H}_{21}\mathbf{H}_{21}^{\dagger}|
\end{eqnarray*}
where
\begin{eqnarray*}
\underline{\mathbf{H}}_{32} &=& [\mathbf{H}_{31}~ \mathbf{H}_{32}]\\
\underline{\mathbf{H}}_{21} &=& [\mathbf{H}_{21}~ \sqrt{\rho^{1-\alpha}}\mathbf{H}_{22}]\\
\underline{\mathbf{H}}_{31} &=& [\mathbf{H}_{31}~ \mathbf{H}_{32}]\\
\underline{\mathbf{H}}_{23} &=& [\mathbf{H}_{23}~ \sqrt{\rho^{1-\alpha}}\mathbf{H}_{22}]\\
\underline{\mathbf{H}}_{13} &=& [\mathbf{H}_{13}~ \mathbf{H}_{12}]
\end{eqnarray*}
We can loosen this bound a little such that it can be used to bound
the gap between the outer bound and inner bound. Consider the third
term in the above equation.
\begin{eqnarray}
&&\log|\mathbf{I}+\rho^{\alpha}\underline{\mathbf{H}}_{21}(\mathbf{I}+\rho^{\alpha}\underline{\mathbf{H}}_{31}^{\dagger}\underline{\mathbf{H}}_{31})^{-1}\underline{\mathbf{H}}_{21}^{\dagger}+\frac{\rho^{\alpha}}{1+\rho^{\alpha}}\mathbf{H}_{23}\mathbf{H}_{23}^{\dagger}|\\
&<&\log|\mathbf{I}+\rho^{\alpha}\underline{\mathbf{H}}_{21}(\mathbf{I}+\rho^{\alpha}\underline{\mathbf{H}}_{31}^{\dagger}\underline{\mathbf{H}}_{31})^{-1}\underline{\mathbf{H}}_{21}^{\dagger}|+\log|\mathbf{I}+\frac{\rho^{\alpha}}{1+\rho^{\alpha}}\mathbf{H}_{23}\mathbf{H}_{23}^{\dagger}|\\
&<&\log|\mathbf{I}+\underline{\mathbf{H}}_{21}\underline{\mathbf{H}}_{21}^{\dagger}|+\log|\mathbf{I}+\frac{\rho^{\alpha}}{1+\rho^{\alpha}}\mathbf{H}_{23}\mathbf{H}_{23}^{\dagger}|\\
&=&\log|\mathbf{I}+\mathbf{H}_{21}\mathbf{H}_{21}^{\dagger}+\rho^{1-\alpha}\mathbf{H}_{22}\mathbf{H}_{22}^{\dagger}|+\log|\mathbf{I}+\frac{\rho^{\alpha}}{1+\rho^{\alpha}}\mathbf{H}_{23}\mathbf{H}_{23}^{\dagger}|\\
&<&\log|\mathbf{I}+\mathbf{H}_{21}\mathbf{H}_{21}^{\dagger}|+\log|\mathbf{I}+\rho^{1-\alpha}\mathbf{H}_{22}\mathbf{H}_{22}^{\dagger}|+\log|\mathbf{I}+\mathbf{H}_{23}\mathbf{H}_{23}^{\dagger}|\\
&=&\log(1+\rho^{1-\alpha})+2
\end{eqnarray}
Similarly,
\begin{eqnarray}
\log|\mathbf{I}+\rho^{\alpha}\underline{\mathbf{H}}_{23}(\mathbf{I}+\rho^{\alpha}\underline{\mathbf{H}}_{13}^{\dagger}\underline{\mathbf{H}}_{13})^{-1}\underline{\mathbf{H}}_{23}^{\dagger}+\frac{\rho^{\alpha}}{1+\rho^{\alpha}}\mathbf{H}_{21}\mathbf{H}_{21}^{\dagger}|<\log(1+\rho^{1-\alpha})+2
\end{eqnarray}
Therefore,
\begin{eqnarray}
&&R_1+2R_2+R_3 \notag\\&\leq&
\log|\mathbf{I}+\rho^{\alpha}\mathbf{\underline{H}}_{12}\mathbf{\underline{H}}_{12}^{\dagger}+\frac{\rho}{1+\rho^{\alpha}}\mathbf{H}_{11}\mathbf{H}_{11}^{\dagger}|+\log|\mathbf{I}+\rho^{\alpha}\underline{\mathbf{H}}_{32}\underline{\mathbf{H}}_{32}^{\dagger}+\frac{\rho}{1+\rho^{\alpha}}
\mathbf{H}_{33}\mathbf{H}_{33}^{\dagger}|+2(2+\log(1+\rho^{1-\alpha}))\notag\\
&<&\log|\mathbf{I}+\rho^{\alpha}\mathbf{\underline{H}}_{12}\mathbf{\underline{H}}_{12}^{\dagger}+\rho^{1-\alpha}\mathbf{H}_{11}\mathbf{H}_{11}^{\dagger}|+\log|\mathbf{I}+\rho^{\alpha}\underline{\mathbf{H}}_{32}\underline{\mathbf{H}}_{32}^{\dagger}+\rho^{1-\alpha}
\mathbf{H}_{33}\mathbf{H}_{33}^{\dagger}|+2(2+\log(1+\rho^{1-\alpha}))\notag\\
&\stackrel{(a)}{=}&2\log|\mathbf{I}+\rho^{\alpha}\mathbf{\underline{H}}_{12}\mathbf{\underline{H}}_{12}^{\dagger}+\rho^{1-\alpha}\mathbf{H}_{11}\mathbf{H}_{11}^{\dagger}|+2(2+\log(1+\rho^{1-\alpha}))\notag\\
\Rightarrow R &\leq&
\frac{1}{2}\log|\mathbf{I}+\rho^{\alpha}\mathbf{\underline{H}}_{12}\mathbf{\underline{H}}_{12}^{\dagger}+\rho^{1-\alpha}\mathbf{H}_{11}\mathbf{H}_{11}^{\dagger}|+\frac{1}{2}(2+\log(1+\rho^{1-\alpha}))\label{gapouterbound2}
\end{eqnarray}
where (a) follows from the assumption about the symmetry of
directions of channel vectors at different receivers.

The last bound is obtained by applying Lemma \ref{lemma:mtoo} to the
3 user symmetric SIMO interference channel. We have
\begin{eqnarray}
&&R_1+R_2+R_3\notag\\
&\leq&\log|\mathbf{I}+\rho^{\alpha}\underline{\mathbf{H}}_{12}\underline{\mathbf{H}}_{12}^{\dag}+
\rho\mathbf{H}_{11}\mathbf{H}_{11}^{\dag}|-\log|\mathbf{I}+\rho^{\alpha}\underline{\mathbf{H}}_{12}\underline{\mathbf{H}}_{12}^{\dag}|+2\log(1+\rho)+\log|\mathbf{I}+\frac{\rho^{\alpha}}{1+\rho}\underline{\mathbf{H}}_{12}\underline{\mathbf{H}}_{12}^{\dag}|\notag\\
&<&\log|\mathbf{I}+\rho^{\alpha}\underline{\mathbf{H}}_{12}\underline{\mathbf{H}}_{12}^{\dag}+
\rho\mathbf{H}_{11}\mathbf{H}_{11}^{\dag}|-
\log|\mathbf{I}+\rho^{\alpha}\underline{\mathbf{H}}_{12}^{\dag}\underline{\mathbf{H}}_{12}|+2\log(1+\rho)+\log|\mathbf{I}+\rho^{\alpha-1}\underline{\mathbf{H}}_{12}^{\dag}\underline{\mathbf{H}}_{12}|\\
\Rightarrow R &\leq&
\frac{1}{3}\log|\mathbf{I}+\rho^{\alpha}\underline{\mathbf{H}}_{12}\underline{\mathbf{H}}_{12}^{\dag}+
\rho\mathbf{H}_{11}\mathbf{H}_{11}^{\dag}|-\frac{1}{3}\log|\mathbf{I}+\rho^{\alpha}\underline{\mathbf{H}}_{12}\underline{\mathbf{H}}_{12}^{\dag}|+\frac{2}{3}\log(1+\rho)+\frac{1}{3}\log|\mathbf{I}+\rho^{\alpha-1}\underline{\mathbf{H}}_{12}\underline{\mathbf{H}}_{12}^{\dag}|\notag\\\label{gapouterbound3}
\end{eqnarray}

Since the achievable schemes are different for $\alpha >1 $ and $0
\leq \alpha \leq 1$, we consider two cases separately.

\subsection{$\alpha >1 $}
In this regime, the achievable scheme is to let each receiver decode
all messages. Thus, the achievable rate region is the intersection
of 3 MAC capacity regions, one at each receiver. Due to symmetry,
consider the MAC at Receiver 1. The rate region is specified by
\begin{eqnarray}
R_{1} &\leq& \log(1+\rho)\\
R_{2} &\leq& \log(1+\rho^{\alpha})\\
R_{3} &\leq& \log(1+\rho^{\alpha})\\
R_{1}+R_{2} &\leq&
\log|\mathbf{I}+\rho\mathbf{H}_{11}\mathbf{H}_{11}^{\dag}+\rho^{\alpha}\mathbf{H}_{12}\mathbf{H}_{12}^{\dagger}|\\
R_{1}+R_{3} &\leq&
\log|\mathbf{I}+\rho\mathbf{H}_{11}\mathbf{H}_{11}^{\dag}+\rho^{\alpha}\mathbf{H}_{13}\mathbf{H}_{13}^{\dagger}|\\
R_{2}+R_{3} &\leq&
\log|\mathbf{I}+\rho^{\alpha}\mathbf{H}_{12}\mathbf{H}^{\dag}_{12}+\rho^{\alpha}\mathbf{H}_{13}\mathbf{H}_{13}^{\dagger}|\\
R_{1}+R_{2}+R_{3} &\leq&
\log|\mathbf{I}+\rho\mathbf{H}_{11}\mathbf{H}_{11}^{\dag}+\rho^{\alpha}\mathbf{H}_{12}\mathbf{H}^{\dag}_{12}+\rho^{\alpha}\mathbf{H}_{13}\mathbf{H}_{13}^{\dagger}|
\end{eqnarray}
The constraints at Receiver 2 (3) can be obtained by swapping the
indices 1 and 2 (3). Due to symmetry of directions of channel
vectors, the achievable symmetric rate is unaffected by swapping
user indices. Thus, the achievable symmetric rate is specified by
the following constraints:
\begin{eqnarray}
R &\leq& \log(1+\rho)\label{stronginner1}\\
R &\leq& \frac{1}{2}
\log|\mathbf{I}+\rho\mathbf{H}_{11}\mathbf{H}_{11}^{\dag}+\rho^{\alpha}\mathbf{H}_{12}\mathbf{H}_{12}^{\dagger}|\label{stronginner2}\\
R &\leq& \frac{1}{2} \log|\mathbf{I}+\rho\mathbf{H}_{11}\mathbf{H}_{11}^{\dag}+\rho^{\alpha}\mathbf{H}_{13}\mathbf{H}_{13}^{\dagger}|\label{stronginner3}\\
R &\leq& \frac{1}{2}
\log|\mathbf{I}+\rho^{\alpha}\mathbf{H}_{12}\mathbf{H}^{\dag}_{12}+\rho^{\alpha}\mathbf{H}_{13}\mathbf{H}_{13}^{\dagger}|\label{stronginner4}\\
R &\leq& \frac{1}{3}
\log|\mathbf{I}+\rho\mathbf{H}_{11}\mathbf{H}_{11}^{\dag}+\rho^{\alpha}\mathbf{H}_{12}\mathbf{H}^{\dag}_{12}+\rho^{\alpha}\mathbf{H}_{13}\mathbf{H}_{13}^{\dagger}|\label{stronginner5}
\end{eqnarray}
Next, we will calculate the gap between each inner bound and its
corresponding outer bound. For \eqref{stronginner1}, from the single
user bound, the gap is 0. For \eqref{stronginner2}, the
corresponding outer bound is \eqref{gapouterbound1}. Calculating the
difference between \eqref{gapouterbound1} and \eqref{stronginner2},
we have the gap
\begin{eqnarray}
\frac{1}{2}\log(1+\frac{\rho}{1+\rho^{\alpha}}) \leq 0.5
\end{eqnarray}
Similarly, the gap for \eqref{stronginner3} is no more than 0.5
bit/channel use. For \eqref{stronginner4}, the corresponding outer
bound is \eqref{gapouterbound2}. Calculating the difference between
\eqref{gapouterbound2} and \eqref{stronginner4}, we have the gap
\begin{eqnarray*}
&&\frac{1}{2}\log|\mathbf{I}+\rho^{\alpha}\mathbf{\underline{H}}_{12}\mathbf{\underline{H}}_{12}^{\dagger}+\rho^{1-\alpha}\mathbf{H}_{11}\mathbf{H}_{11}^{\dagger}|+\frac{1}{2}(2+\log(1+\rho^{1-\alpha}))
-\frac{1}{2}
\log|\mathbf{I}+\rho^{\alpha}\mathbf{H}_{12}\mathbf{H}^{\dag}_{12}+\rho^{\alpha}\mathbf{H}_{13}\mathbf{H}_{13}^{\dagger}|\\
&<&\frac{1}{2}\log|\mathbf{I}+\rho^{\alpha}\mathbf{\underline{H}}_{12}\mathbf{\underline{H}}_{12}^{\dagger}|+\frac{1}{2}\log|\mathbf{I}+\rho^{1-\alpha}\mathbf{H}_{11}\mathbf{H}_{11}^{\dagger}|+\frac{1}{2}(2+\log(1+\rho^{1-\alpha}))
-\frac{1}{2}
\log|\mathbf{I}+\rho^{\alpha}\mathbf{H}_{12}\mathbf{H}^{\dag}_{12}+\rho^{\alpha}\mathbf{H}_{13}\mathbf{H}_{13}^{\dagger}|\\
&=& 1+\log(1+\rho^{1-\alpha}) \\&\leq& 2
\end{eqnarray*}
For \eqref{stronginner5}, the corresponding outer bound is
\eqref{gapouterbound3}. Then the gap is
\begin{eqnarray}
&&\frac{2}{3}\log(1+\rho)+\frac{1}{3}\log|\mathbf{I}+\rho^{\alpha-1}\underline{\mathbf{H}}_{12}\underline{\mathbf{H}}_{12}^{\dag}|-\frac{1}{3}\log|\mathbf{I}+\rho^{\alpha}\underline{\mathbf{H}}_{12}\underline{\mathbf{H}}_{12}^{\dag}|\\
&=&\frac{2}{3}\log(1+\rho)+\frac{1}{3}\log(1+2\rho^{\alpha-1}+\rho^{2\alpha-2}|\bar{c}|^2)-\frac{1}{3}\log(1+2\rho^{\alpha}+\rho^{2\alpha}|\bar{c}|^2)\\
&=&\frac{2}{3}\log(1+\rho)-\frac{2}{3}\log\rho+\frac{1}{3}\log(\rho^2+2\rho^{\alpha+1}+\rho^{2\alpha}|\bar{c}|^2)-\frac{1}{3}\log(1+2\rho^{\alpha}+\rho^{2\alpha}|\bar{c}|^2)\\
&=&\frac{2}{3}(\log(1+\rho)-\log\rho)+\frac{1}{3}\log\frac{\rho^2+2\rho^{\alpha+1}+\rho^{2\alpha}|\bar{c}|^2}{1+2\rho^{\alpha}+\rho^{2\alpha}|\bar{c}|^2}\\
&<&\frac{2}{3}(\log(1+\rho)-\log\rho)+\frac{1}{3}\log\frac{\rho^2+2\rho^{\alpha+1}+\rho^{2\alpha}|\bar{c}|^2}{\rho^{2\alpha}|\bar{c}|^2}\\
&\stackrel{(a)}{<}& \frac{2}{3} +\frac{1}{3}\log\frac{4}{|\bar{c}|^2}\\
&=& \frac{4}{3}-\frac{1}{3}\log(|\bar{c}|^2)\label{gap1}
\end{eqnarray}
where (a) uses the fact that $|\bar{c}|^2\leq 1$ and $\alpha >1$.

\subsection{$ 0 \leq \alpha \leq 1$}
For this weak interference regime, the achievable scheme is the
Han-Kobayashi type scheme. The achievable rate is the sum of the
rate for the common messages and the rate for the private messages.
Due to symmetric assumption, from \eqref{achprivate}, the achievable
rate for the private message is
\begin{eqnarray*}
R_p=\log|\mathbf{I}+\underline{\mathbf{H}}_{12}\underline{\mathbf{H}}_{12}^{\dag}+\rho^{1-\alpha}\mathbf{H}_{11}\mathbf{H}_{11}^{\dag}|-\log|\mathbf{I}+\underline{\mathbf{H}}_{12}\underline{\mathbf{H}}_{12}^{\dag}|
\end{eqnarray*}
where
$\underline{\mathbf{H}}_{12}=[\mathbf{H}_{12}~\mathbf{H}_{13}]$. For
the achievable rate for the common message, it is the intersection
of 3 MAC capacity regions, one at each receiver. The MAC constraints
at Receiver 1 are
\begin{eqnarray}
R_{1c} \leq
\log|\mathbf{I}+\mathbf{H}_{13}\mathbf{H}^{\dag}_{13}+\mathbf{H}_{12}\mathbf{H}_{12}^{\dagger}+\rho\mathbf{H}_{11}\mathbf{H}_{11}^{\dag}|-\log|\mathbf{I}+\underline{\mathbf{H}}_{12}\underline{\mathbf{H}}_{12}^{\dag}+\rho^{1-\alpha}\mathbf{H}_{11}\mathbf{H}_{11}^{\dag}|
\\
R_{2c} \leq
\log|\mathbf{I}+\mathbf{H}_{13}\mathbf{H}^{\dag}_{13}+\rho^{1-\alpha}\mathbf{H}_{11}\mathbf{H}_{11}^{\dag}+\rho^{\alpha}\mathbf{H}_{12}\mathbf{H}_{12}^{\dagger}|-\log|\mathbf{I}+\underline{\mathbf{H}}_{12}\underline{\mathbf{H}}_{12}^{\dag}+\rho^{1-\alpha}\mathbf{H}_{11}\mathbf{H}_{11}^{\dag}|
\\
R_{3c} \leq
\log|\mathbf{I}+\mathbf{H}_{12}\mathbf{H}^{\dag}_{12}+\rho^{1-\alpha}\mathbf{H}_{11}\mathbf{H}_{11}^{\dag}+\rho^{\alpha}\mathbf{H}_{13}\mathbf{H}_{13}^{\dagger}|-\log|\mathbf{I}+\underline{\mathbf{H}}_{12}\underline{\mathbf{H}}_{12}^{\dag}+\rho^{1-\alpha}\mathbf{H}_{11}\mathbf{H}_{11}^{\dag}|
\\
R_{1c}+R_{2c} \leq
\log|\mathbf{I}+\mathbf{H}_{13}\mathbf{H}^{\dag}_{13}+\rho\mathbf{H}_{11}\mathbf{H}_{11}^{\dag}+\rho^{\alpha}\mathbf{H}_{12}\mathbf{H}_{12}^{\dagger}|-\log|\mathbf{I}+\underline{\mathbf{H}}_{12}\underline{\mathbf{H}}_{12}^{\dag}+\rho^{1-\alpha}\mathbf{H}_{11}\mathbf{H}_{11}^{\dag}|
\\
R_{1c}+R_{3c} \leq
\log|\mathbf{I}+\mathbf{H}_{12}\mathbf{H}^{\dag}_{12}+\rho\mathbf{H}_{11}\mathbf{H}_{11}^{\dag}+\rho^{\alpha}\mathbf{H}_{13}\mathbf{H}_{13}^{\dagger}|-\log|\mathbf{I}+\underline{\mathbf{H}}_{12}\underline{\mathbf{H}}_{12}^{\dag}+\rho^{1-\alpha}\mathbf{H}_{11}\mathbf{H}_{11}^{\dag}|
\\
R_{2c}+R_{3c} \leq
\log|\mathbf{I}+\rho^{\alpha}\mathbf{H}_{12}\mathbf{H}^{\dag}_{12}+\rho^{1-\alpha}\mathbf{H}_{11}\mathbf{H}_{11}^{\dag}+\rho^{\alpha}\mathbf{H}_{13}\mathbf{H}_{13}^{\dagger}|-\log|\mathbf{I}+\underline{\mathbf{H}}_{12}\underline{\mathbf{H}}_{12}^{\dag}+\rho^{1-\alpha}\mathbf{H}_{11}\mathbf{H}_{11}^{\dag}|
\\
R_{1c}+R_{2c}+R_{3c} \leq
\log|\mathbf{I}+\rho\mathbf{H}_{11}\mathbf{H}_{11}^{\dag}+\rho^{\alpha}\mathbf{H}_{12}\mathbf{H}^{\dag}_{12}+\rho^{\alpha}\mathbf{H}_{13}\mathbf{H}_{13}^{\dagger}|-\log|\mathbf{I}+\underline{\mathbf{H}}_{12}\underline{\mathbf{H}}_{12}^{\dag}+\rho^{1-\alpha}\mathbf{H}_{11}\mathbf{H}_{11}^{\dag}|
\end{eqnarray}
The constraints at Receiver 2 (3) can be obtained by swapping the
indices 1 and 2 (3). Due to symmetry of directions of channel
vectors, the achievable symmetric rate is unaffected by swapping
user indices. Adding the the private message's rate, we have the
following achievable symmetric rate:
\begin{eqnarray}
R &\leq& \log|\mathbf{I}+\mathbf{H}_{13}\mathbf{H}^{\dag}_{13}+\mathbf{H}_{12}\mathbf{H}_{12}^{\dagger}+\rho\mathbf{H}_{11}\mathbf{H}_{11}^{\dag}|-\log|\mathbf{I}+\underline{\mathbf{H}}_{12}\underline{\mathbf{H}}_{12}^{\dag}|\label{achievablerate1}\\
R &\leq& \log|\mathbf{I}+\mathbf{H}_{13}\mathbf{H}^{\dag}_{13}+\rho^{1-\alpha}\mathbf{H}_{11}\mathbf{H}_{11}^{\dag}+\rho^{\alpha}\mathbf{H}_{12}\mathbf{H}_{12}^{\dagger}|-\log|\mathbf{I}+\underline{\mathbf{H}}_{12}\underline{\mathbf{H}}_{12}^{\dag}|\label{achievablerate2}\\
R &\leq& \log|\mathbf{I}+\mathbf{H}_{12}\mathbf{H}^{\dag}_{12}+\rho^{1-\alpha}\mathbf{H}_{11}\mathbf{H}_{11}^{\dag}+\rho^{\alpha}\mathbf{H}_{13}\mathbf{H}_{13}^{\dagger}|-\log|\mathbf{I}+\underline{\mathbf{H}}_{12}\underline{\mathbf{H}}_{12}^{\dag}|\label{achievablerate3}\\
R &\leq& \frac{1}{2}\log|\mathbf{I}+\mathbf{H}_{13}\mathbf{H}^{\dag}_{13}+\rho\mathbf{H}_{11}\mathbf{H}_{11}^{\dag}+\rho^{\alpha}\mathbf{H}_{12}\mathbf{H}_{12}^{\dagger}|+\frac{1}{2}\log|\mathbf{I}+\underline{\mathbf{H}}_{12}\underline{\mathbf{H}}_{12}^{\dag}+\rho^{1-\alpha}\mathbf{H}_{11}\mathbf{H}_{11}^{\dag}|-\log|\mathbf{I}+\underline{\mathbf{H}}_{12}\underline{\mathbf{H}}_{12}^{\dag}|\notag\\
\label{achievablerate4}\\
R &\leq& \frac{1}{2}\log|\mathbf{I}+\mathbf{H}_{12}\mathbf{H}^{\dag}_{12}+\rho\mathbf{H}_{11}\mathbf{H}_{11}^{\dag}+\rho^{\alpha}\mathbf{H}_{13}\mathbf{H}_{13}^{\dagger}|+\frac{1}{2}\log|\mathbf{I}+\underline{\mathbf{H}}_{12}\underline{\mathbf{H}}_{12}^{\dag}+\rho^{1-\alpha}\mathbf{H}_{11}\mathbf{H}_{11}^{\dag}|-\log|\mathbf{I}+\underline{\mathbf{H}}_{12}\underline{\mathbf{H}}_{12}^{\dag}|\notag\\
\label{achievablerate5}\\
R &\leq& \frac{1}{2}\log|\mathbf{I}+\rho^{\alpha}\mathbf{\underline{H}}_{12}\mathbf{\underline{H}}_{12}^{\dagger}+\rho^{1-\alpha}\mathbf{H}_{11}\mathbf{H}_{11}^{\dagger}|+\frac{1}{2}\log|\mathbf{I}+\mathbf{\underline{H}}_{12}\mathbf{\underline{H}}_{12}^{\dagger}+\rho^{1-\alpha}\mathbf{H}_{11}\mathbf{H}_{11}^{\dagger}|-\log|\mathbf{I}+\mathbf{\underline{H}}_{12}\mathbf{\underline{H}}_{12}^{\dagger}|\label{achievablerate6}\\
R &\leq&
\frac{1}{3}\log|\mathbf{I}+\rho^{\alpha}\mathbf{\underline{H}}_{12}\mathbf{\underline{H}}_{12}^{\dagger}+\rho\mathbf{H}_{11}\mathbf{H}_{11}^{\dagger}|+\frac{2}{3}\log|\mathbf{I}+\mathbf{\underline{H}}_{12}\mathbf{\underline{H}}_{12}^{\dagger}+\rho^{1-\alpha}\mathbf{H}_{11}\mathbf{H}_{11}^{\dagger}|-\log|\mathbf{I}+\mathbf{\underline{H}}_{12}\mathbf{\underline{H}}_{12}^{\dagger}|\label{achievablerate7}
\end{eqnarray}
We will calculate the gap for each case.

For \eqref{achievablerate1}, the corresponding outer bound is the
single user bound. Thus, the gap is
\begin{eqnarray*}
R_{\text{up}}- R_{\text{low}}&=&
\log(1+\rho)-\log|\mathbf{I}+\mathbf{H}_{13}\mathbf{H}^{\dag}_{13}+\mathbf{H}_{12}\mathbf{H}_{12}^{\dagger}+\rho\mathbf{H}_{11}\mathbf{H}_{11}^{\dag}|+\log|\mathbf{I}+\underline{\mathbf{H}}_{12}\underline{\mathbf{H}}_{12}^{\dag}|\\
&<&\log(1+\rho)-\log|\mathbf{I}+\rho\mathbf{H}_{11}\mathbf{H}_{11}^{\dag}|+\log|\mathbf{I}+\underline{\mathbf{H}}_{12}\underline{\mathbf{H}}_{12}^{\dag}|\\
 &\leq& 2
\end{eqnarray*}

For \eqref{achievablerate2}, the achievable rate can be bounded
below as
\begin{eqnarray}
R&=&\log|\mathbf{I}+\mathbf{H}_{13}\mathbf{H}^{\dag}_{13}+\rho^{1-\alpha}\mathbf{H}_{11}\mathbf{H}_{11}^{\dag}+\rho^{\alpha}\mathbf{H}_{12}\mathbf{H}_{12}^{\dagger}|-\log|\mathbf{I}+\underline{\mathbf{H}}_{12}\underline{\mathbf{H}}_{12}^{\dag}|\\
&>&\log|\mathbf{I}+\rho^{1-\alpha}\mathbf{H}_{11}\mathbf{H}_{11}^{\dag}+\rho^{\alpha}\mathbf{H}_{12}\mathbf{H}_{12}^{\dagger}|-\log|\mathbf{I}+\underline{\mathbf{H}}_{12}\underline{\mathbf{H}}_{12}^{\dag}|\label{gapweakachievble2}
\end{eqnarray}
For the outer bound, by giving $x_3$ to Receiver 1 and 2, we
essentially have a two user interference channel. We will use a sum
rate bound for the 2 user MIMO interference channel derived in
\cite{parker:Gdofmimo}. This bound is the generalization of the ETW
bound to multiple antennas case.
\begin{eqnarray*}
R_1+R_2\leq
h(\mathbf{Y}_1|\mathbf{S}_{21},x_3)+h(\mathbf{Y}_2|\mathbf{S}_{12},x_3)-h(\mathbf{Z}_1)-h(\mathbf{Z}_2)
\end{eqnarray*}
where
\begin{eqnarray*}
\mathbf{S}_{21}=\sqrt{\rho^{\alpha}}\mathbf{H}_{21}x_1+\mathbf{Z}_{2}\quad
\mathbf{S}_{12}=\sqrt{\rho^{\alpha}}\mathbf{H}_{12}x_2+\mathbf{Z}_{1}
\end{eqnarray*}
Thus,
\begin{eqnarray}
R_1+R_2&\leq&\log|\mathbf{I}+\rho^{\alpha}\mathbf{H}_{12}\mathbf{H}_{12}^{\dagger}+\frac{\rho}{1+\rho^{\alpha}}\mathbf{H}_{11}\mathbf{H}_{11}^{\dagger}|+\log|\mathbf{I}+\rho^{\alpha}\mathbf{H}_{21}\mathbf{H}_{21}^{\dagger}+\frac{\rho}{1+\rho^{\alpha}}\mathbf{H}_{22}\mathbf{H}_{22}^{\dagger}|\notag\\
&<&\log|\mathbf{I}+\rho^{\alpha}\mathbf{H}_{12}\mathbf{H}_{12}^{\dagger}+\rho^{1-\alpha}\mathbf{H}_{11}\mathbf{H}_{11}^{\dagger}|+\log|\mathbf{I}+\rho^{\alpha}\mathbf{H}_{21}\mathbf{H}_{21}^{\dagger}+\rho^{1-\alpha}\mathbf{H}_{22}\mathbf{H}_{22}^{\dagger}|\\
&\stackrel{(a)}{=}& 2\log|\mathbf{I}+\rho^{\alpha}\mathbf{H}_{12}\mathbf{H}_{12}^{\dagger}+\rho^{1-\alpha}\mathbf{H}_{11}\mathbf{H}_{11}^{\dagger}|\\
\Rightarrow
R&<&\log|\mathbf{I}+\rho^{\alpha}\mathbf{H}_{12}\mathbf{H}_{12}^{\dagger}+\rho^{1-\alpha}\mathbf{H}_{11}\mathbf{H}_{11}^{\dagger}|\label{gapweakouter2}
\end{eqnarray}
where (a) follows from the assumption about the symmetry of
directions of channel vectors. By calculating the difference between
\eqref{gapweakouter2} and \eqref{gapweakachievble2}, the gap is
\begin{eqnarray*}
\log|\mathbf{I}+\underline{\mathbf{H}}_{12}\underline{\mathbf{H}}_{12}^{\dag}|
\leq 2
\end{eqnarray*}
Similarly, for \eqref{achievablerate3} the gap is also 2
bits/channel use.

For \eqref{achievablerate4}, the achievable rate can be bounded
below as
\begin{eqnarray*}
R&=&\frac{1}{2}\log|\mathbf{I}+\mathbf{H}_{13}\mathbf{H}^{\dag}_{13}+\rho\mathbf{H}_{11}\mathbf{H}_{11}^{\dag}+\rho^{\alpha}\mathbf{H}_{12}\mathbf{H}_{12}^{\dagger}|+\frac{1}{2}\log|\mathbf{I}+\underline{\mathbf{H}}_{12}\underline{\mathbf{H}}_{12}^{\dag}+\rho^{1-\alpha}\mathbf{H}_{11}\mathbf{H}_{11}^{\dag}|-\log|\mathbf{I}+\underline{\mathbf{H}}_{12}\underline{\mathbf{H}}_{12}^{\dag}|\\
&>&\frac{1}{2}\log|\mathbf{I}+\rho\mathbf{H}_{11}\mathbf{H}_{11}^{\dag}+\rho^{\alpha}\mathbf{H}_{12}\mathbf{H}_{12}^{\dagger}|+\frac{1}{2}\log|\mathbf{I}+\rho^{1-\alpha}\mathbf{H}_{11}\mathbf{H}_{11}^{\dag}|-\log|\mathbf{I}+\underline{\mathbf{H}}_{12}\underline{\mathbf{H}}_{12}^{\dag}|
\end{eqnarray*}
The outer bound is \eqref{gapouterbound1}. Therefore, the gap is
\begin{eqnarray*}
R_{\text{gap}}&<&\frac{1}{2}\log|\mathbf{I}+\rho^{\alpha}\mathbf{H}_{12}\mathbf{H}_{12}^{\dagger}+\rho\mathbf{H}_{11}\mathbf{H}_{11}^{\dagger}|+\frac{1}{2}\log(1+\frac{\rho}{1+\rho^{\alpha}})
-\frac{1}{2}\log|\mathbf{I}+\rho\mathbf{H}_{11}\mathbf{H}_{11}^{\dag}+\rho^{\alpha}\mathbf{H}_{12}\mathbf{H}_{12}^{\dagger}|\\
&-&\frac{1}{2}\log|\mathbf{I}+\rho^{1-\alpha}\mathbf{H}_{11}\mathbf{H}_{11}^{\dag}|+\log|\mathbf{I}+\underline{\mathbf{H}}_{12}\underline{\mathbf{H}}_{12}^{\dag}|\\
&=&\frac{1}{2}\log(1+\frac{\rho}{1+\rho^{\alpha}})-\frac{1}{2}\log(1+\rho^{1-\alpha})+\log|\mathbf{I}+\underline{\mathbf{H}}_{12}\underline{\mathbf{H}}_{12}^{\dag}|\\
&<& 2
\end{eqnarray*}
Similarly, for \eqref{achievablerate5}, the gap is 2 bits/channel
user.

For \eqref{achievablerate6}, the achievable rate can be bounded
below as
\begin{eqnarray}
R&=&\frac{1}{2}\log|\mathbf{I}+\rho^{\alpha}\mathbf{\underline{H}}_{12}\mathbf{\underline{H}}_{12}^{\dagger}+\rho^{1-\alpha}\mathbf{H}_{11}\mathbf{H}_{11}^{\dagger}|+\frac{1}{2}\log|\mathbf{I}+\mathbf{\underline{H}}_{12}\mathbf{\underline{H}}_{12}^{\dagger}+\rho^{1-\alpha}\mathbf{H}_{11}\mathbf{H}_{11}^{\dagger}|-\log|\mathbf{I}+\mathbf{\underline{H}}_{12}\mathbf{\underline{H}}_{12}^{\dagger}|\notag\\
&>&\frac{1}{2}\log|\mathbf{I}+\rho^{\alpha}\mathbf{\underline{H}}_{12}\mathbf{\underline{H}}_{12}^{\dagger}+\rho^{1-\alpha}\mathbf{H}_{11}\mathbf{H}_{11}^{\dagger}|+\frac{1}{2}\log(1+\rho^{1-\alpha})-\log|\mathbf{I}+\mathbf{\underline{H}}_{12}\mathbf{\underline{H}}_{12}^{\dagger}|
\end{eqnarray}
The outer bound is \eqref{gapouterbound2}. Therefore, the gap is
\begin{eqnarray}\label{gap3}
R_{\text{gap}}&<&\frac{1}{2}\log|\mathbf{I}+\rho^{\alpha}\mathbf{\underline{H}}_{12}\mathbf{\underline{H}}_{12}^{\dagger}+\rho^{1-\alpha}\mathbf{H}_{11}\mathbf{H}_{11}^{\dagger}|+\frac{1}{2}(2+\log(1+\rho^{1-\alpha}))\\
&&-\frac{1}{2}\log|\mathbf{I}+\rho^{\alpha}\mathbf{\underline{H}}_{12}\mathbf{\underline{H}}_{12}^{\dagger}+\rho^{1-\alpha}\mathbf{H}_{11}\mathbf{H}_{11}^{\dagger}|-\frac{1}{2}\log(1+\rho^{1-\alpha})+\log|\mathbf{I}+\mathbf{\underline{H}}_{12}\mathbf{\underline{H}}_{12}^{\dagger}|\\
&=&1+\log|\mathbf{I}+\mathbf{\underline{H}}_{12}\mathbf{\underline{H}}_{12}^{\dagger}|\\
&\leq& 3
\end{eqnarray}
Thus the gap for the symmetric capacity is 3 bits/channel use.

For \eqref{achievablerate7}, the achievable rate can be bounded
below as
\begin{eqnarray}
R &=&\frac{1}{3}\log|\mathbf{I}+\rho^{\alpha}\mathbf{\underline{H}}_{12}\mathbf{\underline{H}}_{12}^{\dagger}+\rho\mathbf{H}_{11}\mathbf{H}_{11}^{\dagger}|+\frac{2}{3}\log|\mathbf{I}+\mathbf{\underline{H}}_{12}\mathbf{\underline{H}}_{12}^{\dagger}+\rho^{1-\alpha}\mathbf{H}_{11}\mathbf{H}_{11}^{\dagger}|-\log|\mathbf{I}+\mathbf{\underline{H}}_{12}\mathbf{\underline{H}}_{12}^{\dagger}|\\
&>&\frac{1}{3}\log|\mathbf{I}+\rho^{\alpha}\mathbf{\underline{H}}_{12}\mathbf{\underline{H}}_{12}^{\dagger}+\rho\mathbf{H}_{11}\mathbf{H}_{11}^{\dagger}|+\frac{2}{3}\log|\mathbf{I}+\rho^{1-\alpha}\mathbf{H}_{11}\mathbf{H}_{11}^{\dagger}|-\log|\mathbf{I}+\mathbf{\underline{H}}_{12}\mathbf{\underline{H}}_{12}^{\dagger}|
\end{eqnarray}
The outer bound is \eqref{gapouterbound3}. Therefore, the gap is
\begin{eqnarray}
&&R_{\text{up}}-R_{\text{low}} \notag\\
&<&
\frac{2}{3}\log(1+\rho)+\frac{1}{3}\log|\mathbf{I}+\rho^{\alpha-1}\underline{\mathbf{H}}_{12}^{\dag}\underline{\mathbf{H}}_{12}|-\frac{1}{3}\log|\mathbf{I}+\rho^{\alpha}\underline{\mathbf{H}}_{12}^{\dag}\underline{\mathbf{H}}_{12}|-\frac{2}{3}\log(1+\rho^{1-\alpha})+\log|\mathbf{I}+\mathbf{\underline{H}}_{12}\mathbf{\underline{H}}_{12}^{\dagger}|\notag\\
&<&\frac{2}{3}\log(1+\rho)+\frac{1}{3}\log|\mathbf{I}+\underline{\mathbf{H}}_{12}^{\dag}\underline{\mathbf{H}}_{12}|-\frac{1}{3}\log(1+2\rho^{\alpha}+\rho^{2\alpha}|\bar{c}|^2)-\frac{2}{3}\log(1+\rho^{1-\alpha})+\log|\mathbf{I}+\mathbf{\underline{H}}_{12}\mathbf{\underline{H}}_{12}^{\dagger}|\notag\\
&=&\frac{1}{3}\log\frac{(1+\rho)^2}{(1+2\rho^{\alpha}+\rho^{2\alpha}|\bar{c}|^2)(1+\rho^{1-\alpha})^2}+\frac{4}{3}\log|\mathbf{I}+\underline{\mathbf{H}}_{12}^{\dag}\underline{\mathbf{H}}_{12}|\\
&<&\frac{8}{3}-\frac{1}{3}\log(|\bar{c}|^2)\label{gap2}
\end{eqnarray}

Consider all cases, the gap is the maximum one, i.e., $\max\{3,
\frac{8}{3}-\frac{1}{3}\log(|\bar{c}|^2)\}$ bits/channel use.
\end{proof}

\end{document}